\documentclass[10pt,conference,compsocconf,letterpaper]{IEEEtran}

\usepackage{cite}
\usepackage{algpseudocode}

\usepackage{graphicx}

\usepackage{caption}

\usepackage{amsthm}
\usepackage{amsmath}
\usepackage{amssymb}
\usepackage{algorithm}
\usepackage{subfig}
\usepackage{blindtext}
\newcommand{\set}[1]{\left\{#1\right\}}

\newcommand{\parentheses}[1]{ \left({#1}\right)}

\newcommand{\cdotpa}[1]{\cdot\parentheses{#1}}

\newcommand{\rap}{Randomized Admission Policy}
\newcommand{\RAP}{RAP}
\newcommand{\pss}{\rap}
\newcommand{\PSS}{\RAP}
\title{\pss{} for Efficient Top-k and Frequency Estimation}

\author{\IEEEauthorblockN{Ran Ben Basat}
    \IEEEauthorblockA{Computer Science\\
        Technion\\
        sran@cs.technion.ac.il
    }
    \and
    \IEEEauthorblockN{Gil Einziger}
	\IEEEauthorblockA{Electrical Engineering\\
		Politechnico di Torino\\
		gilga@polito.it
	}
	\and
	\IEEEauthorblockN{Roy Friedman}
    \IEEEauthorblockA{Computer Science\\
        Technion\\
        roy@cs.technion.ac.il
    }
    \and
	\IEEEauthorblockN{Yaron Kassner}
    \IEEEauthorblockA{Computer Science\\
        Technion\\
        kassnery@cs.technion.ac.il
    }
}

\date{}

\newtheorem{theorem}{Theorem}[section]

\newtheorem*{observation*}{Observation}
\newenvironment{definition}[1][Definition]{\begin{trivlist}
		\item[\hskip \labelsep {\bfseries #1}]}{\end{trivlist}}

\begin{document}
%\newcommand*{\TENPAGES}{}
%\newcommand*{\NINEPAGES}{}
%\ifdefined \NINEPAGES
%\newcommand*{\TENPAGES}{}
%\fi
\newcommand*{\EXTENDED}{}
\maketitle

\begin{abstract}
	Network management protocols often require timely and meaningful insight about per flow network traffic.
	%Attaining this functionality in networking devices is a difficult challenge since commodity SRAM memories, which operate at line speed, are too small to monitor all flows.
	%Previous works suggested tracking only the largest and most significant flows, but they perform poorly on heavy-tailed workloads and are not hardware friendly.
	%Further, they rely on complex data structures and dynamic memory allocations, and are therefore not hardware friendly.
	%
	This paper introduces~\emph{\pss{} (\PSS{})} -- a novel algorithm for the \emph{frequency} and \emph{top-k} estimation problems, which are fundamental in network monitoring.
	%PSS requires significantly less space to obtain a given accuracy when compared to the alternatives.
	We demonstrate space reductions compared to the alternatives by a factor of up to 32 on real packet traces and up to 128 on heavy-tailed workloads.
	For top-$k$ identification, \PSS{} exhibits memory savings by a factor of between 4 and 64 depending on the workloads' skewness.
	These empirical results are backed by formal analysis, indicating the asymptotic space improvement of our probabilistic admission approach.
	Additionally, we present \emph{$d$-Way \PSS{}}, a hardware friendly variant of \PSS{} that empirically maintains its space and accuracy benefits.
\end{abstract}

\newcommand{\matrixCellWidth}{5.8cm}

%\begin{multicols}{2}
\section{Introduction}
\label{sec:intro}

\subsection{Background}
Network management and traffic engineering protocols rely on flow counters based network monitoring.
Examples include effective routing, load balancing, QoS enforcement, network caching, anomaly detection and intrusion detection~\cite{ApproximateFairness,IntrusionDetection,TrafficEngeneering,IntrusionDetection2,LoadBalancing,TinyLFU}.
Typically, monitoring utilities track millions of flows~\cite{CounterArray1,CounterArray2}, and the counter of a monitored flow is updated on the arrival of each of its packets.
Often, the most frequently appearing flows, known as \emph{heavy hitters}, are also the most interesting, since their impact on the above is the most crucial.
%Flow counters enable network owners to perform various networking tasks such as load balancing, routing, fairness, network caching and intrusion detection~\cite{ApproximateFairness,IntrusionDetection,IntrusionDetection2,LoadBalancing,TinyLFU}.

Maintaining such counters is a challenging task with today's storage technology.
The difficulty arises as DRAM is too slow to keep up with line rates, while the faster SRAM is expensive and thus too small for keeping an exact counter for each flow.
These limitation were tackled using various approaches.
%Various approaches for tackling these limitations have been proposed.

\emph{Estimators} reduce the size of counters using probabilistic techniques~\cite{ICE-Buckets,CEDAR,DISCO}.
This enables maintaining one counter per flow in SRAM at the cost of reduced accuracy.
The downside of estimators is that they require an explicit flow to counter mapping for every flow.
This mapping often becomes the dominant factor in memory consumption~\cite{CounterBraids}.

%Estimators are similar in essence to SRAM/DRAM architectures as they too only store a small counter for every flow in SRAM. However, instead of periodically flushing these counters to DRAM, a probabilistic technique is used to allow short counters to represent large numbers at the cost of precision~\cite{ICE-Buckets,CEDAR,DISCO}. These methods introduce error to flow monitoring despite the fact that all flows are monitored. The downside of these methods is that they still require a flow to counter mapping for every flow. Thus even as the counter size is reduced to 8-12 bits, the flow to counter mapping becomes the dominant part of memory consumption. As line speed continue to increase, more and more flows need to be tracked simultaneously making this mapping even more difficult to accommodate in SRAM~\cite{Counterbraids}.

The \emph{shared counters} approach, also known as \emph{sketches}, solves the mapping problem using hashing algorithms that implicitly assign flows to counters.
Well known examples include \emph{Multi Stage Filters}~\cite{CUSketch} and \emph{Count Min Sketch}~\cite{CountMinSketch}.
Yet, to reduce the impact of hash collisions on counters' reading accuracy, these methods must allocate considerably more space and more counters than predicted by lower bounds.

%The inability of maintaining a flow to counter association for all the flows has motivated
%\emph{Shared counters} approaches. Typically in these methods counting sketches and similar multi hash algorithms are used in order to implicitly assign flows to counters.
%Known examples include \emph{multi stage filters}~\cite{CUSketch} and \emph{count min sketch}~\cite{CountMinSketch}.
%These approaches indeed provides an online estimation solution, unfortunately Counting sketches are relatively (and asymptotically) space inefficient. Space efficient variants such as Counter Braids~\cite{Counterbraids} and Randomized Counter Sharing~\cite{RandomizedCounterSharing} solve this problem but can only be decoded offline.

Databases and data analytics face similar problems, known in these domains as \emph{frequency estimation} and \emph{top-$k$ identification}, i.e., identifying who are the $k$ most frequent flows.
These domains typically favor \emph{counter based} solutions over sketches since the former are considered superior to sketches, both asymptotically and in practice~\cite{SpaceSavingIsTheBest, SpaceSavingIsTheBest2010}.
Counter based algorithms maintain a fixed size set of counters and aspire to allocate these counters only to the more frequent flows.
These include Lossy Counting~\cite{LC}, Frequent~\cite{BatchDecrement} and Space Saving~\cite{SpaceSavings}.
The latter is also considered state of the art~\cite{SpaceSavingIsTheBest, SpaceSavingIsTheBest2010, SpaceSavingIsTheBest2011}.
Alas, these algorithms cannot be easily ported into networking devices as they utilize complex data structures and dynamic memory allocation.

%\subsection{Motivation}
%Another significant shortcoming of counter based solutions is that they allocate a counter to the flow of each arriving packet, regardless of how frequent this flow is, even if it means evacuating a counter associated with a more frequent flow.
Another significant shortcoming of counter based solutions is that they update the state of allocated counters on the arrival of each packet belonging to a unmonitored flow, regardless of how frequent this flow is.
Doing so hurts their space to accuracy tradeoff to the point that they become ineffective on heavy-tailed workloads, which are common in network switches and routers.

%Top-k and frequency estimation algorithms are often interested only in a small subset of high frequency flows. Previous algorithms change state prior to every packet arrival, even if its flow is %unmonitored. In networking traces, this approach is often wasteful as a large portion of the unmonitored traffic belongs to tail flows that should not be monitored.

%Existing counter algorithms change their state prior to every packet arrival even if that flow is unmonitored.. Specifically, most unmonitored tail flows would are ignored while frequent flows are eventually admitted and monitored.

% can be ignored and accuracy can be augmented. Intuitively, unmonitored tail items appear only a few times and are unlikely to pass the probabilistic admission, while frequent unmonitored items appear many items and are thus likely to (eventually) be admitted.

%that is when a packet of an unmonitored flow arrives our algorithm only changes its state with a certain (small) probability. That way we are able to ignore most tail flows and are still likely to (eventually) admit all the high frequency top-k flows.

%Shared counters are also utilized in various architectures~\cite{Ciuffoletti06architectureof,paxson1998architecture,BetterNetflow}.
%Such architectures are used to collect and analyze statistics from many networking devices~\cite{HadoopArchitecture}.

\subsection{Contributions}
In this work, we promote the concept of using a randomized admission policy for allocating counters to non-monitored flows, and show that it can significantly improve accuracy.
Intuitively, such a policy ignores most of the tail flows and is still able to eventually admit the high frequency flows.

Specifically, this idea is realized in a novel counter based algorithm called \emph{\pss{}} (\PSS{}) as well as a hardware friendly variant called \emph{$d$-Way associative \PSS{}} ($d$W-\PSS{}).
\PSS{} is simpler to analyze, while $d$W-\PSS{} maps well into limited associativity cache designs and empirically maintains most of the benefits of \PSS{}.
We extensively evaluate \PSS{} and $d$W-\PSS{} over two real packet traces~\cite{CAIDA,UCLA}, a YouTube access trace~\cite{youtube} and synthetic Zipf distributions.

For the frequency estimation problem, \PSS{} and $d$W-\PSS{} achieve the same \emph{mean square error} (MSE) as the leading alternatives while using a fraction of the required memory.
%The measured reduction of MSE is up to 100 times in synthetic distributions and 10 times in packet traces.
For top-$k$ identification, \PSS{} and $d$W-\PSS{} exhibit significantly higher recall and precision, even when allocated with half the space given to the alternative methods.
In particular, when the distribution is only mildly skewed (or heavy-tailed), \PSS{} and $d$W-\PSS{} are the only techniques that successfully identify a high percentage of the top-$k$ flows.

\section{Related Work}
\label{ref:related}
The frequent items and top-$k$ identification problems appear in slight variations across multiple domains.
Algorithms for these problems are often categorized as either \emph{counter based} or \emph{sketch based}.
In addition, the specific challenges of network monitoring have spawned solutions that are especially tailored for the memory limitations in the networking case.

\subsection{Counter based algorithms}
Counter based algorithms are usually designed for software implementations and maintain a table of monitored items.
The differences between these algorithms lie in the question of admission and eviction of entries to and from the table.
From a networking perspective, counter based algorithms maintain an explicit flow to counter mapping for monitored items.
For a stream of $N$ events and an accuracy parameter $\varepsilon$, the goal is to approximate a given flow's frequency to within an additive error of $N\cdot \varepsilon$.
For this task, $\Omega(\frac{1}{\varepsilon})$ counters are required~\cite{SpaceSavings}, and this is achieved by some of the algorithms below.
%The optimal number of table entries is $\frac{1}{\varepsilon}$, and in that case the estimation is bounded by $N\cdot \varepsilon$ where $N$ is the number of events, and $\varepsilon$ is the estimation parameter.

\emph{Lossy Counting}~\cite{LC} increments an arriving item's counter on every arrival.
If the counter is not in the table, it is admitted with a counter value of 1.
Lossy Counting keeps the table size bounded by periodically decrementing table counters and evicting items whose counter reaches 0.
Unfortunately, Lossy Counting requires a maximal number of $\frac{1}{\varepsilon}\cdot\log(N)$ table entries.
\emph{Probabilistic Lossy Counting}~\cite{PLC} requires fewer table entries on average but only provides a probabilistic guarantee.

In \emph{Frequent (FR)}~\cite{frequent4,BatchDecrement}, whenever an item arrives and the table already contains $\frac{1}{\varepsilon}$ entries, the item is not admitted.
Instead, FR decrements every entry in the table, evicting entries whose counter reached 0.
The main benefit of FR is that it requires the optimal number of $O(\frac{1}{\varepsilon})$ table entries.

\emph{Space Saving (SS)~\cite{SpaceSavings}} requires the same number of entries as FR, but maintains additional information to improve accuracy.
Space Saving admits \emph{any} arriving item at the expense of evicting the minimum-frequency item.
Space Saving is considered to be state of the art~\cite{SpaceSavingIsTheBest, SpaceSavingIsTheBest2010, SpaceSavingIsTheBest2011}.

\subsection{Sketch based algorithms}
\emph{Sketches}, such as \emph{Multi Stage Filters}~\cite{MultiStageFilters}, \emph{Count Sketch}~\cite{CountSketch} and \emph{Count Min Sketch }~\cite{CountMinSketch}, are very common in networking domains as they are simple to implement in hardware and have low implementation overheads.
The most popular example, Count Min Sketch, provides the following guarantee --- given an item $x$, with probability of at least $1-\delta$, the estimation error of $x$ is at most $N\cdot \epsilon$.

Count Min Sketch does not require storing flow identifiers or maintaining a flow to counter association.
Instead, it maintains an array of $\ln(\frac{1}{\delta})$ rows, each with $\frac{e}{\varepsilon}$ counters.
%When an item arrives, we calculate a hash function for each row and increment the corresponding counter in each row.
When an item arrives, a hash function is calculated for each row and its corresponding counter is incremented.
To estimate the frequency of an item, the corresponding counters are read and the minimum counter value is returned as the estimation.

Asymptotically, Count Min Sketch requires a suboptimal number of counters.
However, it does not store flow ids and has only minor overheads for hardware implementation.
Despite being suboptimal, sketches still require a sub-linear number of counters, can completely reside in SRAM, and provide online frequency estimation.

On the contrary, \emph{Counter Braids~\cite{CounterBraids}} and \emph{Counter Tree}~\cite{counterTree} use an hierarchical sketch where overflowing counters are hashed to a higher level sketch.
They are able to encode items just like Count Min Sketch would, but the decoding process is complex and can only be performed offline, estimating all flow values together.

In \emph{Randomized Counter Sharing~\cite{RandomizedCounterSharing}}, every time an item is added, a random hash function is used and the corresponding counter is incremented.
The flow identifier is recorded, but without an explicit mapping to frequency.
When a measurement ends, we estimate the flow's frequency by summing all of the corresponding counters or by performing a maximum-likelihood estimation.
Both of these estimations are quite slow and cannot be performed online.
%These estimates are too slow for online queries.
%The estimation process is too slow

In summary, sketches are space suboptimal and only solve the frequency estimation problem.
Further, they only support point queries and their answers are only correct within a certain probability.
Despite these limitations, sketches are used for many networking applications~\cite{SpectralBloom,CountMinSketch,PLC,MultiStageFilters,ML-CBF,HeavyHitters,TinyTable}.

\subsection{Network monitoring architectures}
In \emph{hybrid SRAM/DRAM architectures}~\cite{CounterArray2,CounterArray1}, the LSB bits of counters are stored in SRAM and the MSB in DRAM. This way, the space allocated for each flow in SRAM is small. However, the SRAM counters have to periodically be synchronized with the DRAM counters, which increases the contention on the memory bus.
Further, estimating a flow's frequency requires accessing DRAM and therefore cannot be used for online network monitoring.
%Hybrid DRAM/SRAM architectures~\cite{CounterArray2,CounterArray1} store only the least significant bits of each counter in SRAM and the rest of the counter in (slower) DRAM.
%These architectures periodically update DRAM counters.
%Aside form being offline, their main drawback is the traffic between SRAM to DRAM that can become a bottleneck of its own.

Brick~\cite{Brick} uses an efficient encoding in order to reduce the number of bits allocated per counter.
Brick enables storing more counters, under the assumption that the total value of increments is known in advance. Brick is most effective when there are many very small flows.
%TinyTable~\cite{TinyTable} encompass a compact hash table design, with variable length counters and thus reduce the amount of space required. However, TinyTable is still linear with the number of required flows.

Estimators use fixed size small counters in order to represent large numbers. These methods trade precision for space and allow more counters to be contained in SRAM.
This idea was first introduced by \emph{Approximate Counting}~\cite{ApproximateCounting} and was adapted to networking devices~\cite{SAC,DISCO,CEDAR,ICE-Buckets}. The downside of estimators is that they require storing a flow-to-counter mapping for every flow, a requirement that has many overheads.  Sampling techniques are another alternative that trades accuracy for space. Unfortunately, these methods can only monitor large flows that are frequent enough to be sampled~\cite{Sample1,BetterNetflow}.

\section{\pss{} (\PSS{})}
\label{sec:\PSS{}}
\PSS{} maintains a table ($C$) which contains $M\triangleq\frac{1}{\epsilon}$ entries.
The intuition behind \PSS{} is to minimize the error inflicted upon arrival of a non-monitored item $x\notin C$ when the table is full.
That is, we identify inefficiencies in the way previous works behave in this case.
E.g., FR needlessly increases the error of all counters by decrementing all of them.
In contrast, Space Saving always evicts the item with the minimal counter.
This eviction introduces an error, as the monitored element is often more frequent than a randomly arriving item without a counter.
This is especially true for heavy-tailed workloads,
%,in which a large portion of randomly arriving items are insignificant ``tail elements'' and should not be admitted into the table.
where a large fraction of the stream consists of ``tail elements'' that should not be admitted into the table.
In \PSS{}, we take a more conservative approach.
When an item $x\notin C$ arrives, we find the item ($m$) with the minimal counter value ($c_m$).
$x$ is then admitted into $C$ with probability $\frac{1}{c_m+1}$ at the expense of $m$; otherwise, $x$ is simply discarded.
Algorithm~\ref{alg:\PSS{}} provides a pseudo code of the \PSS{}'s \textsc{Add} method.
%, the function Random(), provides a uniform sample between 0 and 1, the sample is smaller than any $x \in (0,1)$ with probability $x$.

%On average, if an item ($x$) replaces $m$, whose counter value is $c_{m}+1$, $x$ has to arrive $c_{m}+1$ times.
In order for an item $x$ to replace the minimal element $m$, it has to arrive $c_{m}+1$ times on average.
Infrequent items are therefore unlikely to be admitted into $C$, and most of them will not affect any of the counters.
Therefore, \PSS{} is considerably more accurate, especially for heavy tailed workloads where a large portion of the items are infrequent.
In contrast, every tail item in Space Saving affects the counters, thereby contributing to the total estimation error.
Our approach is not without risks, as if an arriving item turns out to be frequent, Space Saving admits that item sooner than \PSS{}. 

Given a query for the frequency of element $x$, \PSS{} estimates it as $c_x$ if $x\in C$ and 0 otherwise.

\PSS{} can be implemented with existing data structures and it processes packets at $O(1)$ runtime~\cite{SpaceSavings,HeavyHitters}.
%(when increments are only $+1$).
It stores a single counter per table entry, while Space Saving entries are slightly larger as they store two values.

\begin{algorithm}[h]
	\ifdefined \NINEPAGES
	\small
	\fi
	Initialization: $C\gets\emptyset, \forall i: c_i\gets0$
	\begin{algorithmic}[1]
		\Function{Add}{Item $x$}
		\If {$x\in C$}
		\State  $c_x\gets c_x+1$
		\Else\If {$|C| < M$}
		\State $c_x \gets 1$
		\State $C \gets C \cup \{x\}$
		\Else
		\State $m \gets \text{argmin}_{y\in C} c_y$
		\If {$random()<{\frac{1}{c_m+1}}$}  \Comment w.p $\frac{1}{c_m+1}$.
		\State $C \gets (C \setminus \{m\}) \cup \{x\}$
		\State $c_x \gets c_m + 1$
		\EndIf\EndIf\EndIf
		\EndFunction
	\end{algorithmic}
	\caption{\pss{}}
	\label{alg:\PSS{}}
\end{algorithm}
\normalsize

\subsection{Analysis}
We start our analysis with theoretic bounds for the top-$k$ problem. These show that our probabilistic approach is asymptotically better for i.i.d. streams. We then explore the properties of \PSS{} for the frequency estimation problem.

\subsubsection{Top-$k$ Problem}
%First, we need to formally define the top-$k$ problem:
We say that an algorithm successfully solves top-$k$ if it identifies the $k$ most frequent flows in a stream. Our goal is to bound the number of table entries required for successful identification of top-$k$.

%We analyze the case for i.i.d. Zipf distributions over a domain of size $D$.
Denote $\Gamma_\alpha(D) \triangleq \sum_{i=1}^{D}i^{-\alpha}$. A stream will be called an i.i.d. Zipf stream with skew $\alpha$ over domain $D$ if all of its elements are sampled independently and follow the distribution in which item $i\in\set{1,2,\ldots,D}$ appears with probability $\frac{i^{-\alpha}}{\Gamma_\alpha(D)}$. Such a stream is denoted $\mathcal Z^{D}_{\alpha}$.
%These streams are defined with a parameter $\alpha$; for $i\in\{1,2,\ldots,D\}$ the probability of each element
%The notation $Z^{D}_{\alpha}$ refers to a Zipf distribution over $D$ distinct items, with Zipf parameter $\alpha$. This distribution is a characteristic of Internet traces, including networking traces.

\begin{theorem}\label{thm:pss-heavy-tailed}
	Let $\alpha < 1$. For any fixed $k$, Space Saving requires $O( D^{1-\alpha})$ counters to solve top-$k$ on $\mathcal Z^{D}_{\alpha}$, while a randomized admission policy requires only to~$O(D^{\frac{1-\alpha}{1+\alpha}})$.
\end{theorem}
\ifdefined\TENPAGES
Next, we analyze the performance of a probabilistic admission filter for streams with higher skew.
The proof of Theorem~\ref{thm:pss-heavy-tailed}, and that of the following theorem, appear in the full version of the paper~\cite{full-version}.
\fi
\ifdefined\EXTENDED
%The proof of Theorem~\ref{thm:pss-heavy-tailed} appears in Appendix~\ref{sec:topk\PSS{}}.
Next, we analyze the performance of a probabilistic admission filter for streams with higher skew.
The proofs of both theorems appear in Appendix~\ref{sec:topk\PSS{}}.
%, where we reduce space $\sqrt{\log D}$ improvement is achieved
%the proof of the following theorem also appears in Appendix~\ref{sec:topk\PSS{}}.
\fi
\begin{theorem}\label{thm:pss-heavy-tailed2}
	For $\mathcal Z^{D}_{1}$ and any fixed $k$,  Space Saving requires $O(\log D)$ counters to solve top-$k$, while a randomized admission policy  reduces the required number to $O(\sqrt{\log D})$.
	
\end{theorem}
%Unfortunately, these guarantees are weaker than those of deterministic algorithms.
\vspace{-0.2cm}
\subsubsection{Frequency Estimation Problem}
We now present a brief mathematical analysis of \PSS{}, including deterministic and probabilistic upper bounds for the estimation error.
%We start with a deterministic upper bound on \PSS{}'s estimation error.
\begin{theorem}
	Let $f_x$ be the true frequency of $x$, $\widehat{f_x}$ be \PSS{}'s estimation of $f_x$, and $m$ be $\min_{y\in C} C_y$.
	Then $\widehat{f_x}\le f_x+m$.
\end{theorem}

\begin{proof}
	The proof is by a case analysis.
	First, suppose $x \notin C$ at the time of the query.
	In this case, $\widehat{f_x}=0$ and the claim trivially holds.
	Conversely, assume that $x \in C$ at the time of the query;
	consider the last time $t$ in which $x$ was admitted into $C$.
	At that point, $c^t_x=m_{t^-}+1$ where $c^t_x$ is the value of $x$'s counter at time $t$ and $m_{t^-}$ is the minimum counter in the table just before time $t$.
	Notice that the algorithm can only increase the minimum counter in the table due to a packet arrival, at which point either no counter changes or the minimal counter is incremented.
	Hence, $c^t_x\le m_{t^-}+1$.
	Next, suppose that $x$ has arrived exactly $n$ times between $t$ and the present;
	$n\le f_x-1$ since we know that at time $t$, $x$ arrived once.
	It follows that $c_x=c^t_x+n\le m_{t^-}+1+f_x-1=f_x+m_{t^-} \leq f_x+m$.
\end{proof}

Next, we show that the estimation given by \PSS{} is in expectation smaller than or equal to the true frequency.
\begin{theorem}
	$\mathbb{E}\left[\widehat{f_x}\right]\le f_x$.
\end{theorem}
\begin{proof}
	In this proof, we use the notion of time to describe the events in the stream.
	The first event is at $t=0$, the next at $t=1$ and so on.
	We prove the claim by induction on the time $t$.
	\textit{Base:} At time $0$, $\widehat{f_x}=0$.
	\textit{Step:} If at time $t$ an item different than $x$ has arrived, then $\mathbb{E}\left[\widehat{f^t_x}\right]\le\mathbb{E}\left[\widehat f^{t-1}_x\right]$ ; in case $c_x$ was the smallest counter at time $t-1$, its estimation can only decrease, and otherwise its estimation does not change.
	However, if  $x$ arrived at time $t$, let $\Delta E^t_x$ be the change in $\mathbb{E}\left[\widehat{f_x}\right]$, that is  $\Delta E^t_x =\mathbb{E}\left[\widehat{f^t_x}-\widehat f^{t-1}_x\right].$
	
	There can be two cases:
	if $x\in C$, then $\widehat{f^t_x}-\widehat f^{t-1}_x=1$.
	Otherwise $x\notin C$, hence its estimation is either increased by $c_m+1$ with probability $1\over c_m+1$ or remains the same.
	Thus, in all cases $\widehat{f^t_x}-\widehat f^{t-1}_x$ grows by $1$ in expectation and $\Delta E^t_x=1$.
	Hence, the induction hypothesis holds and $\mathbb{E}\left[\widehat{f_x}\right]\le$~$ f_x$.
\end{proof}

\section{Hardware Friendliness}
\label{sec:hardware}
\PSS{} can be efficiently implemented in software with existing data structures~\cite{SpaceSavings, HeavyHitters}. 
These complex data structures might be difficult to efficiently implement in hardware.

In this section, we present~\emph{d-Way \pss{}} ($d$W-\PSS{}), a hardware friendly variant of \PSS{}.
We describe $d$W-\PSS{} as a cache management policy.
Caches are well understood, making $d$W-\PSS{} implementation as a cache policy easy to design as it does not rely on complex data structures. In addition, caches have a proven capability to operate at line speed.
%Therefore, this presentation of $d$W-\PSS{} enables us to implement it as a cache policy, in a manner that requires a minimal amount of additional design and does not rely on complex data structures.
For self containment, Section~\ref{sec:cache-organization} provides a brief introduction to cache topology.

\subsection{Cache Memory Organization}
\label{sec:cache-organization}
%While in software data structures are often fully associative, the high speed requirement from hardware caches prevents them from being that sophisticated.
In order to meet their high speed requirements, hardware caches are usually not fully associative.
As a rule of thumb, the higher the associativity level -- the slower the cache is since the search process becomes more complex.
Limited associativity means that each item can only be placed in a certain logical place in the cache.
If this place is already full, an existing item must be evicted in order to admit the new one.

These logical locations are called \emph{sets} and in each set there are a certain number of places called \emph{ways}.
We use a hash function ($Set(x)$) to map an item to a certain set number; the item can only be stored in that set.
This makes the lookup process simpler as we only need to search for the item in a specific set, rather than in the entire cache.

The more ways we add to the cache -- the slower the cache works, as there are more places that an item could be found in.
Therefore, to ensure fast performance, the number of ways is kept small, typically $2-32$.
A cache with $d$ different ways is called \emph{d-way set associative}, a cache with only a single set is called \emph{fully associative}.
A cache with a single way is called~\emph{direct mapped}.

Figure~\ref{fig:kway} illustrates the basic topology of a 4-way set associative cache.
In this example, the $Set$ function is used to determine the set for $x$.
The set selected is the one marked with orange (horizontal line) and since
the cache has 4 ways, $x$ could be placed in either of these ways. The cache first checks whether $x$ appears in these ways. If it is not found, a cache policy is used to decide whether to admit $x$ into the cache, at the expense of evicting some other item, or not.

\begin{figure}[h]
	
	\center{
		\includegraphics[scale=0.28]{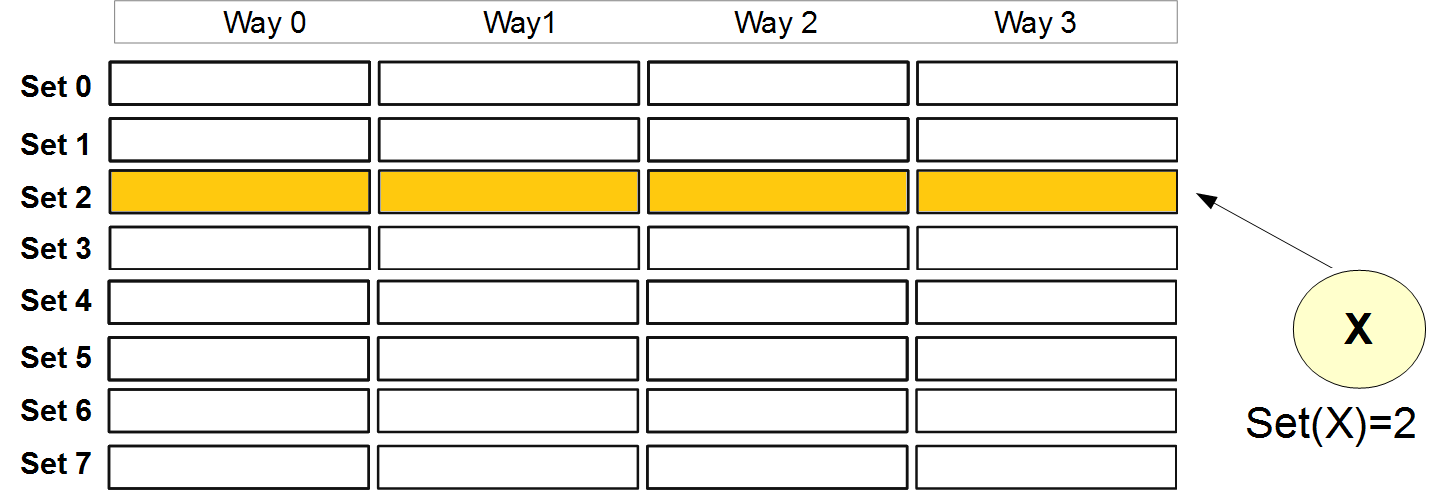}
		\caption{\label{fig:kway} A 4-way set associative cache with 8 lines. When an item ($x$) arrives, Set($x$)=2 is calculated and the item can be stored in any of the ways of set 2.
		}}	
		
	\end{figure}
	
	\subsection{Cache Policy}
	A fundamental cache management question is what to do when an item arrives and its corresponding set is full.
	A \emph{cache policy} is an algorithm that answers these questions.
	Cache policies can sometimes be partitioned into two sub policies: an \emph{admission policy} and an \emph{eviction policy}~\cite{TinyLFU}.
	The former decides whether to admit an item into the cache and the latter decides on the cache \emph{victim}.
	
	\subsection{$d$W-\PSS{} as a Cache Policy}
	Algorithm~\ref{alg:\PSS{}} implements \PSS{} assuming (implicitly) a fully associative memory organization. We now describe $d$W-\PSS{} as a cache policy for a $d$-way cache organization.
	
	\paragraph{Metadata}
	In $d$W-\PSS{}, each entry contains a counter that is used for both frequency/top-$k$ estimation, and for the cache admission and eviction policies.
	
	\paragraph{Metadata Update} In $d$W-\PSS{}, every time a cached item is accessed, including right after the initial admission, its counter is incremented by $1$.
	
	\paragraph{Eviction Policy} When a set is full the cache victim is always the entry with the minimal counter in the set.
	
	\paragraph{Admission Policy} $d$W-\PSS{}'s cache policy does not always admit an item into the cache.
	Instead, it first identifies the set entry with minimal counter as a potential cache victim.
	If that entry's ID is $m$ and its counter value is $c_m$, a new item is admitted with probability $\frac{1}{c_m+1}$.
	The counter of a new item remains with the same value $(c_m)$, and is later incremented by the metadata update.

	\begin{figure}[h]
		
		\center{
			\includegraphics[scale=0.28]{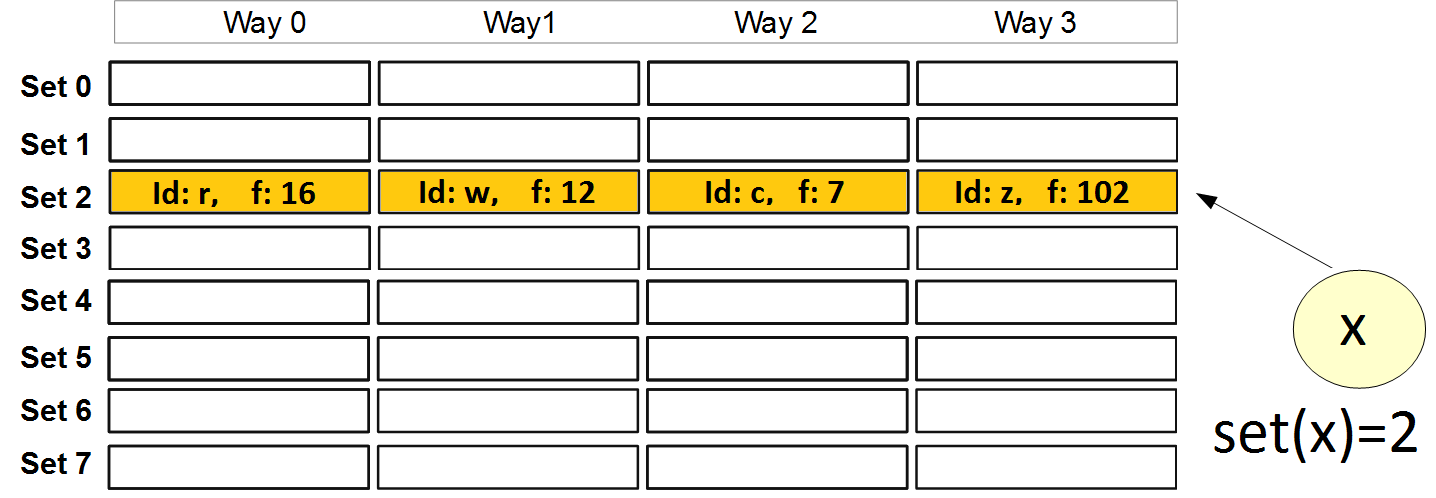}
			\caption{\label{fig:k\PSS{}Cache} A 4-way set associative cache with \PSS{} policy. Item $x$ has Set($x$)=2. Set 2 includes $r$, $w$, $c$ and $z$, while $x$ is not in the cache. The eviction policy selects $c$, because its frequency is the smallest. $x$ will be admitted into the cache with probability $\frac{1}{8}$. If $x$ is admitted, its counter will be incremented to $8$.
			}}	
			
		\end{figure}
		
		An example of  $d$W-\PSS{} is given in Figure~\ref{fig:k\PSS{}Cache}.
		When $x$ arrives, we first look for it in $Set(x)=2$.
		There are 4 items in set 2 and $x$ is not one of them.
		Thus we need to decide on eviction and admission.
		If we choose to admit $x$, we evict the minimal item ($c$) in way 2.
		The frequency of $c$ is $c_c = 7$ and therefore $x$ is only admitted into the cache with probability $\Pr[Admit(x)] = \frac{1}{c_c +1} = \frac{1}{8}.$\\
		If $x$ is admitted into the cache, $c$ is removed from the cache and $c_x$ is set to be $c_c +1 = 8$.
		
		Fortunately, the complexity of implementing $d$W-\PSS{} no longer depends on the number of counters, but only on the associativity level ($d$).
		The larger $d$ is, the more combinatoric logic is used for searching the cache and identifying the minimum.
		As mentioned above, $d$ is typically very small and we can treat the complexity as $O(1)$.
		In Section~\ref{sec:results}, we evaluate $d$W-\PSS{} and show that it is almost as accurate as (the fully associative) \PSS{}, even for relatively small values of $d$.
		\ifdefined\TENPAGES
		We have also experimented with different associativity levels and evaluated their impact. The experiment details and results appear in the full version~\cite{full-version}.
		\fi
		\ifdefined\EXTENDED
		We have also experimented with different associativity levels and evaluated their impact. The experiment details and results appear in Appendix~\ref{apx:assoc}.
		\fi
		
		%As can be observed the admission, eviction and meta-data management are simple and straight forward. Further, since the cache is 4-way set associative, finding the minimal counter is easy and cheap to do. Since usually cache associativity is small it is easy to efficiently implement K-\PSS{} . The logic required to identify the minimum counter only from $K$ different counters is small and simple to implement. In Section~\ref{sec:result} we show that performs very well, even when compared to state of the art (and fully associative) alternatives that require complex data structures and dynamic memory management.

\section{Evaluation}
\label{sec:results}
In this section, we evaluate \PSS{} and $d$W-\PSS{} along with the following previously suggested algorithms -- Frequent~(\emph{FR})~\cite{BatchDecrement} and Space-Saving (\emph{SS})~\cite{SpaceSavings}.
These counter-based algorithms were proven effective for both frequency and top-$k$ estimation.
The latter is considered state of the art~\cite{SpaceSavingIsTheBest, SpaceSavingIsTheBest2010, SpaceSavingIsTheBest2011}.
For frequency estimation, we also compare with sketches such as (\emph{CS})~\cite{CountSketch} and (\emph{CMS})\cite{CountMinSketch}.

%For equal number of counters, frequency sketches are asymptotically worse than counter based algorithms.
%However, they do not maintain a flow to counter association and thus require less memory per counter.

For a fair comparison, we evaluate the performance of CS and CMS using 8 times as many counters as the rest of the (counter based) algorithms. 
To represent their low implementation overhead, they were configured to use 4 lines, which was shown effective in practice~\cite{caffein}.
By giving the sketches more counters, we compensate for their lower overheads, as they do not maintain a flow to counter association and avoid storing flow identifiers.
%This is done to reflect possible additional overheads from maintaining flow to counter association and storing flow identifiers by counter based algorithms.
We consider this a generous comparison, as the flow id and metadata overheads should not take more than $7$ times the counter size.

\subsection{Datasets}
Our evaluation includes the following datasets:
\begin{enumerate}
	\item The CAIDA Anonymized Internet Trace 2015 \cite{CAIDA} , or in short, \emph{CAIDA}. The data is collected from the `equinix-chicago' high-speed monitor and contains 18M elements of mixed UDP, TCP and ICMP packets.
	\item The UCLA Computer Science department packet trace (denoted \emph{UCLA})\cite{UCLA}. This trace contains 32M UDP packets passed through the border router of the CS Department, University of California, Los Angeles.
	\item YouTube Trace from the UMass campus network (referred to as \emph{YouTube})~\cite{youtube}. The trace includes a sequence of 600K accesses to YouTube from within the university.
	\item Zipf streams. Self-generated traces of identical and independently distributed elements sampled from a Zipf distribution with various skew values ($0.6, 0.8, 1.0, 1.2$ and $1.5$).
	Hereafter, the skew $X$ stream is denoted~Zipf$X$.
	%	The stream with skew $X$ is heretofore denoted as \emph{Zipf$X$}.
\end{enumerate}

\subsection{Metrics}
%Our evaluation discusses two problems with the following performance metrics:
Our evaluation considers the following performance metrics:
\begin{enumerate}
	\item \textbf{On-Arrival frequency estimation} \\
	Many networking applications take decisions on a \emph{per-packet} basis.
	For example, if a router identifies excessive traffic originating from a specific source, the router may suspend further routing of its packets to prevent denial of service attacks.
	We refer to this as the \emph{On-Arrival} model, where upon arrival of each packet, the algorithms are required to estimate its flow frequency.
	Formally, a stream $S=s_1,s_2,\ldots$ is revealed one element at a time; consequent to $s_t$ arrival, an algorithm \emph{Alg} is required to provide an estimate $\widehat{f_{s_t}}$ for the number elements in the stream with the same id. We then measure the \emph{Mean Square Error (MSE)} of the algorithm, i.e.,
	$MSE(Alg) \triangleq \frac{1}{N}\sum_{t=1}^{N}(\widehat{f_{s_t}} - f_{s_t})^2.$
	
	\item \textbf{Top-$k$ Identification} \\
	The ability to identify the most frequent flows is also important to many applications.
	We define the \emph{Top-$k$} identification problem as follows:
	Given a stream $S=s_1,s_2,\ldots$ and two query parameters $m$ and $k$, the algorithm is required to output a set of $m$ elements containing as many of the $k$ most frequent stream elements as possible.
	We denote the $k$-highest element frequency by $F_k$.
	For a set of candidates $C$, we measure its quality using the standard recall and precision metrics:
	\ifdefined \EXTENDED
	\begin{align*}
	\text{Precision}(C) &\triangleq\ \  |e\in C : f_e \ge F_k| / |C|\\
	\text{Recall}(C)    &\triangleq\ \  |e\in C : f_e \ge F_k| / k.
	\end{align*}
	\else
	$
	\text{Precision}(C) \triangleq\ \  |e\in C : f_e \ge F_k| / |C|\\
	\text{Recall}(C) \ \ \    \triangleq\ \  |e\in C : f_e \ge F_k| / k
	$.
	\fi
	%We also define the \textbf{\emph{$k$-Accuracy}} measure, which evaluates how many of elements from the top-$k$ the algorithm is able to identify with a set {\bf of size $k$} (note that this quantity also equals the precision \emph{and} recall of the returned set).
	%we consider the problem of approximating the frequency of an arriving element, a
\end{enumerate}
\begin{figure*}[tp!]
	\begin{tabular}{ccc}
		\subfloat[CAIDA]{\label{fig:CaidaMSE}\includegraphics[width = \matrixCellWidth]
			{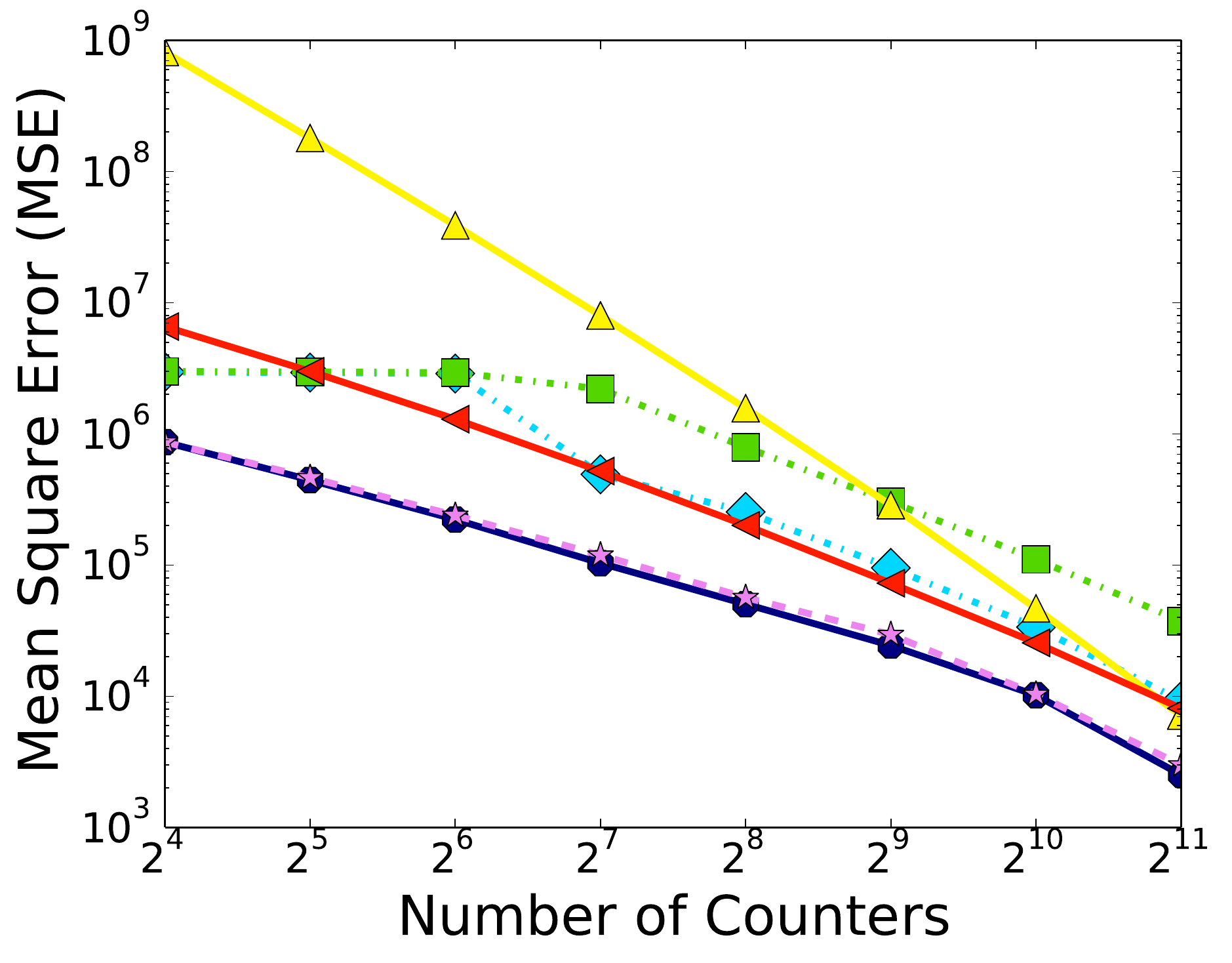}} &
		\subfloat[UCLA]{\label{fig:UCLAMSE}\includegraphics[width = \matrixCellWidth]
			{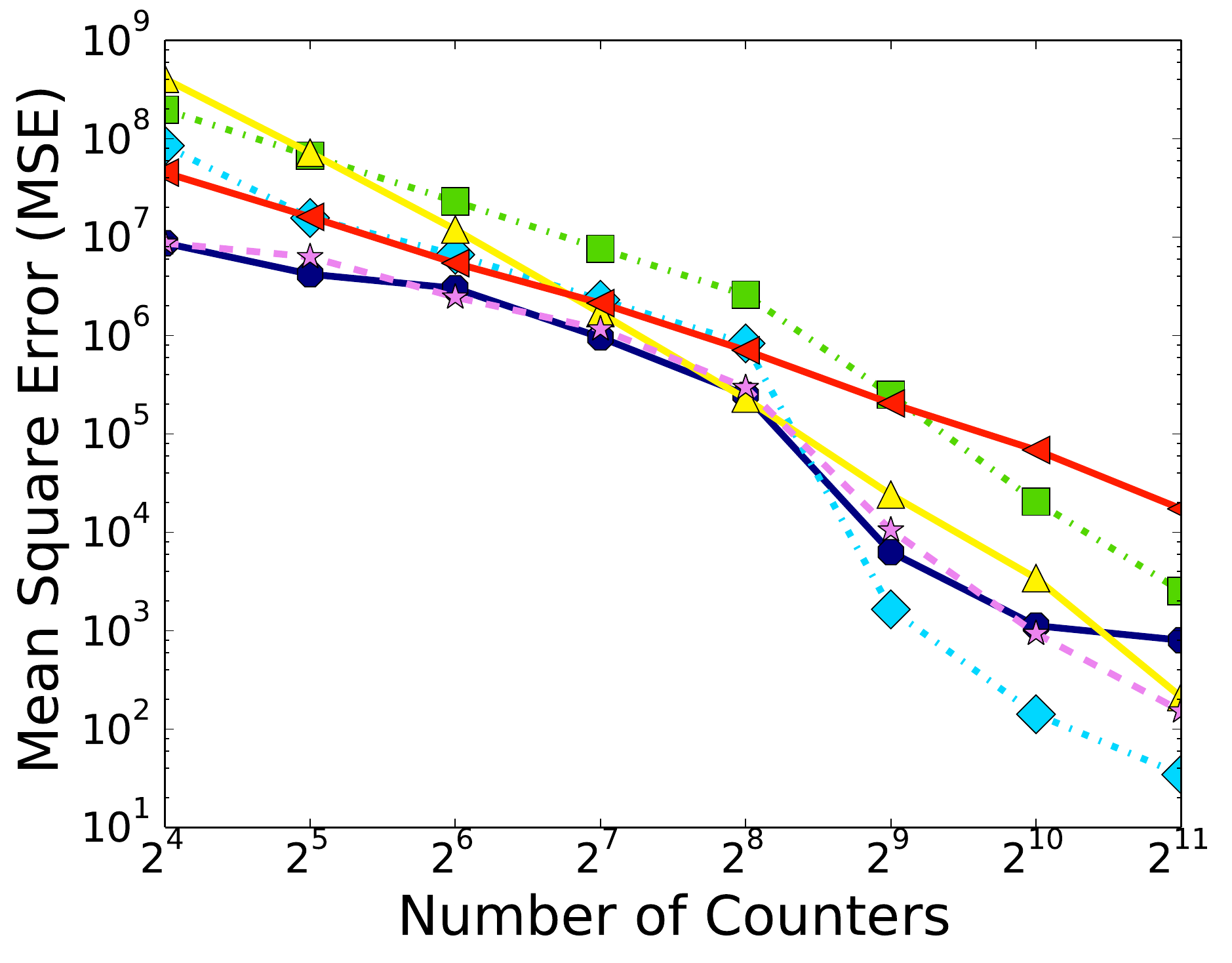}} &
		\subfloat[YouTube]{\label{fig:YouTubeMSE}\includegraphics[width = \matrixCellWidth]
			{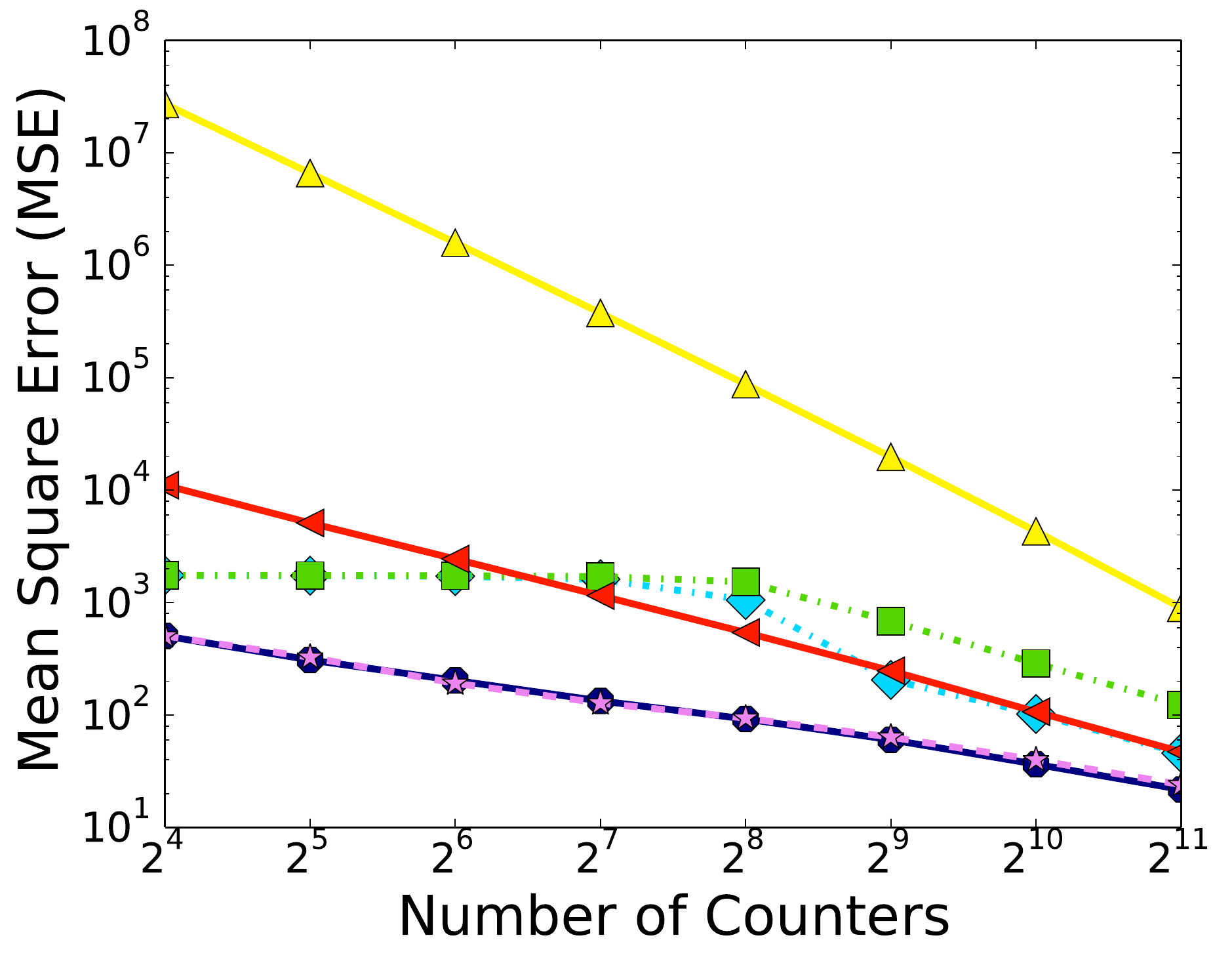}}\\
		\subfloat[Zipf0.6]{\label{fig:Zipf0.6MSE}\includegraphics[width = \matrixCellWidth]
			{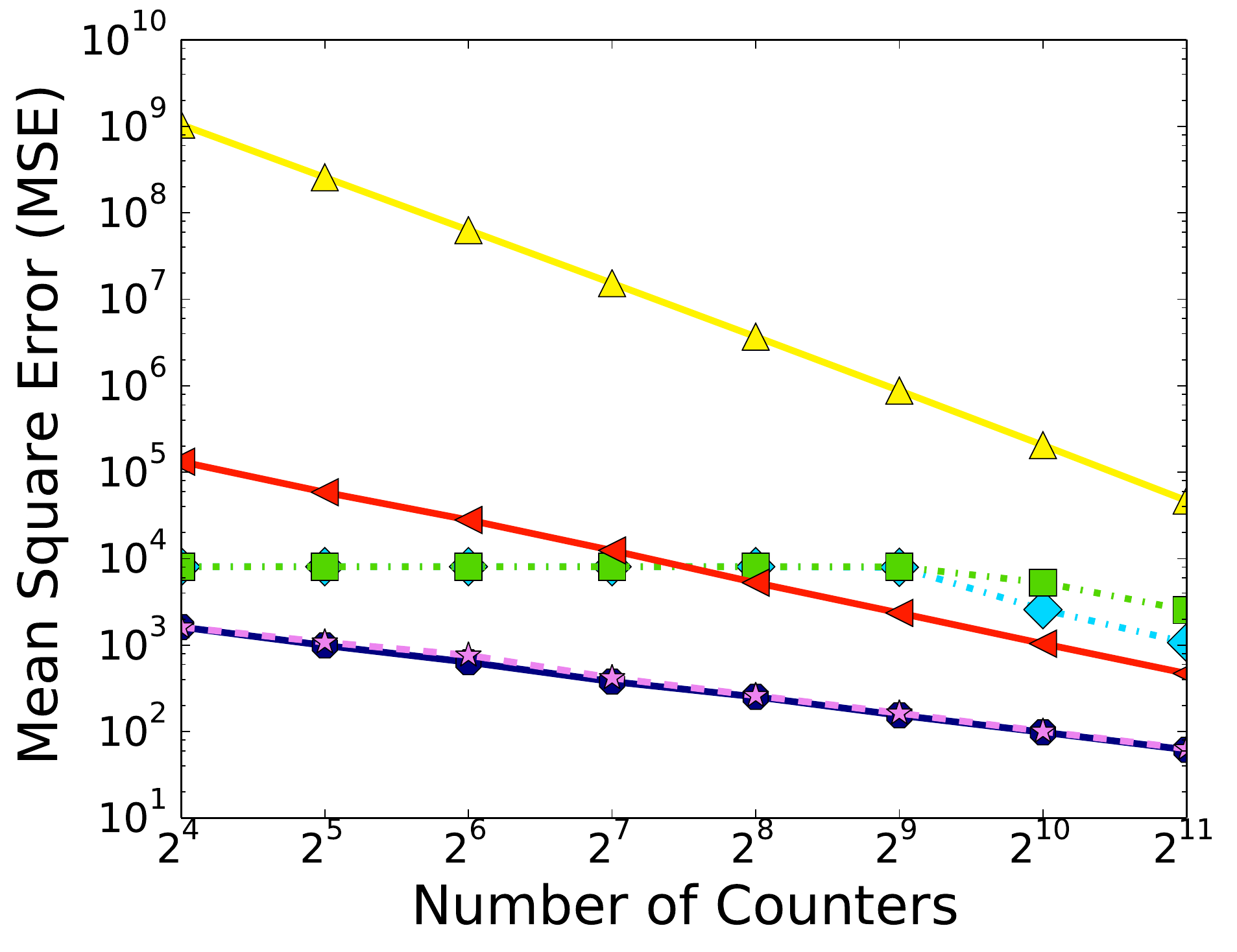}} &
		\hspace*{0.15in}
		\subfloat[Legend]{\label{fig:Zipf0.8MSE}\includegraphics[width = 4.9cm, height=4.3cm]
			{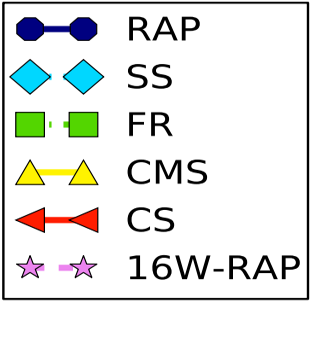}}&		
		\subfloat[Zipf0.8]{\includegraphics[width = \matrixCellWidth]
			{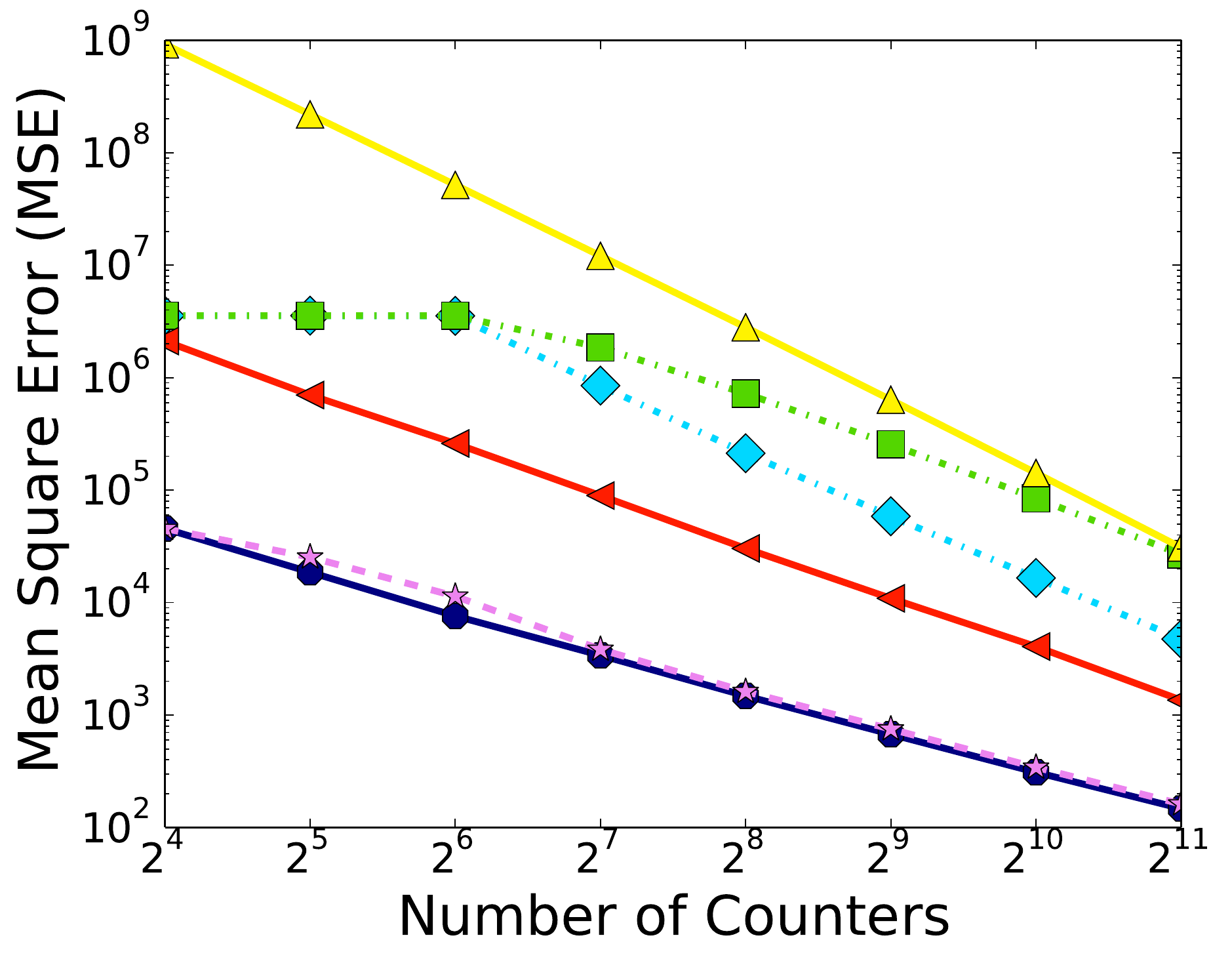}} \\
		\subfloat[Zipf1.0]{\label{fig:Zipf1.0MSE}\includegraphics[width = \matrixCellWidth]
			{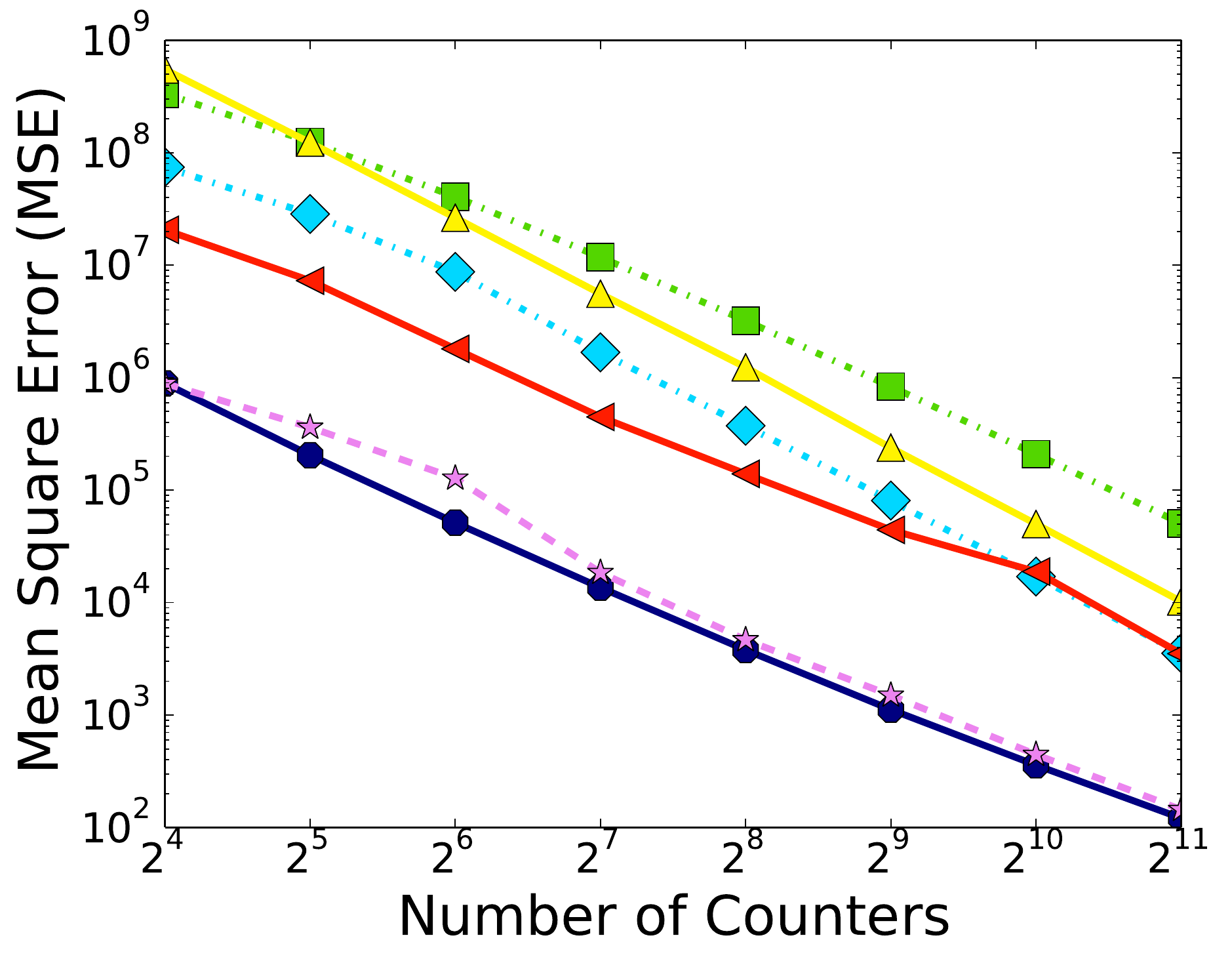}} &
		\subfloat[Zipf1.2]{\label{fig:Zipf1.2MSE}\includegraphics[width = \matrixCellWidth]
			{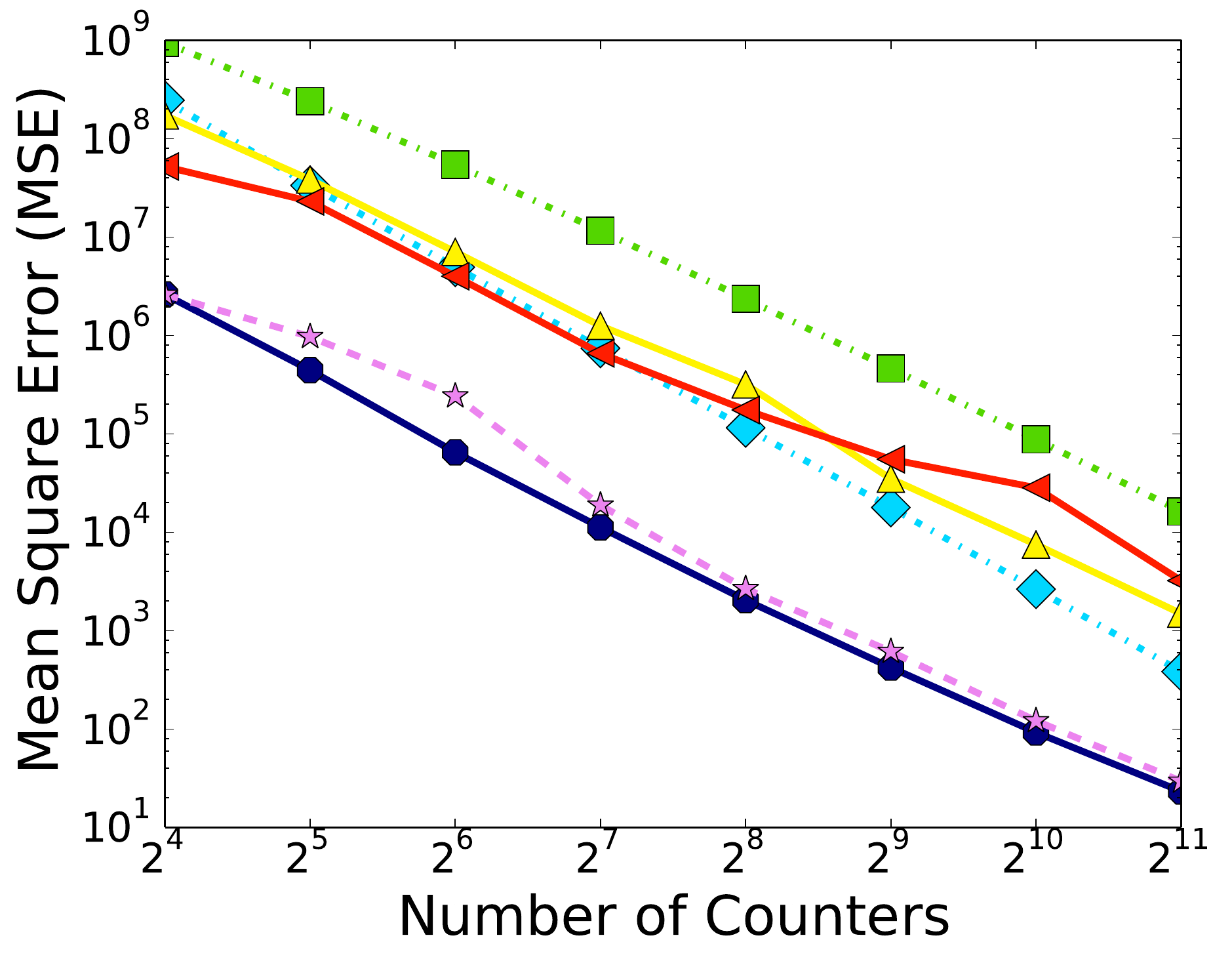}} &
		\subfloat[Zipf1.5]{\label{fig:Zipf1.5MSE}\includegraphics[width = \matrixCellWidth]
			{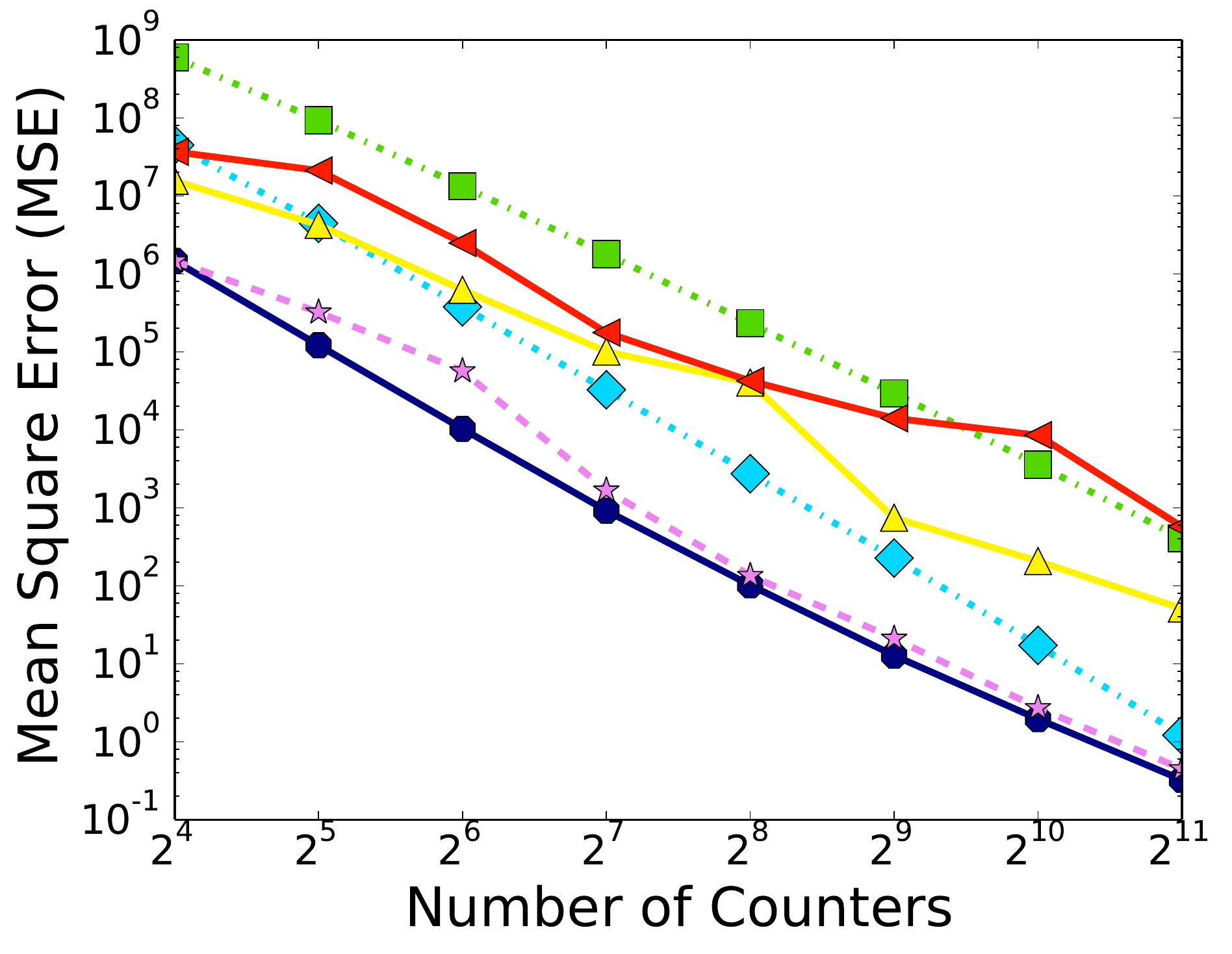}}
		
	\end{tabular}
	\caption{\label{fig:MSE}On Arrival -- Mean square error vs. number of counters. \\Note that CMS and CS have 8 times as many counters!}
\end{figure*}
\begin{figure*}[t!]
	\begin{tabular}{ccc}
		\subfloat[CAIDA]{\label{top32Caida}\includegraphics[width = \matrixCellWidth]{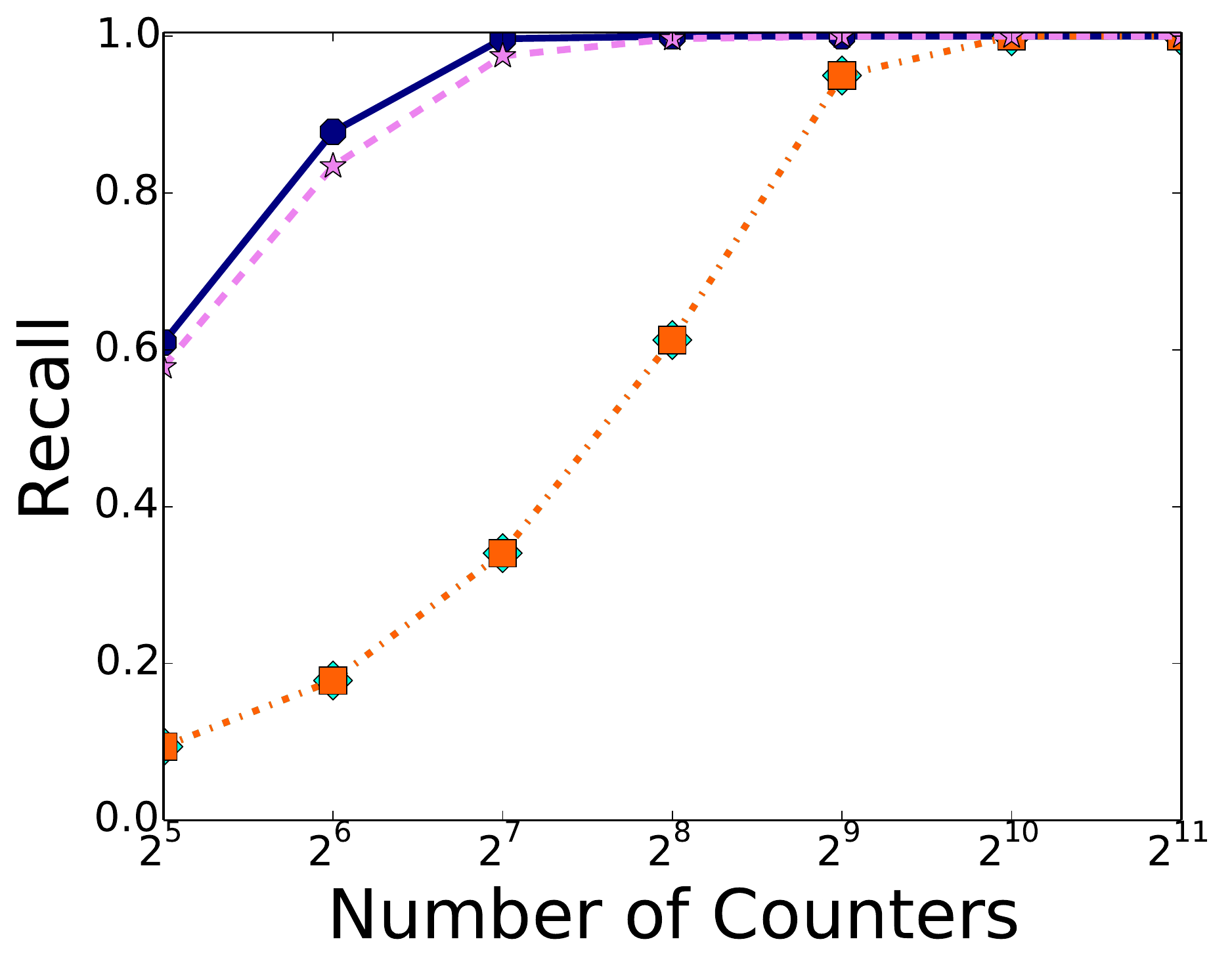}} &
		\subfloat[UCLA]{\label{top32UCLA}\includegraphics[width = \matrixCellWidth]
			{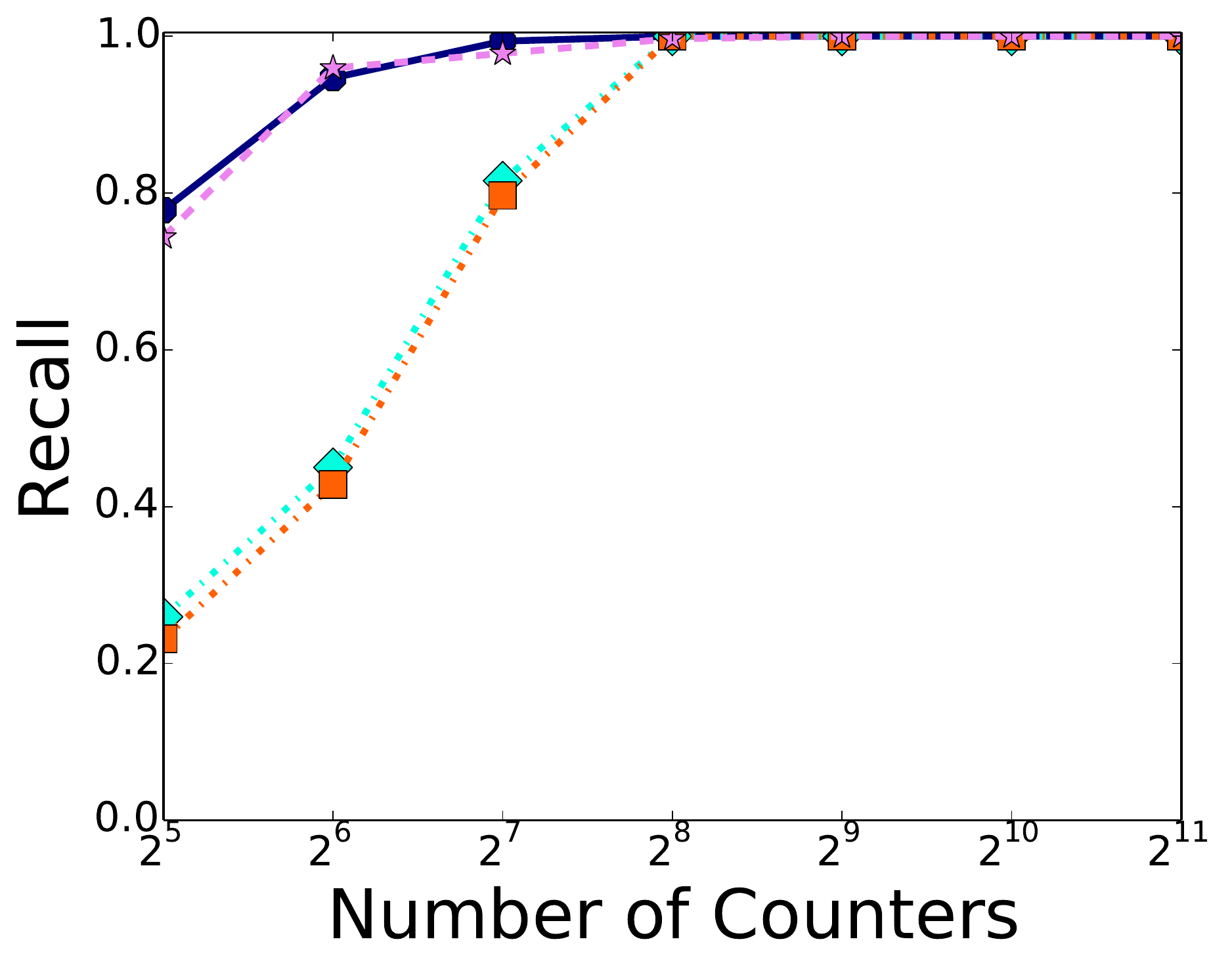}} &
		\subfloat[YouTube]{\label{top32YouTube}\includegraphics[width = \matrixCellWidth]
			{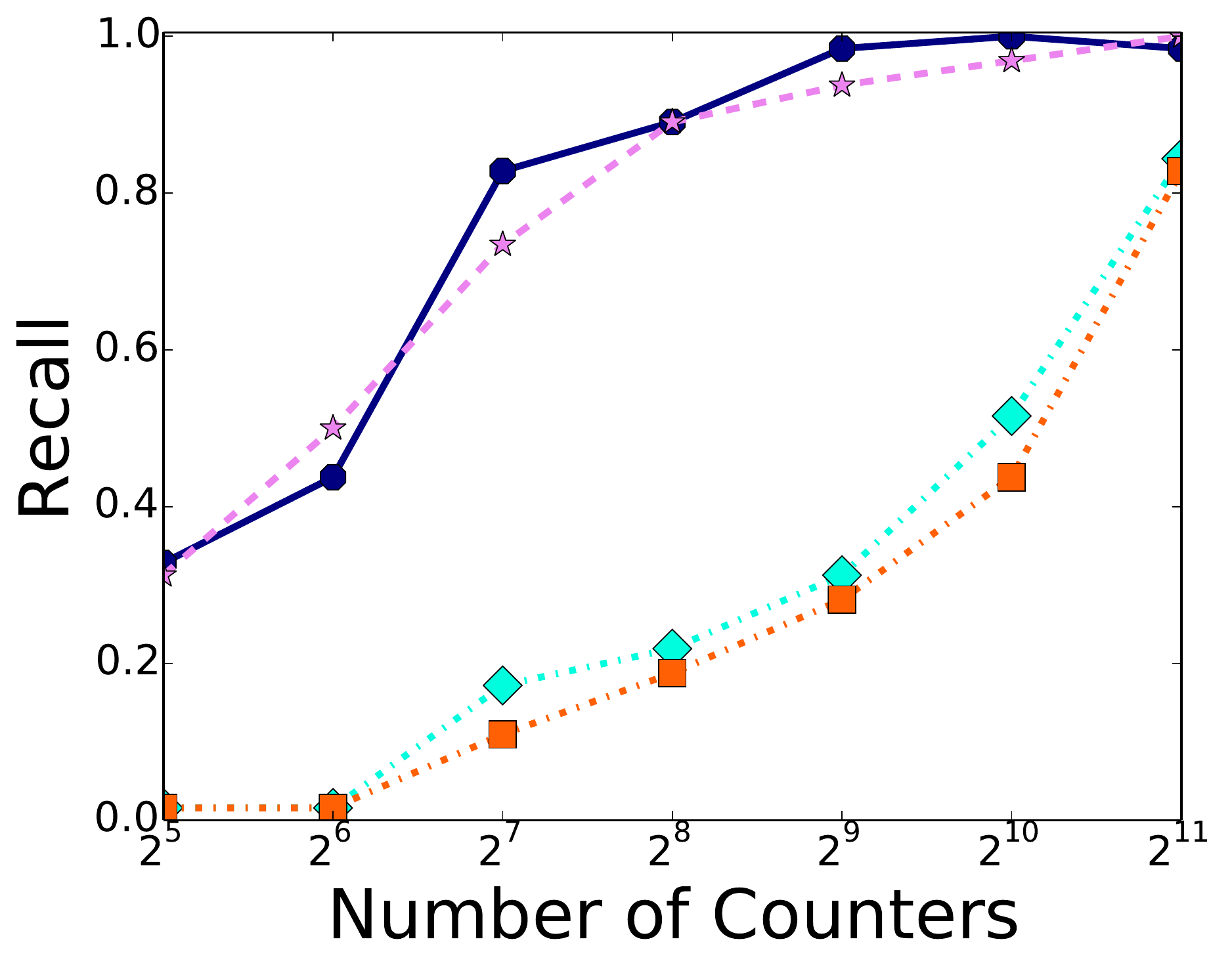}}\\
		\subfloat[Zipf0.6]{\label{top32Zipf06}\includegraphics[width = \matrixCellWidth]
			{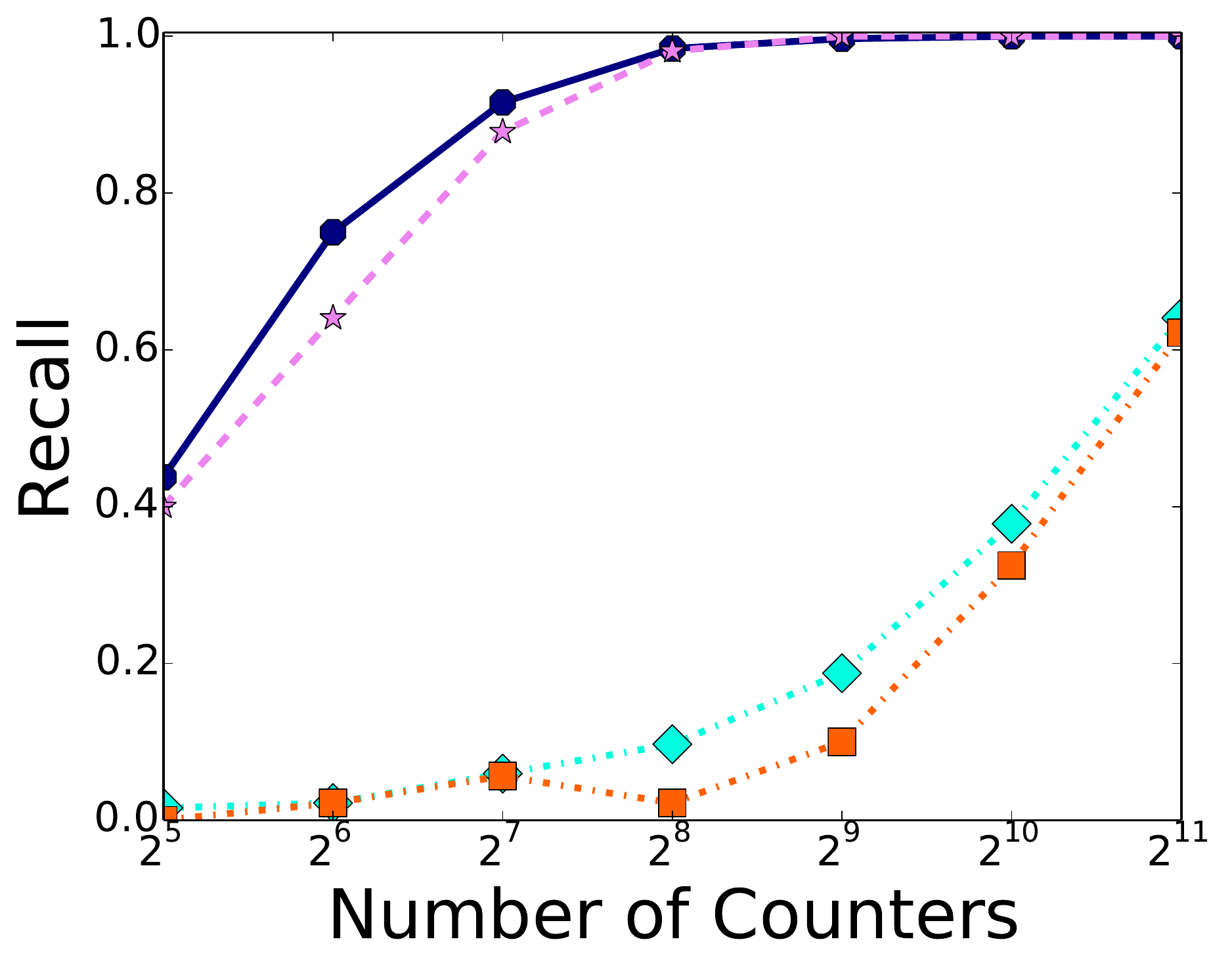}} &
		\hspace*{0.15in}
		\subfloat[Legend]{\includegraphics[width=5cm,height=4.45cm]
			{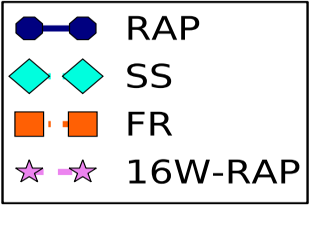}}&		
		\subfloat[Zipf0.8]{\label{top32Zipf08}\includegraphics[width = \matrixCellWidth]
			{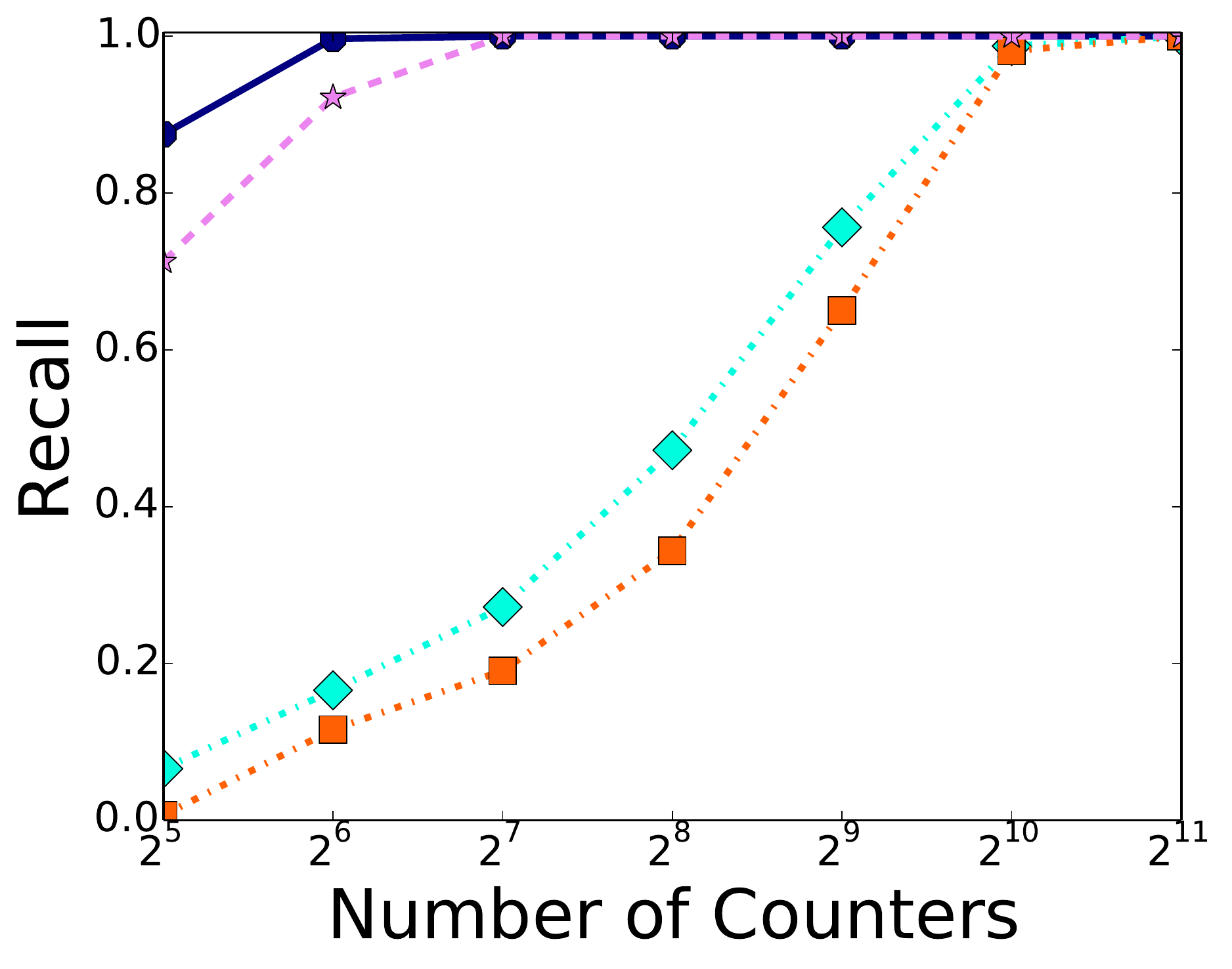}} \\
		\subfloat[Zipf1.0]{\label{top32Zipf10}\includegraphics[width = \matrixCellWidth]
			{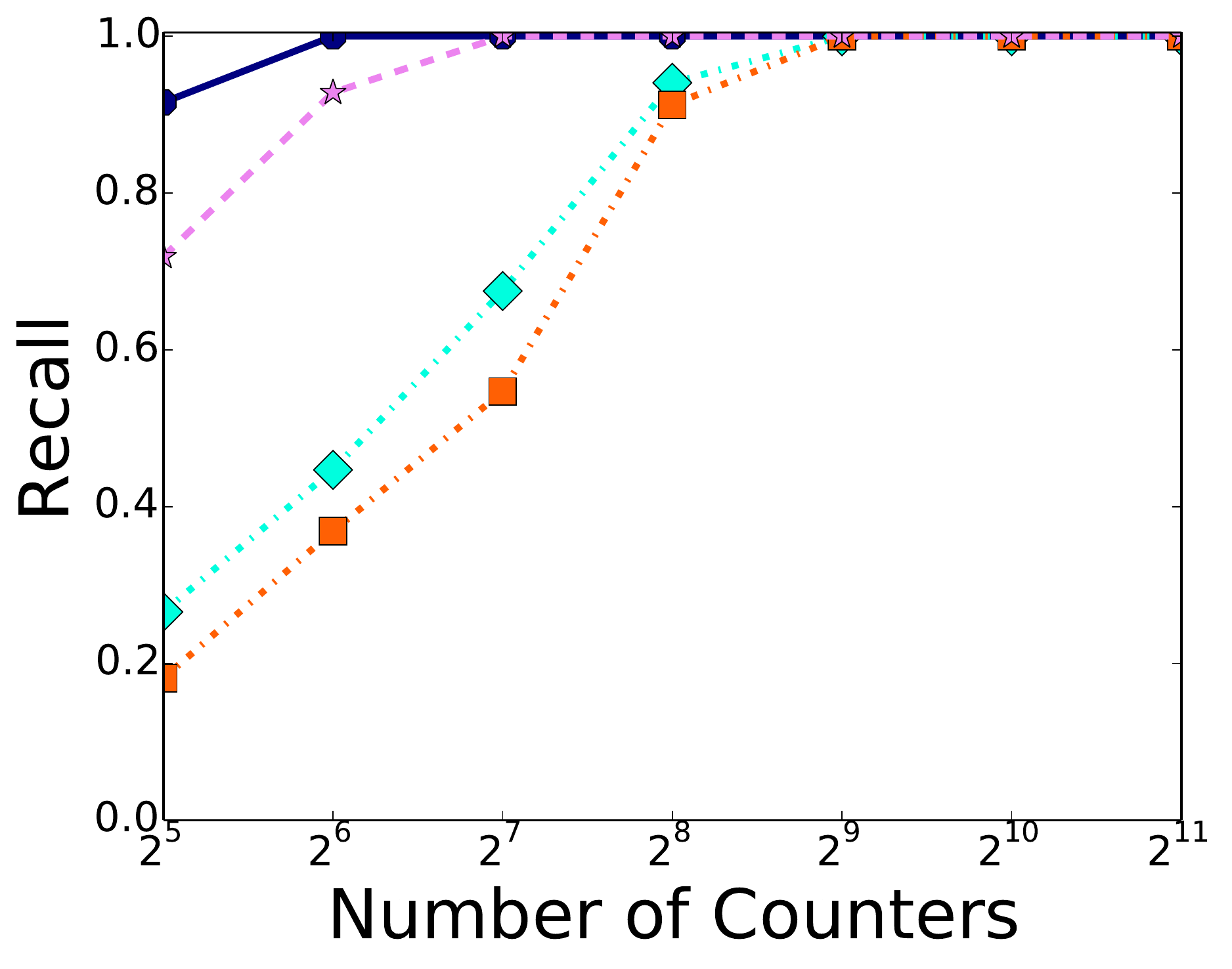}} &
		\subfloat[Zipf1.2]{\label{top32Zipf12}\includegraphics[width = \matrixCellWidth]
			{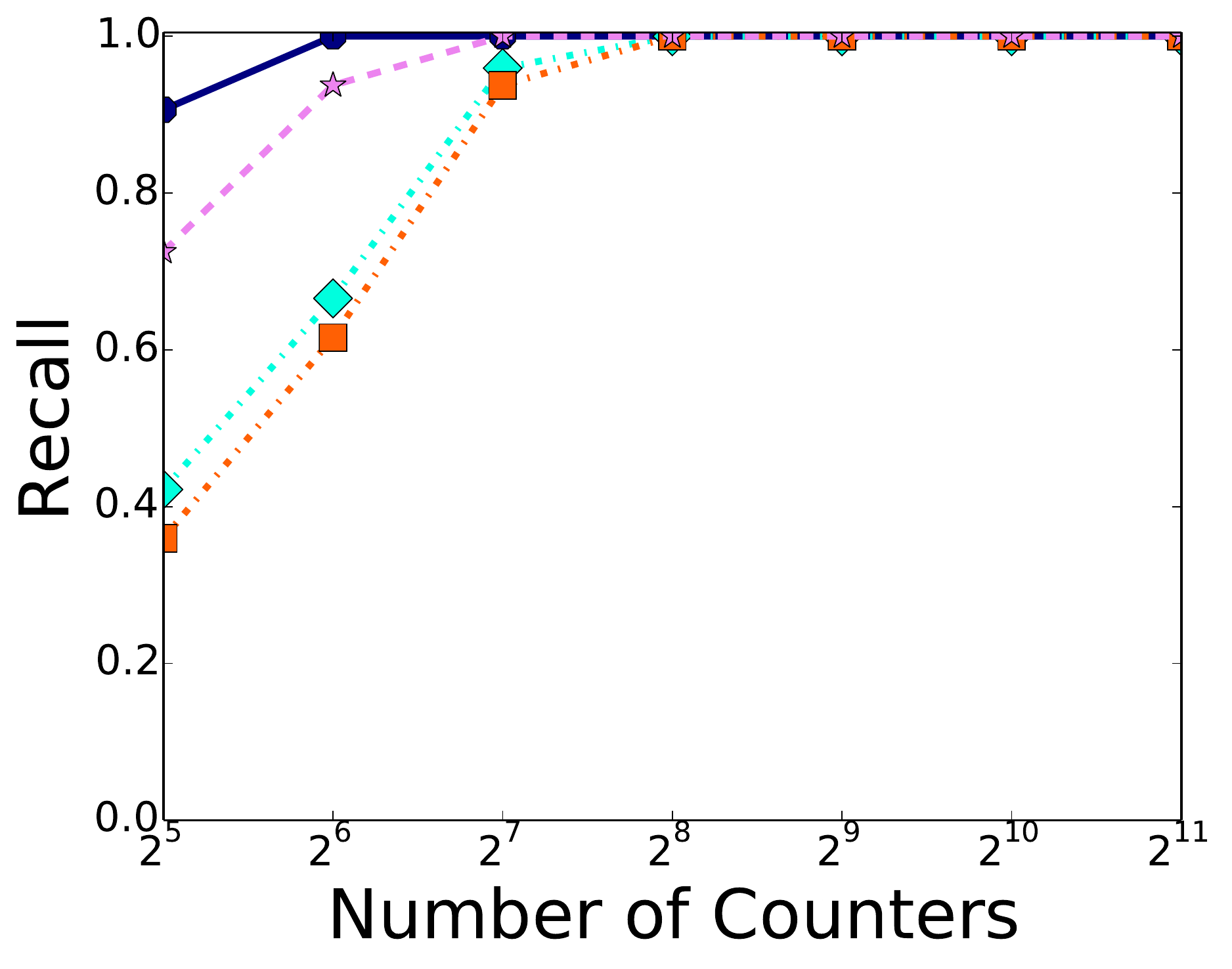}} &
		\subfloat[Zipf1.5]{\label{top32Zipf15}\includegraphics[width = \matrixCellWidth]
			{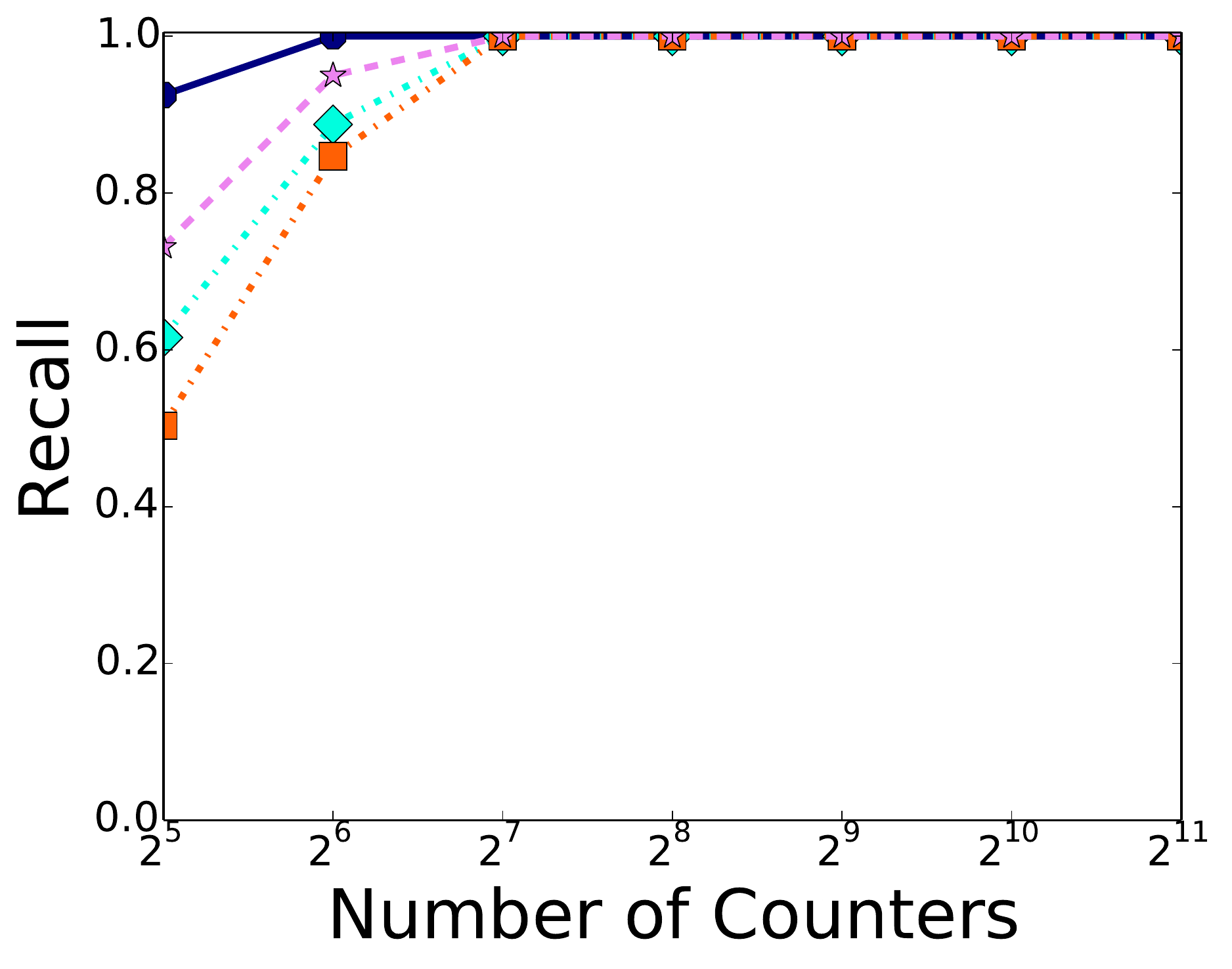}}
		
	\end{tabular}
	\caption{\label{fig:top32}The recall for finding the top-32 elements vs. number of counters.}
\end{figure*}

\subsection{On Arrival Evaluation}\label{sec:eval-on-arrival}
We begin our evaluation with the frequency estimation problem.
In this section, each data point was generated by averaging 10 disjoint batches of 1 million packets each, with the exception of YouTube, which is averaged over two 300,000 batches due to the small size of that trace.

Figure~\ref{fig:MSE} shows the MSE obtained by the different schemes when equipped with an increasing number of counters.
We experimented with different associativity levels to conclude that 16W-\PSS{} behaves almost as good as (the fully associative) \PSS{}.
%We start by stating that 16W-\PSS{} behaves almost as good as the (fully associative) \PSS{}.
\ifdefined\TENPAGES
Additional details appear in the full version~\cite{full-version}.
\fi
\ifdefined\EXTENDED
The experiment, appearing in Appendix~\ref{apx:assoc}, shows that while accuracy does improve as associativity grows, there is only little gain by increasing the associativity beyond $16$.
\fi
\footnote{In modern CPUs, caches have at least 8 ways and 64B lines; thus, each line can accommodate 2 entries, yielding a 16-way structure for our purposes.}
This remains true under all of our tested workloads.

Figure~\ref{fig:CaidaMSE} illustrates our results on the CAIDA packet trace.
For the entire range, \PSS{} and 16W-\PSS{} offer significantly lower error than the alternatives. Among the alternatives, none seems to be superior to the rest, as CS, CMS, FR and SS all have some settings where they are more accurate than the~rest.

Figure~\ref{fig:UCLAMSE} describes the results on the UCLA packet trace.
In this trace, \PSS{} is the leader for small memory configurations while for 512 counters and onwards SS is slightly better.
We believe this is due to the very high skewness of the trace, meaning that the $512$ most frequent elements already consist a significant fraction of the stream and therefore our randomization approach is not needed. 
%We believe that when the MSE is very small, SS may be better than \PSS{}.

The results for the YouTube traces are illustrated in Figure~\ref{fig:YouTubeMSE}.
In this trace, \PSS{} and 16W-\PSS{} are more accurate than the alternatives.
Looking only at the alternatives, it is unclear which is the leading among them.
However, they all require more than x4 the space to match the accuracy of \PSS{}.

\paragraph{Synthetic Traces}
Synthetic Zipf traces provide us with better insight on the conditions where \PSS{} works best.
The least skewed shown distribution is Zipf 0.6 in Figure~\ref{fig:Zipf0.6MSE} and the most skewed distribution is Zipf 1.5 in Figure~\ref{fig:Zipf1.5MSE}.
It appears that \PSS{} performs very well in all these distributions while the alternatives only perform well when the distribution is skewed enough. Figure~\ref{fig:Zipf0.6MSE} shows that for Zipf 0.6, SS with 2048 counters obtains worse MSE than \PSS{} with 32 counters!
In Figure~\ref{fig:Zipf0.8MSE}, we see that 2048 counter SS is about as accurate as a 128 counters \PSS{}.
Figure~\ref{fig:Zipf1.0MSE} exhibits that for Zipf 1, \PSS{} with 256 counters has similar accuracy as SS with 2048 counters.
The trend continues until in Figure~\ref{fig:Zipf1.5MSE} the accuracy of \PSS{} with 1024 counters is similar to SS with 2048.
For the entire range, \PSS{} requires significantly less space.

%We now evaluate \PSS{} in a variety of synthetic distributions. Specifically Zipf 1.5 and 1.2 that have few extremely frequent items without a tail. Zipf 1.0 that has both very frequent items and an heavy tail and Zipf 0.8 and 0.6 that have mildly frequent items and very heavy tails. The latter distributions are considered very hard for approximate frequency estimation and in comparison to the alternatives \PSS{} does very well on these distributions.  Figure~\ref{fig:Zipf0.8MSE} and Figure~\ref{fig:Zipf0.6MSE} show results for Zipf 0.8 and Zipf 0.6. As can be observed for these distributions, \PSS{} is the only alternative that provides an accurate solution. In these cases \PSS{} is almost 100 times more accurate then the best alternative. 	

%Figures~\ref{fig:Zipf1.5MSE} and Figure~\ref{fig:Zipf1.2MSE} illustrate the result for these two power-low distribution. As can be observed while Count Min Sketch and Space Saving are both relatively efficient in these cases. \PSS{} out performs all the algorithms for the selected range.  Figure~\ref{fig:Zipf1.0MSE} shows result for Zipf 1.0 distribution as can be observed, the error with \PSS{} is more than 10 times smaller than the best alternative.

\subsection{Top-k Identification}\label{sec:eval-top-k}
\ifdefined \NINEPAGES
\else
\begin{figure*}[t!]
	\begin{tabular}{ccc}
		\subfloat[CAIDA]{\label{top512Caida}\includegraphics[width = \matrixCellWidth]
			{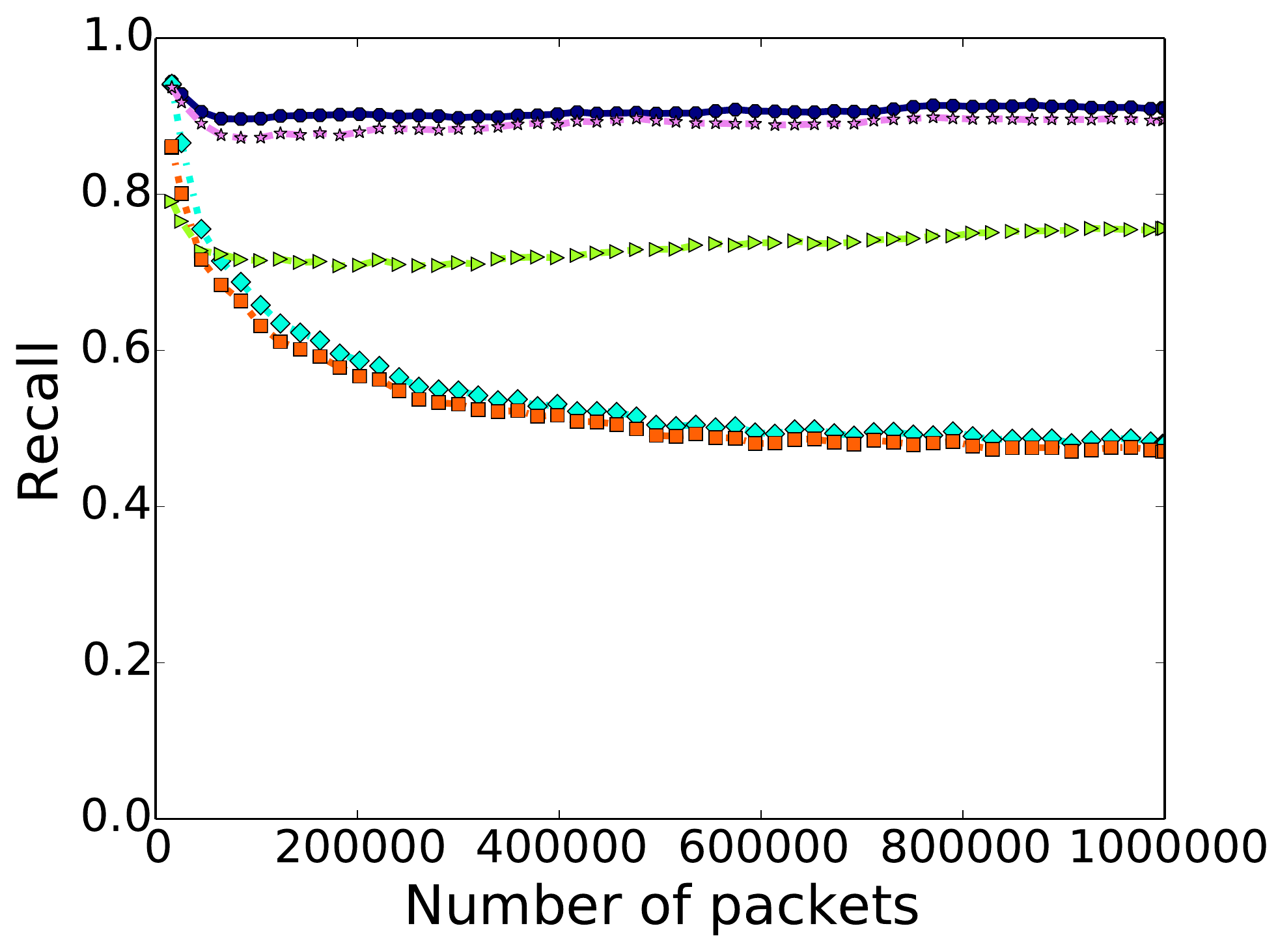}} &
		\subfloat[UCLA]{\label{top512UCLA}\includegraphics[width = \matrixCellWidth]
			{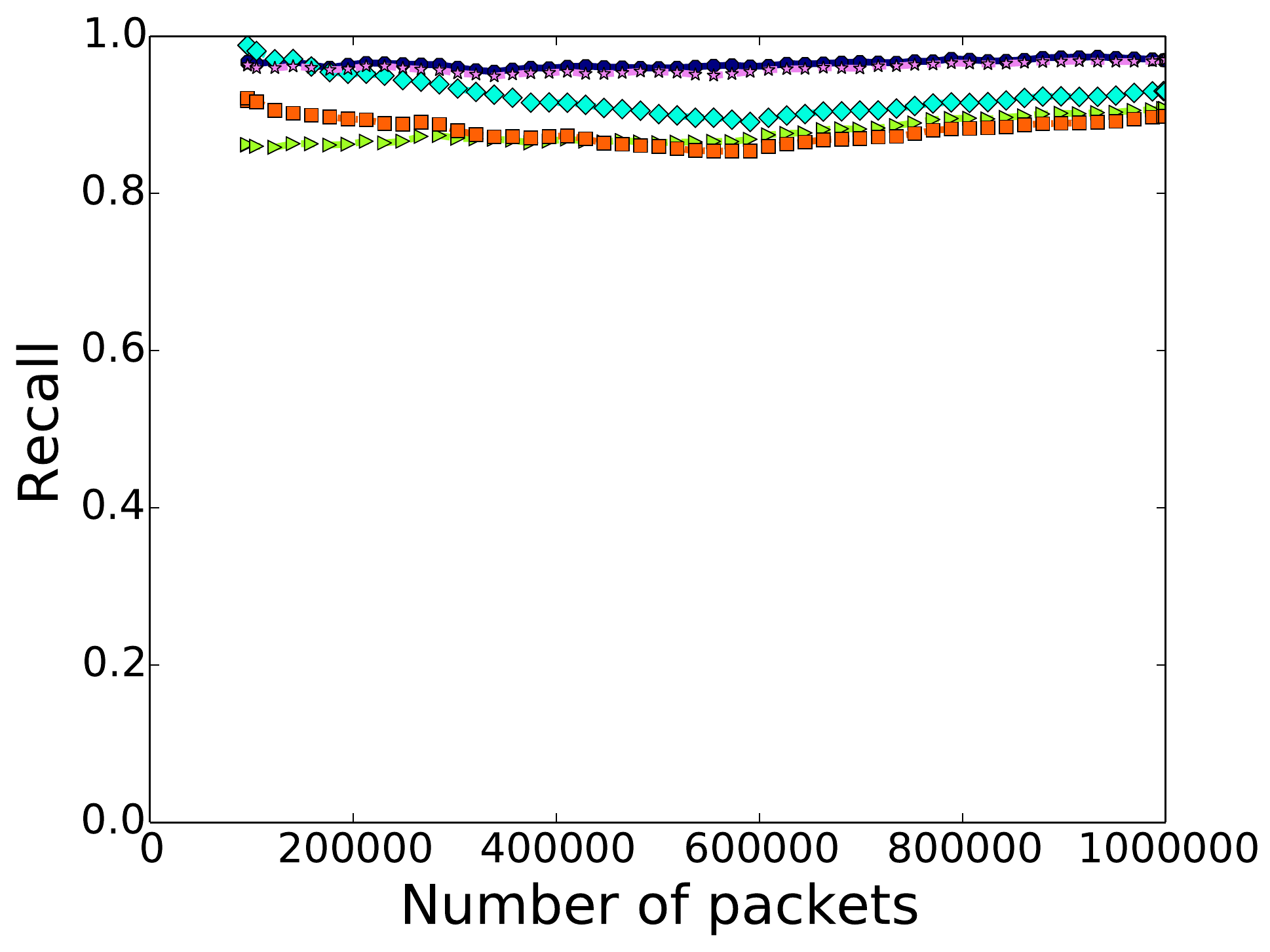}} &
		\subfloat[YouTube]{\label{top512YouTube}\includegraphics[width = \matrixCellWidth]
			{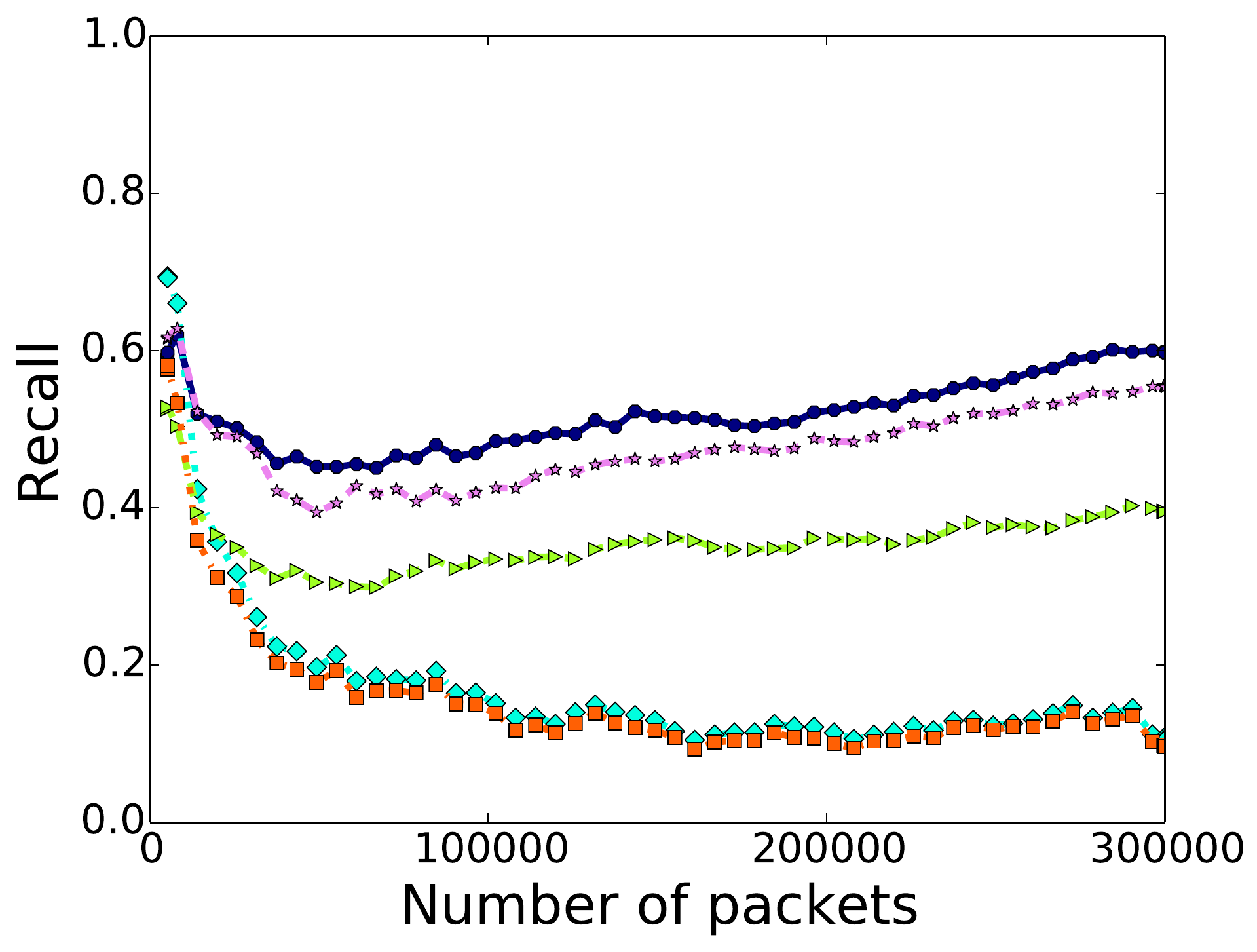}}\\
		\subfloat[Zipf0.6]{\label{top512Zipf06}\includegraphics[width = \matrixCellWidth]
			{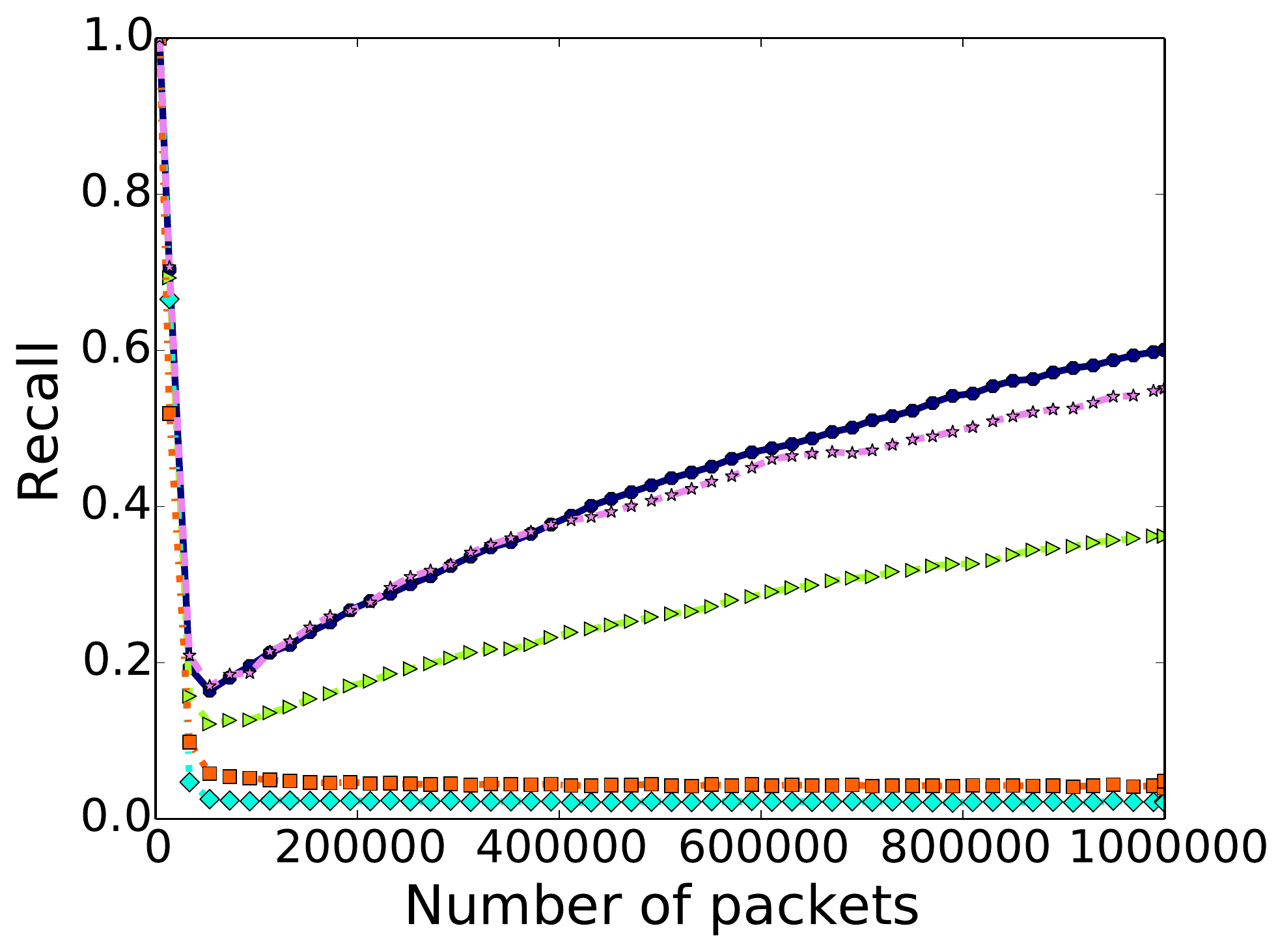}} &
		\hspace*{0.05in}
		\subfloat[Legend]{\includegraphics[width=4.7cm,height=4.2cm]
			{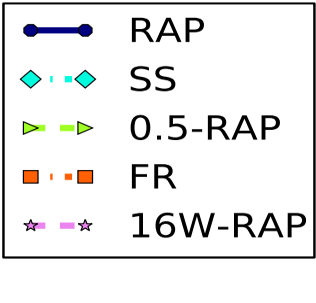}}&		
		\subfloat[Zipf0.8]{\label{top512Zipf08}\includegraphics[width = \matrixCellWidth]
			{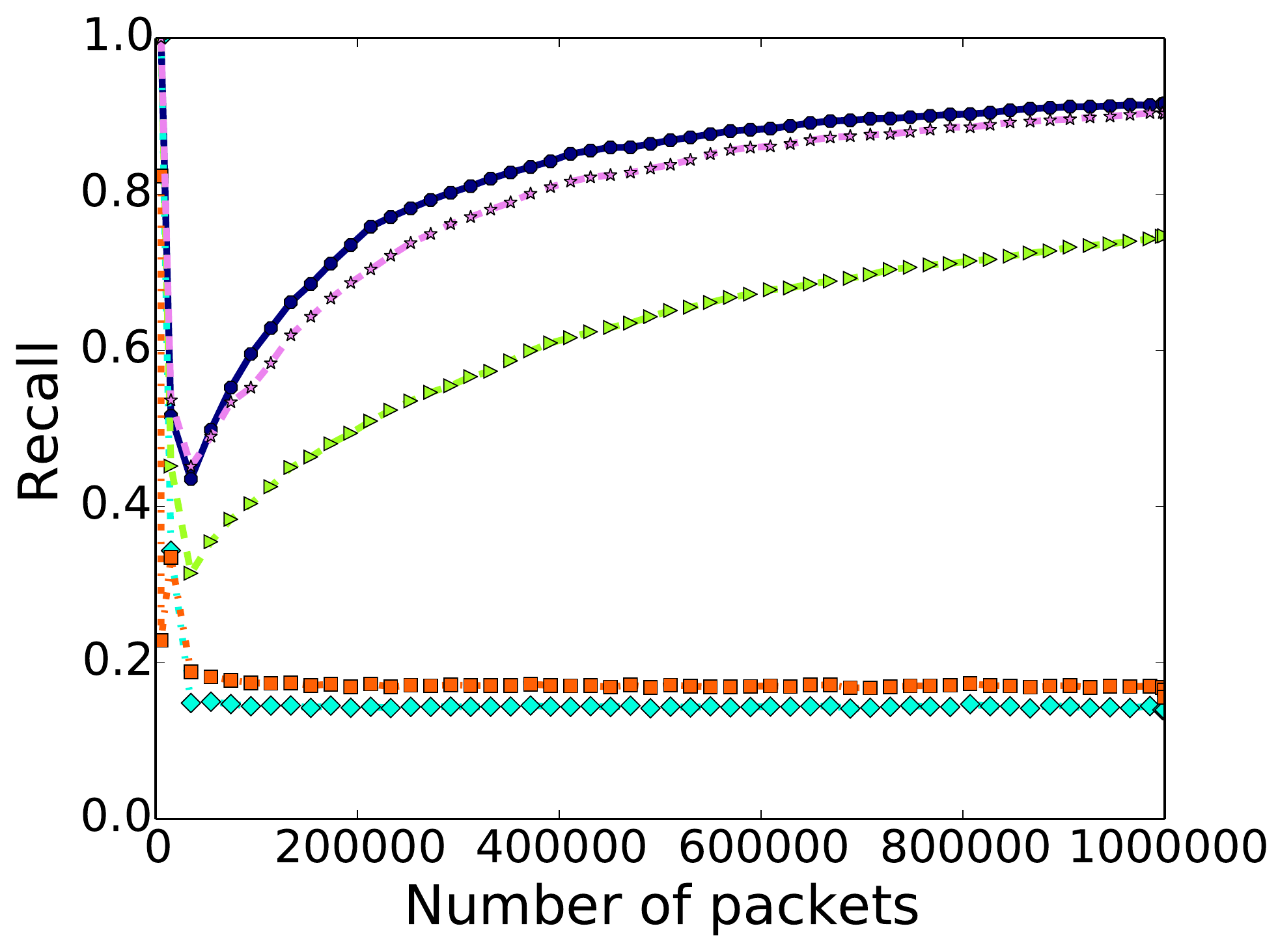}} \\
		\subfloat[Zipf1.0]{\label{top512Zipf10}\includegraphics[width = \matrixCellWidth]
			{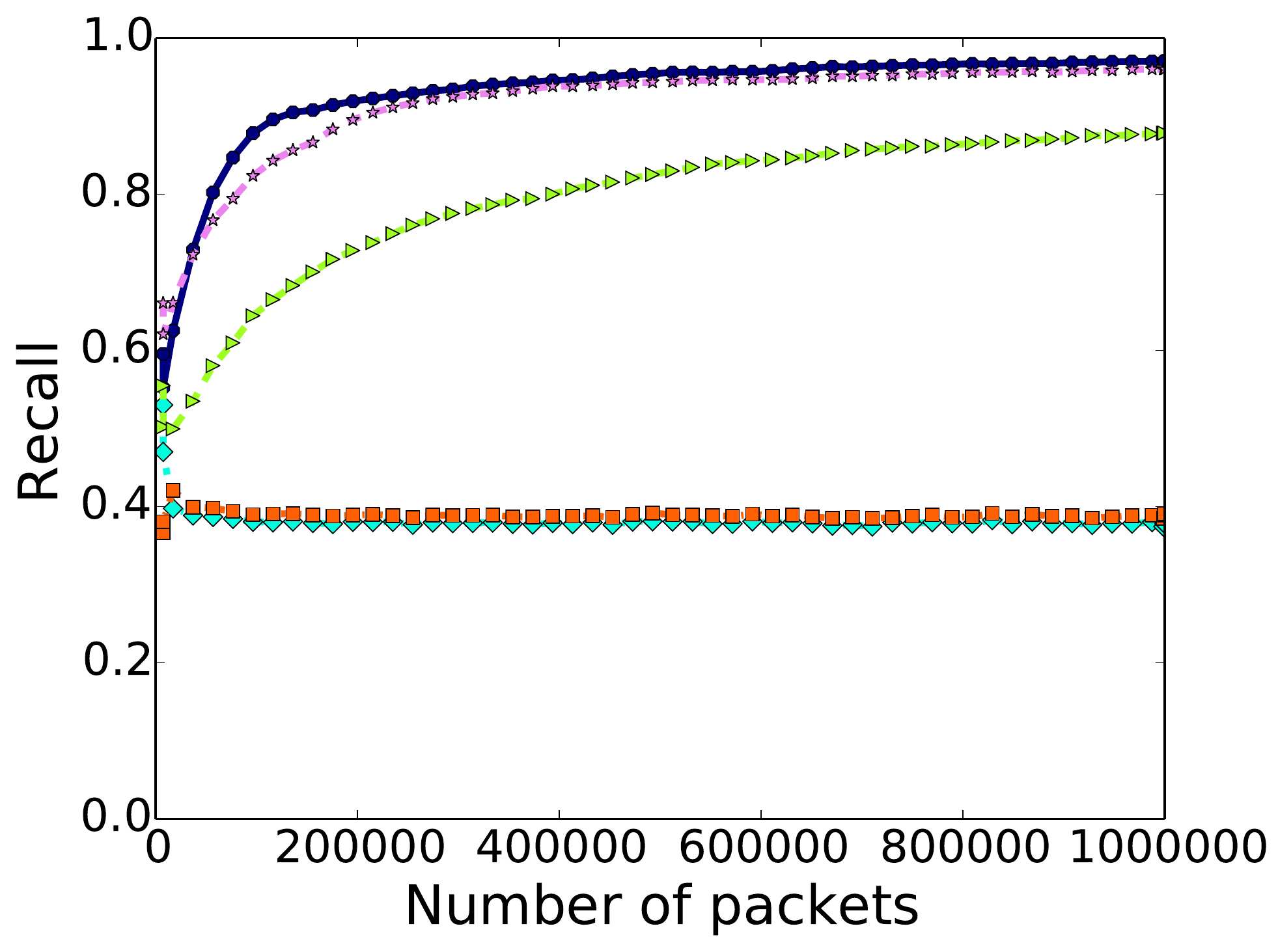}} &
		\subfloat[Zipf1.2]{\label{top512Zipf12}\includegraphics[width = \matrixCellWidth]
			{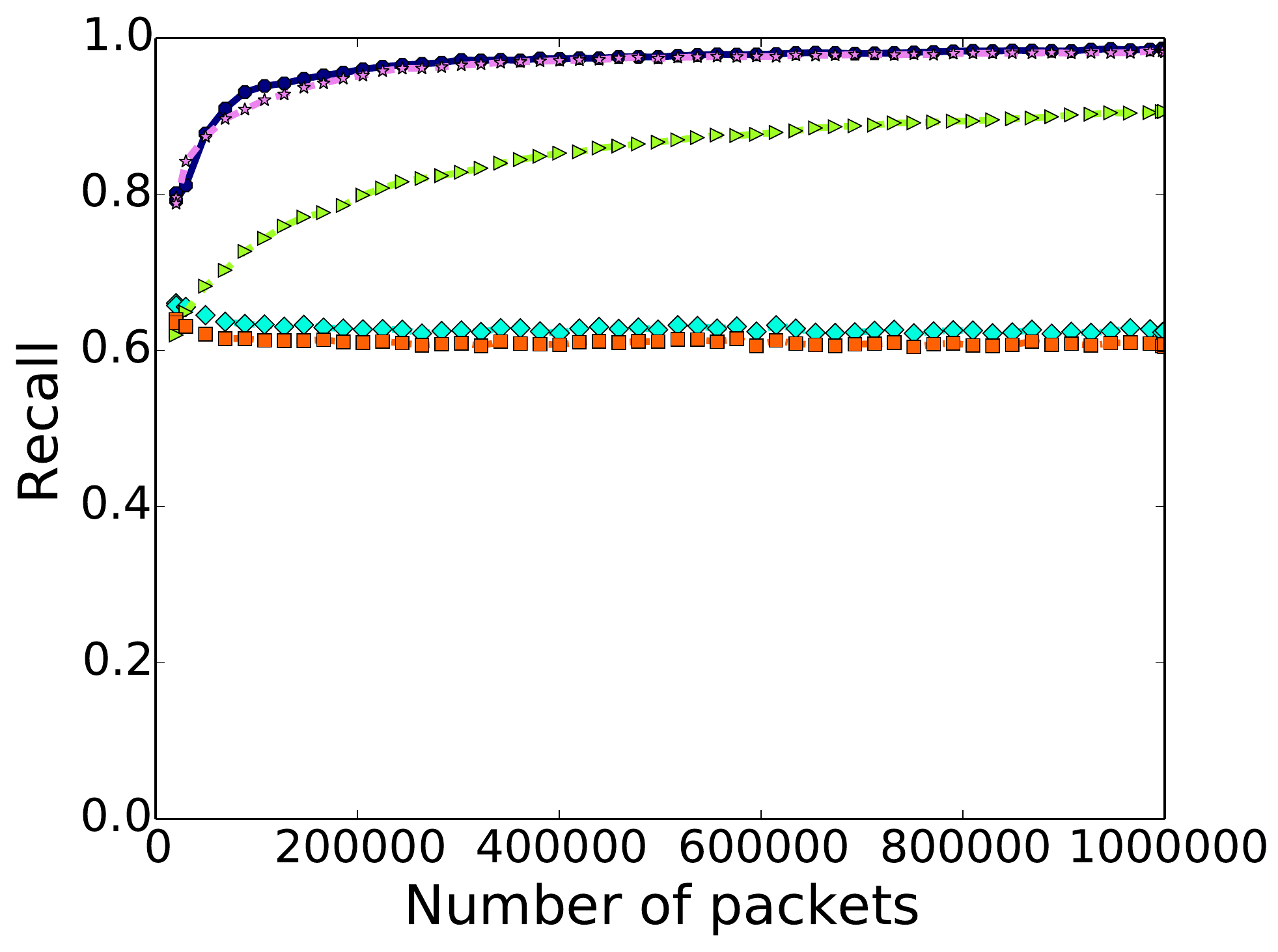}} &
		\subfloat[Zipf1.5]{\label{top512Zipf15}\includegraphics[width = \matrixCellWidth]
			{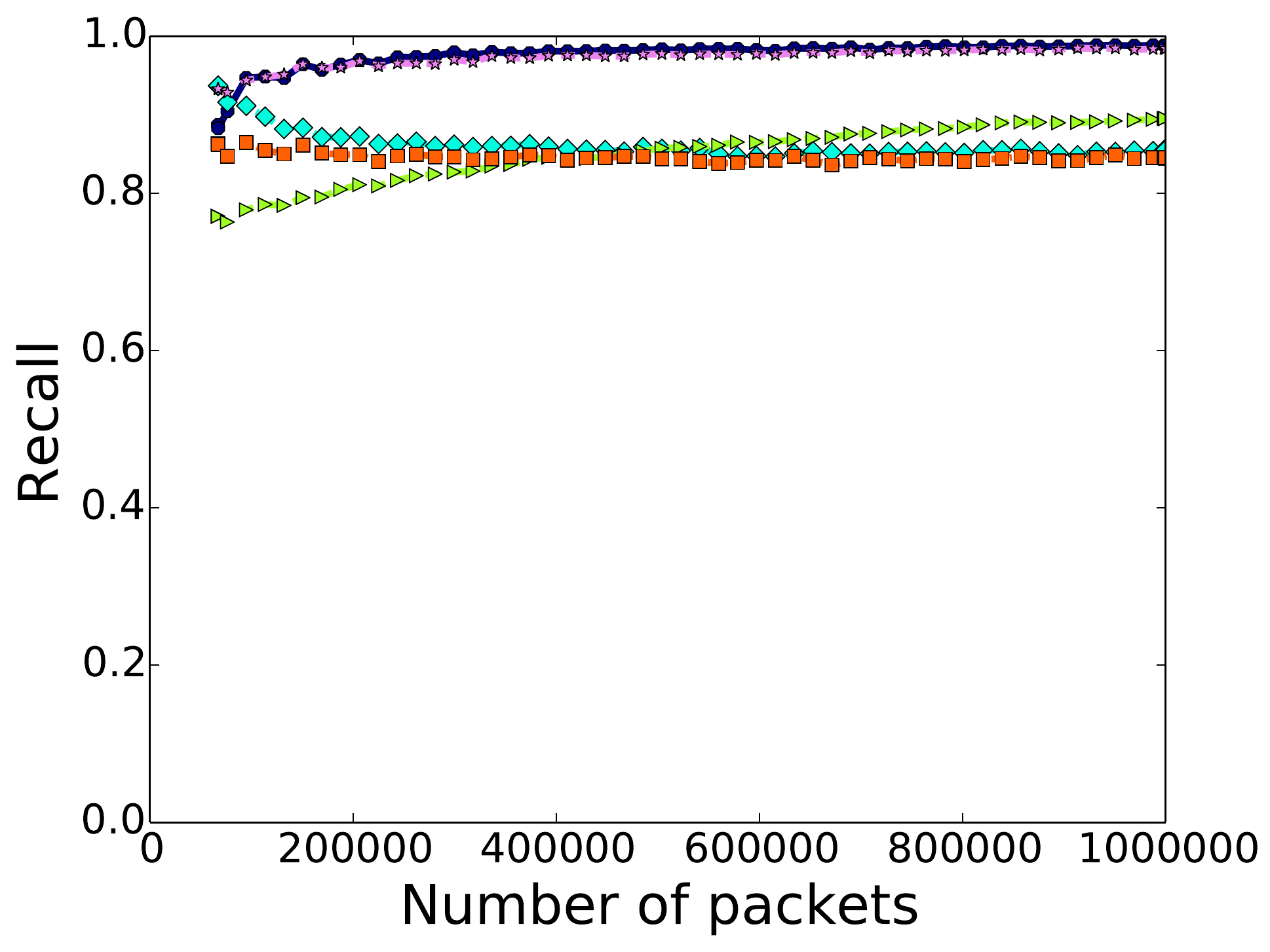}}
	\end{tabular}
	\caption{\label{top512}The average recall achieved for identifying the top-$512$ using 1024 counters, compared to the size of the stream. }
\end{figure*}
\fi

\begin{figure*}[t]
	\begin{tabular}{ccc}
		\subfloat[CAIDA]{\label{top512tradeoffCaida}\includegraphics[width = \matrixCellWidth]
			{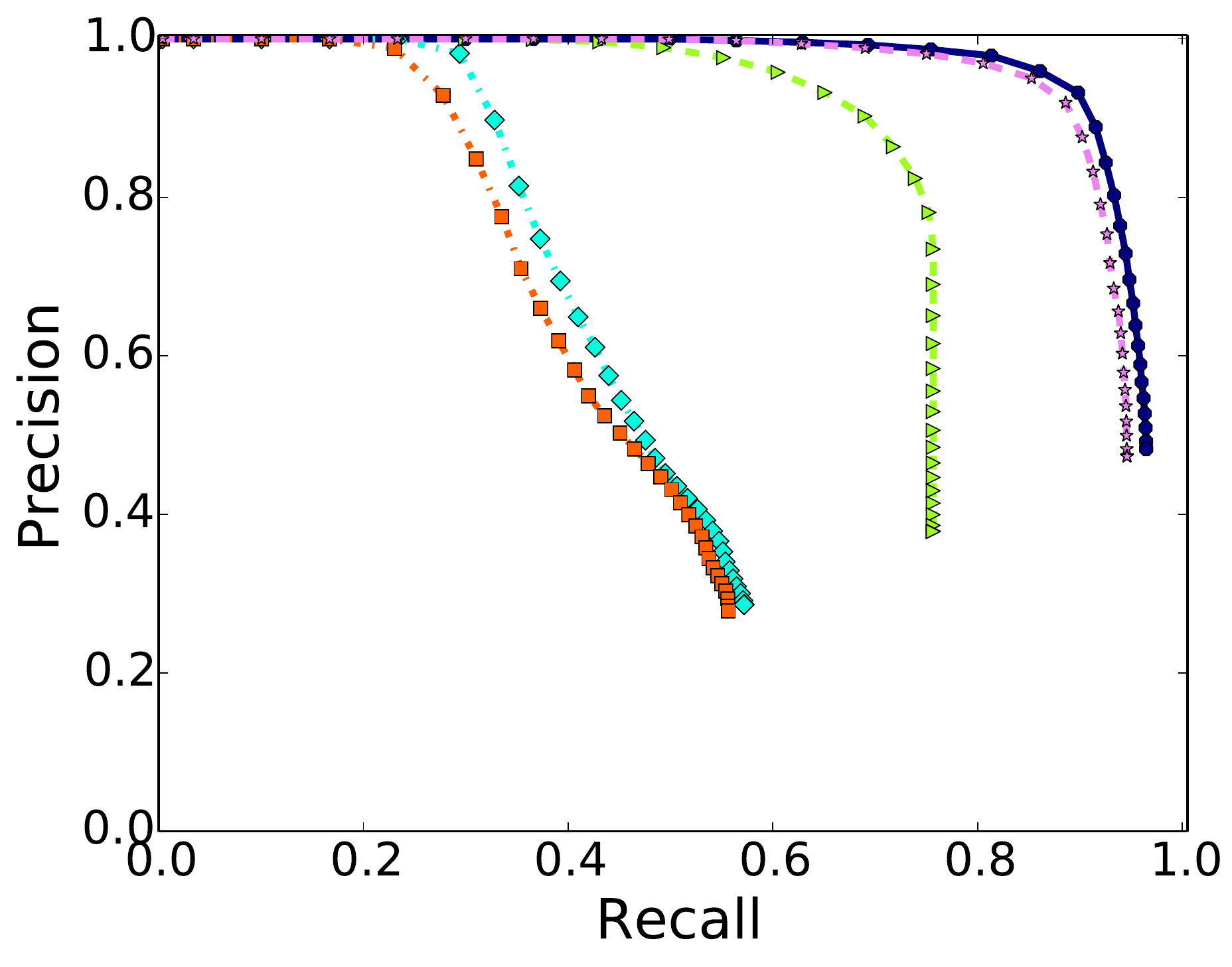}} &
		\subfloat[UCLA]{\label{top512tradeoffUCLA}\includegraphics[width = \matrixCellWidth]
			{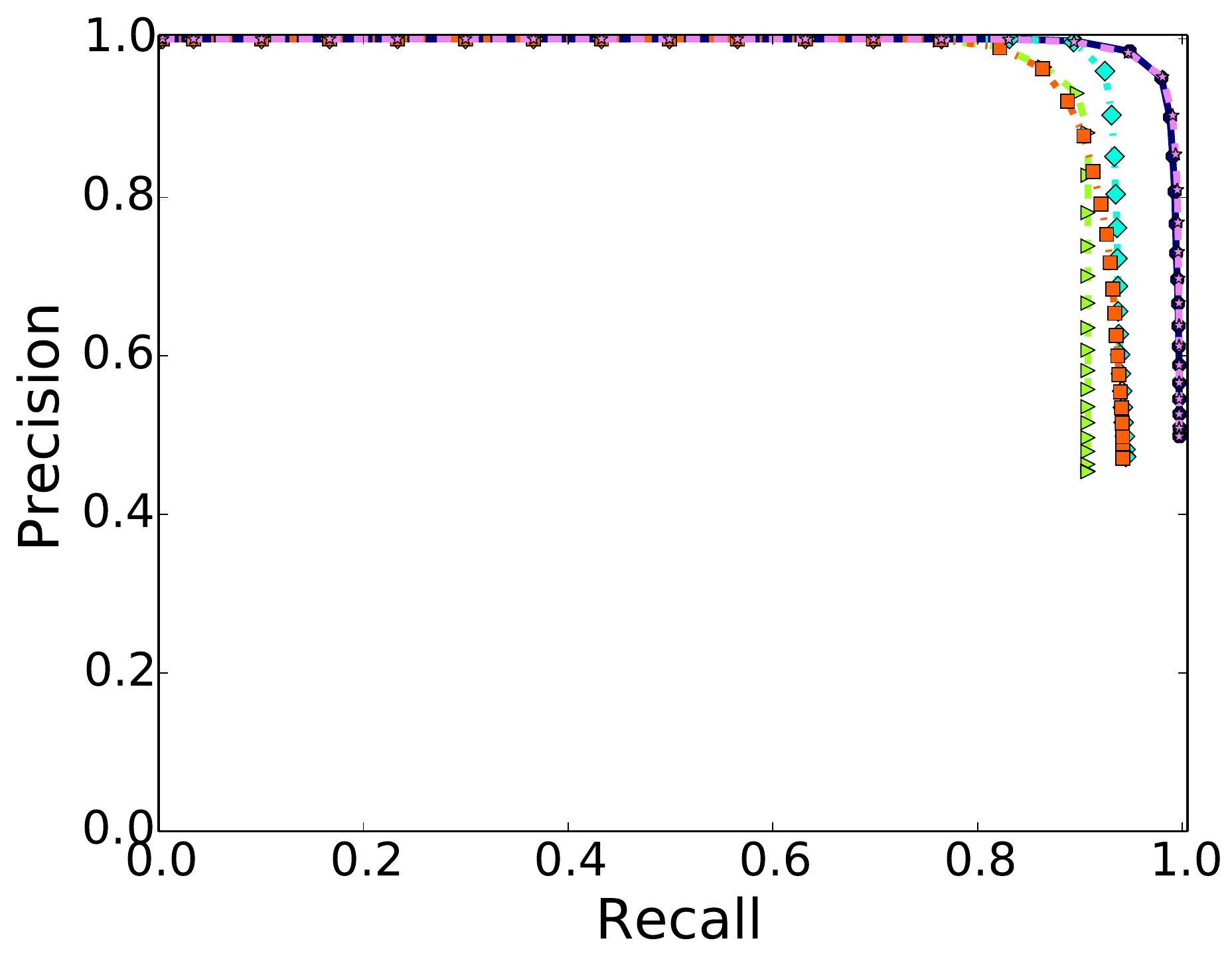}} &
		\subfloat[YouTube]{\label{top512tradeoffYouTube}\includegraphics[width = \matrixCellWidth]
			{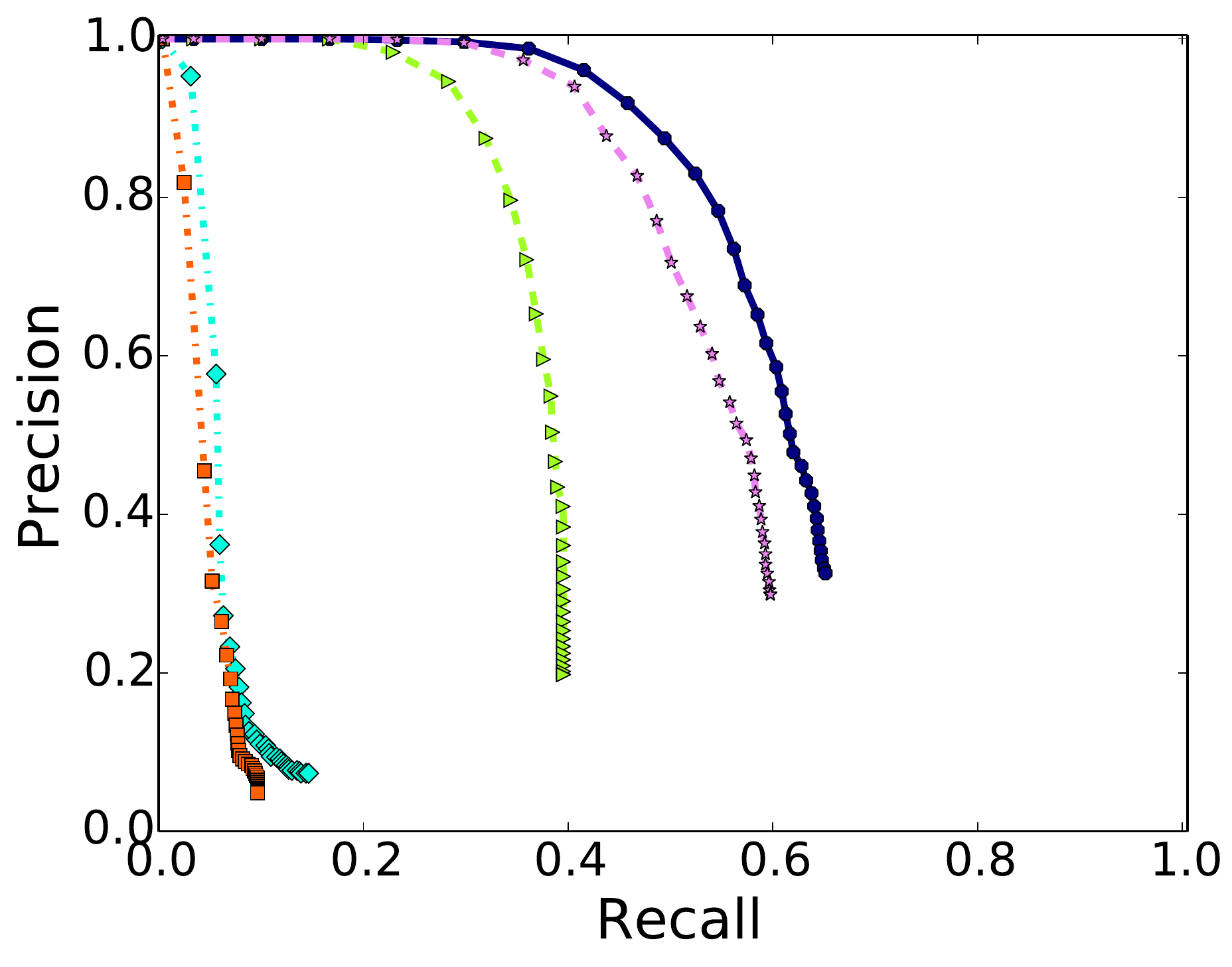}}\\
		\subfloat[Zipf0.6]{\label{top512tradeoffZipf06}\includegraphics[width = \matrixCellWidth]
			{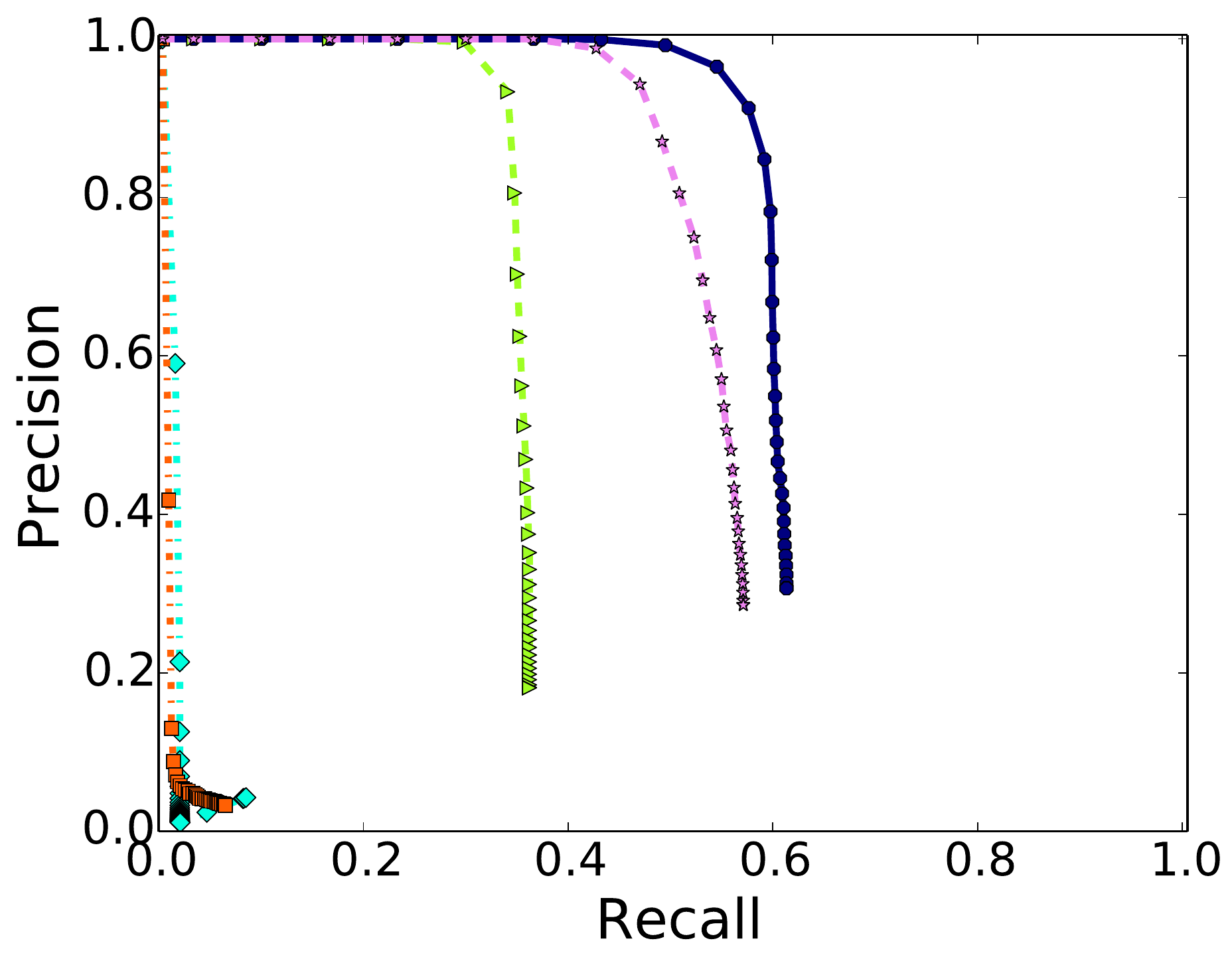}} &
		\hspace*{0.1in}
		\subfloat[Legend]{\includegraphics[width = 5.0cm, height=4.45cm]
			{TopKLegend.png}}&		
		\subfloat[Zipf0.8]{\label{top512tradeoffZipf08}\includegraphics[width = \matrixCellWidth]
			{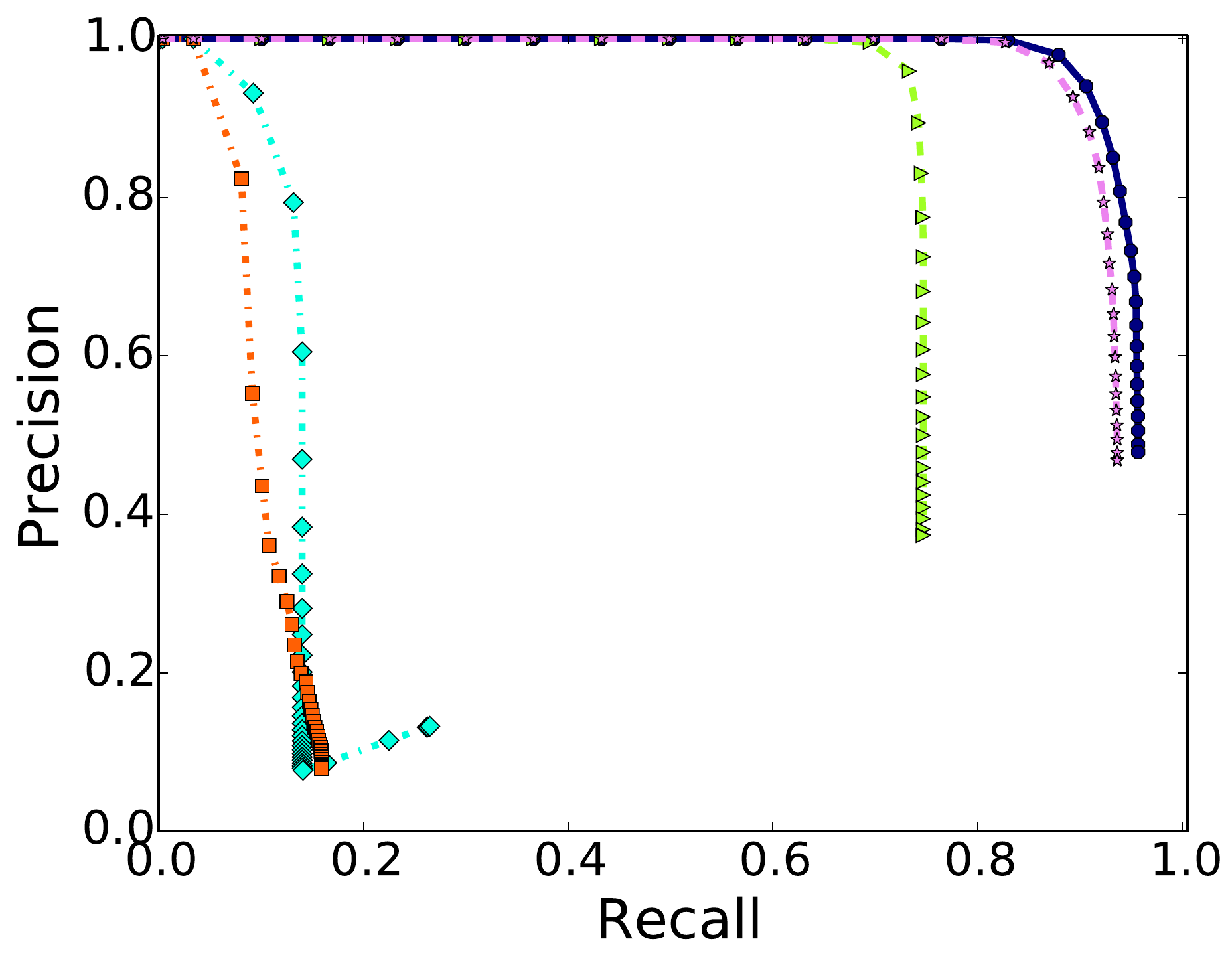}} \\
		\subfloat[Zipf1.0]{\label{top512tradeoffZipf10}\includegraphics[width = \matrixCellWidth]
			{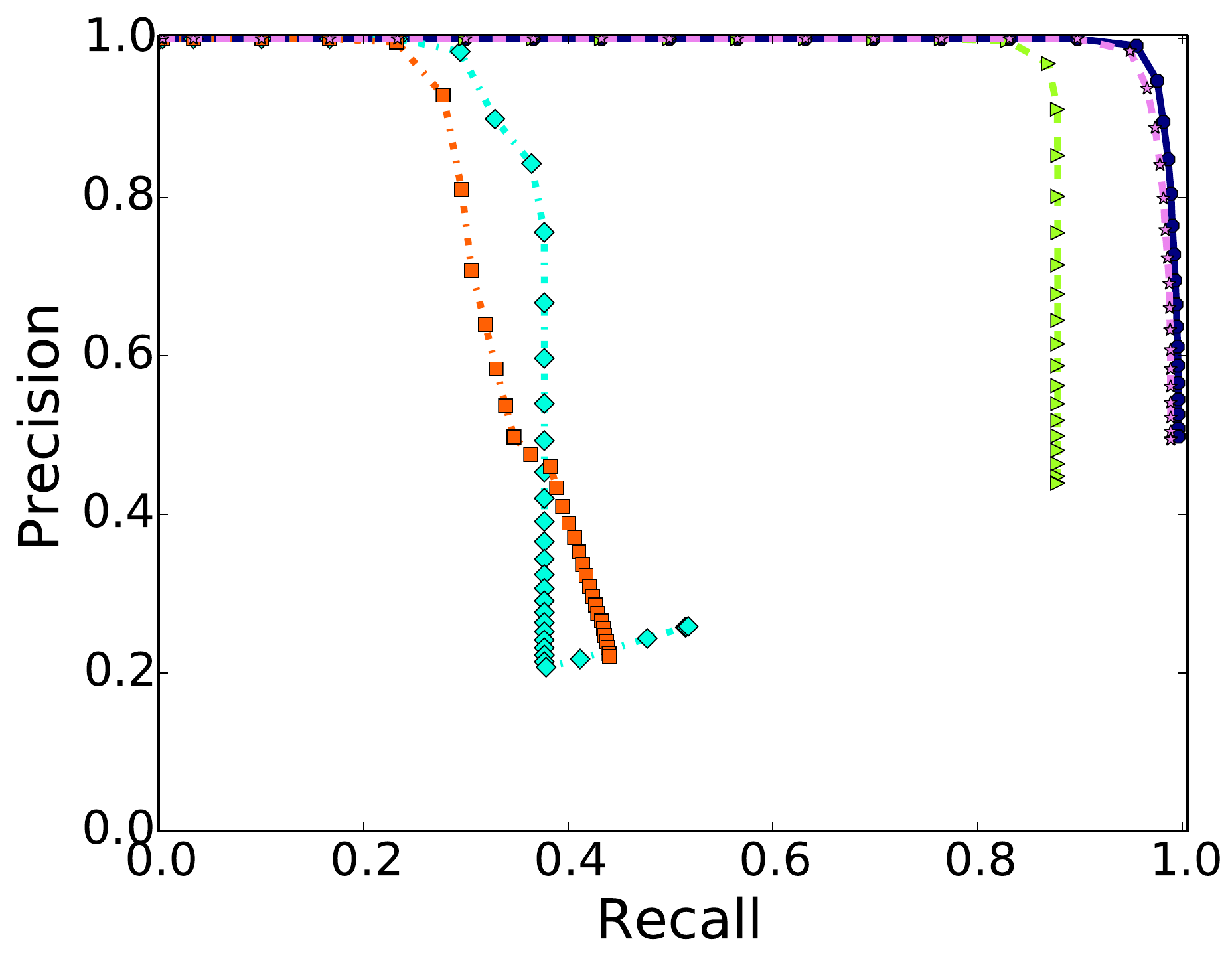}} &
		\subfloat[Zipf1.2]{\label{top512tradeoffZipf12}\includegraphics[width = \matrixCellWidth]
			{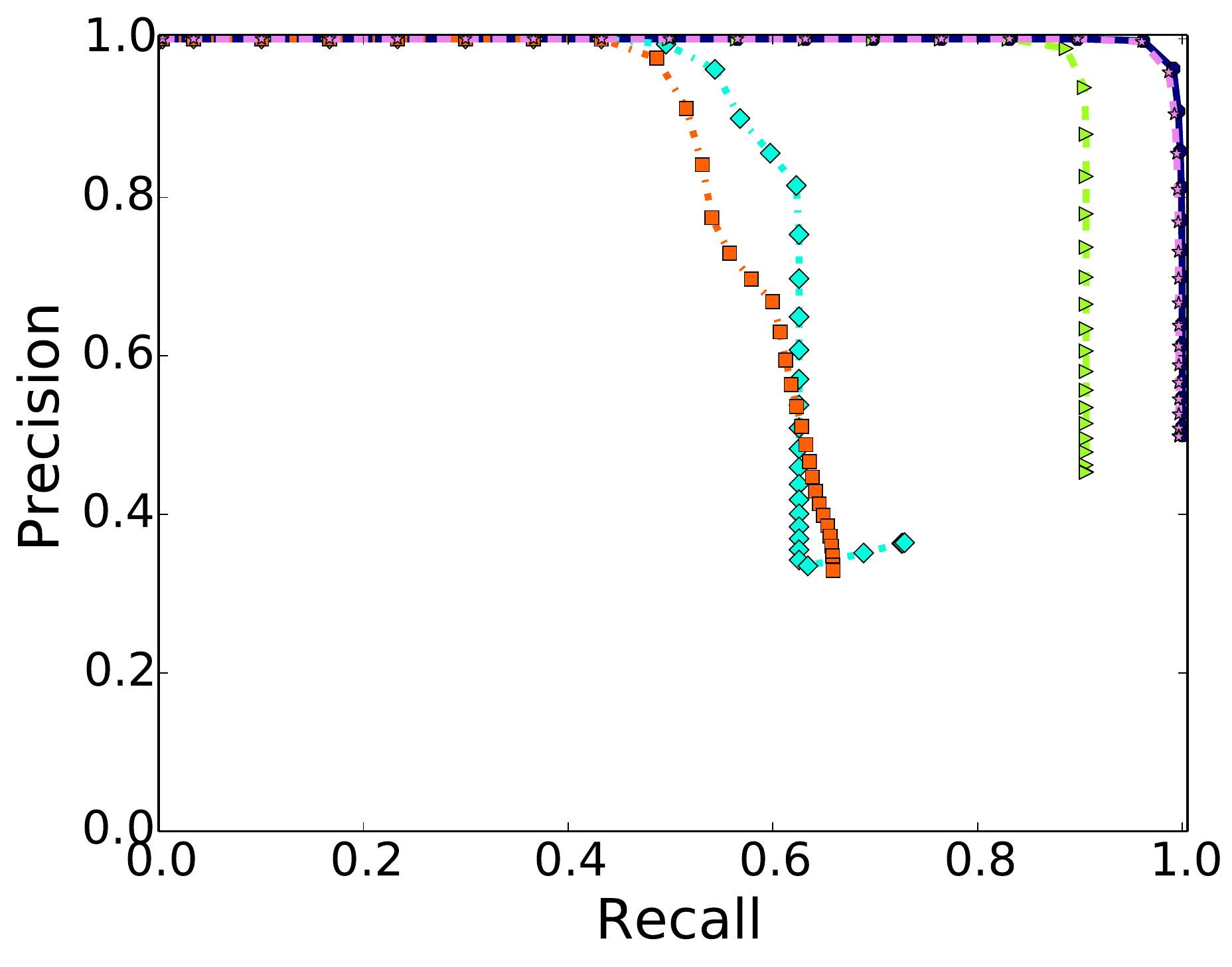}} &
		\subfloat[Zipf1.5]{\label{top512tradeoffZipf15}\includegraphics[width = \matrixCellWidth]
			{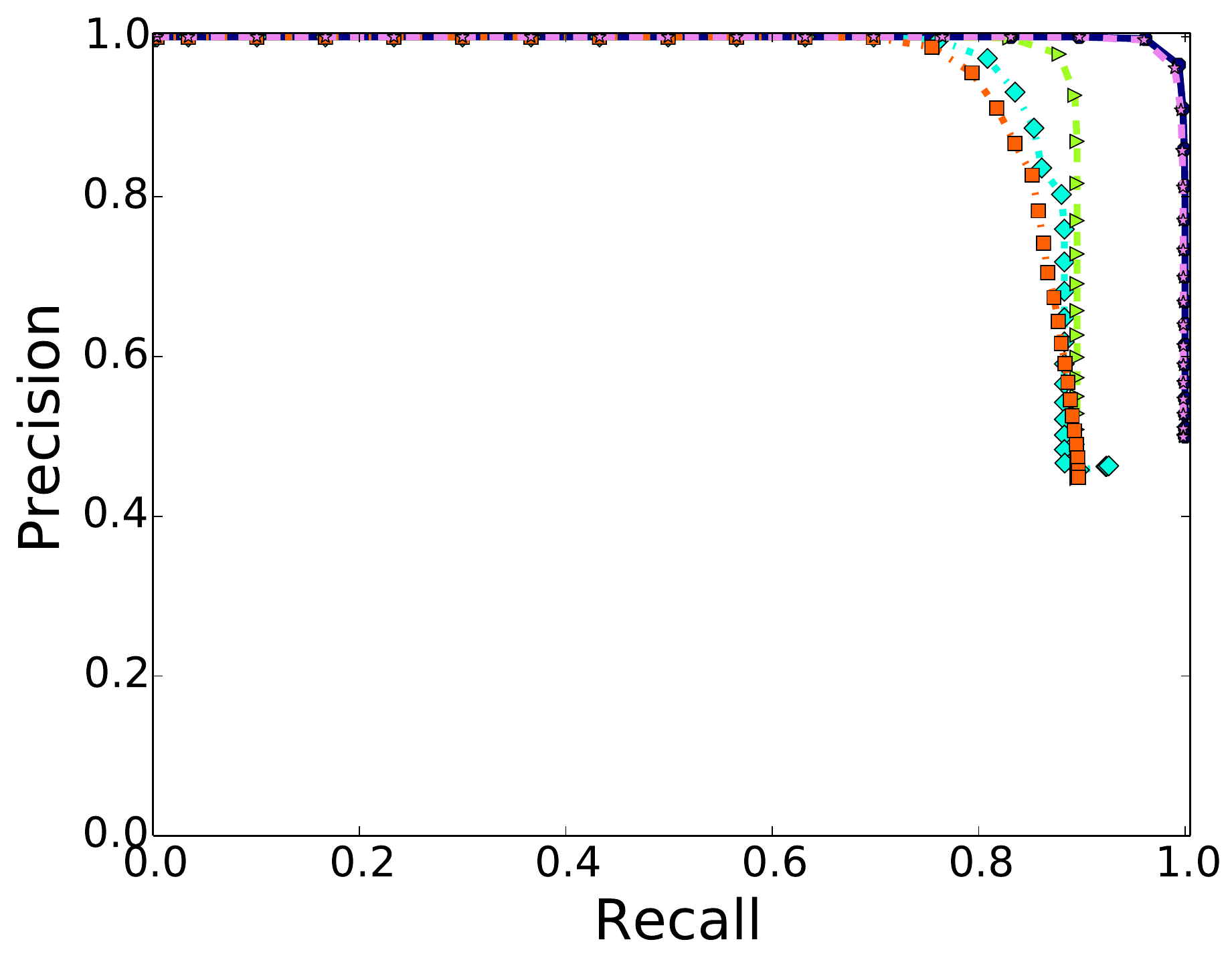}}
	\end{tabular}
	\caption{\label{top512tradeoff}The Precision-Recall curves for identifying the top-$512$ with 1024 counters. }
\end{figure*}

Since, CS and CMS do not solve the top-$k$ problem, we only compare \PSS{} and 16W-\PSS{} to SS and FR.
In some of the figures we also gave \PSS{} half the number of counters as the rest.
That configuration is marked as 0.5-\PSS{}.

\paragraph{Top-32}
First, we consider identifying the top-$32$ flows.
%We measure the relation between the number of counters and the obtained recall.
We measure the obtained recall for a given number of counters.
Our results, summarized in Figure~\ref{fig:top32}, demonstrating that 16W-\PSS{} is almost as accurate as \PSS{} in all workloads.

Figure~\ref{top32Caida} presents results for CAIDA.
As shown, \PSS{} and 16W-\PSS{} achieve near perfect recall with 128 counters.
In contrast, FR and SS require 1024 counters for the same recall.
Results for UCLA are in Figure~\ref{top32UCLA}.
As can  be seen, \PSS{} and 16W-\PSS{} reach near optimal recall with 128 counters while FR and SS require 256 counters.

Figure~\ref{top32YouTube} shows that in the YouTube workload, 1024 counters are not enough for SS and FR, and they reach less than 0.5 recall with 1024 counters.
In comparison, \PSS{} and 16W-\PSS{} reach near perfect recall with 512 counters.

We now use synthetic Zipf distributions to characterize the performance of \PSS{} and 16W-\PSS{}.
Figure~\ref{top32Zipf06} shows that for the mildly skewed Zipf 0.6, both \PSS{} and 16W-\PSS{} achieve near optimal recall with 256 counters, while for the alternatives even 2048 are not enough.
In the more skewed Zipf 0.8 distribution, Figure~\ref{top32Zipf08} shows that \PSS{} requires 64 counters, and 16W-\PSS{} requires 128 to achieve near perfect recall. In contrast, SS and FR require 1024 counters~to~do~the~same.

For Zipf 1.0 distribution, Figure~\ref{fig:Zipf1.0MSE} shows that \PSS{} and 16W-\PSS{} still require 64 and 128 counters to achieve near optimal recall.
FR and SS now require 512 counters to do the same.
In Zipf 1.2 distribution, Figure~\ref{fig:Zipf1.2MSE} shows that \PSS{} and 16W-\PSS{} continue to require 64 and 128 counters, while SS and FR now require 256 counters to achieve near optimal recall.
Finally, for the very skewed Zipf 1.5, Figure~\ref{fig:Zipf1.5MSE} shows that \PSS{} and 16W-\PSS{} still require 64 and 128 counters to achieve near optimal recall while SS and FR now require 128 counters.
To sum it up, \PSS{} shows a reduction of 2x-16x that depending on workload skewness.
% requires 64 counters and 16W-\PSS{} 128 counters. In this workload tough, the alternatives also require 128 counters to identify the top 32. However, even in these conditions \PSS{} requires half as many counters and 16W-\PSS{} converges better than the alternatives.

\paragraph{Convergence Speed}
Since \PSS{} is a randomized algorithm, its convergence speed is as important as its performance for large streams. In order to evaluate the convergence speeds of the different algorithms, we consider the problem of identifying the top-$512$ flows with $1024$ counters.
\ifdefined \NINEPAGES
Our results, deferred to the full version~\cite{full-version} due to lack of space, show that \PSS{} achieves superior recall starting at very early stages of the streams, even when allocated only with $512$ counters.
\else
Note that the top-512 items also change during the trace and the algorithms are required to constantly adjust.
To quantify the improvement of \PSS{} in these settings we use 0.5-\PSS{}, which attempts to identify the top-512 with mere 512 counters.
Our results are illustrated in Figure~\ref{top512}.
As can be seen, \PSS{} and 16W-\PSS{} offer similar recall rates in all the tested workloads.

Figure~\ref{top512Caida} exhibits the results for CAIDA workload.
As shown, \PSS{}, 16W-\PSS{} and even 0.5-\PSS{} are significantly better than SS and FR.
In Figure~\ref{top512UCLA}, we see that in the UCLA workload, \PSS{} and 16W-\PSS{} achieve around 97\% recall while SS and FR achieve less than 92\%.
Interestingly, in this trace, even 0.5-\PSS{} is above 90\% recall with just 512 counters.

Figure~\ref{top512YouTube} shows results for the difficult YouTube trace.
%As can be seen, the YouTube workload is significantly more difficult.
Yet, \PSS{} and 16W-\PSS{} achieve above 50\% recall, while SS and FR are constantly under 20\% recall.
The improvement is more than x2, as 0.5 \PSS{} is significantly better than FR and SS.

We now use synthetic workloads to identify the performance envelope of \PSS{} and 16W-\PSS{}.
We can observe that as the workload becomes more skewed, all the algorithms improve but \PSS{} and 16W-\PSS{} remain considerably better than the alternatives.
For the mildly skewed Zipf 0.6 distribution, Figure~\ref{top512Zipf06} shows that \PSS{} and 16W-\PSS{} achieve over 50\% recall while FR and SS are under 10\%.
In the slightly more skewed Zipf 0.8, Figure~\ref{top512Zipf08} exhibits that \PSS{} and 16W-\PSS{} yield over 90\% recall, while the alternatives are slightly below 20\%. Similarly, in Zipf 1.0, Figure~\ref{top512Zipf10} shows that  \PSS{} and 16W \PSS{} are  close to 100\% recall during the entire trace while SS and FR are only at 40\% recall.
In  Zipf 1.2, Figure~\ref{top512Zipf12} shows that SS and FR reach 60\% while \PSS{} and 16W-\PSS{} remain at 100\%.
Finally for Zipf 1.5, Figure~\ref{top512Zipf15} shows that SS and FR are approaching 85\% recall.
Even in this workload, 0.5-\PSS{} eventually achieves higher recall.
In all these cases, 0.5-\PSS{} achieves higher recall than SS and FR and the space reduction is at least x2.
\fi
\paragraph{Precision and recall trade-off}
%Recall is not everything, as there is also the question of precision.
While recall is an important measure of the algorithms success, when more than $k$ items are reported as suspected top-$k$, the precision is compromised.
Returning all the items yields the maximum recall, but also poor precision when many items are monitored. The ideal behavior of a top-$k$ algorithm is 100\% precision and 100\% recall (the top right position in the graphs).
Figure~\ref{top512tradeoff} illustrates the precision and recall trade-off.
For each recall level, we measure how many elements were returned to achieve it, and compute the corresponding precision. %the maximum precision that satisfies it. %Note that the only variable is, how many entries are returned.

Figures~\ref{top512tradeoffCaida} shows our results for the CAIDA workload.
\PSS{} and 16W-\PSS{} perform the best on this workload as they can provide 80\% recall with near 100\% precision, or $\approx 90\%$ recall with $\approx 90\%$ precision.
At the same time, their maximum recall is close to 100\%, but returning all items drops precision to 50\%.
SS and FR perform worse as their maximum recall is 60\% and they can only ensure 30\% recall with high precision.

Figure~\ref{top512tradeoffUCLA} shows results for the UCLA workload.
As can be seen, \PSS{} and 16W-\PSS{} offer the best precision and recall trade-off, although in this case SS and FR also perform well.

Figure~\ref{top512tradeoffYouTube} shows results for the YouTube workload.
As can be observed, this trace is significantly more difficult.
\PSS{} and 16W-\PSS{} perform better than the rest.
Their maximal recall is slightly over 60\% but over 50\% recall is possible with very high precision.
SS and FR achieve poor recall and precision.

We now look into what happens in synthetic traces.
For Zipf 0.6 distribution, Figure~\ref{top512tradeoffZipf06} shows that \PSS{} and 16W-\PSS{} can achieve over 50\% recall with good accuracy, while SS and FR achieve less than 10\% recall.
In the slightly more skewed Zipf 0.8, Figure~\ref{top512tradeoffZipf08} shows that \PSS{} and 16W \PSS{} perform very well
%are already near optimal
.
They can offer $\approx 90\%$ recall and precision.
SS and FR improve slightly, both offer bad accuracy but the maximum recall of SS is 30\% and FR is slightly less than 20\%.
The non-monotone rise in the SS curve is explained by coincidentally having higher rates of top-$512$ elements in the lower estimated frequency counters.
For Zipf 1.0, 1.2 and 1.5, Figure~\ref{top512tradeoffZipf10}, Figure~\ref{top512tradeoffZipf12} and Figure~\ref{top512tradeoffZipf15} show that while \PSS{} and 16W-\PSS{} provides near optimal precision and recall, SS and FR gradually improve as the skew increases.
They achieve $\approx 50\%$ recall at Zipf 1.0, $\approx 70\%$ recall at Zipf 1.2 and slightly over 80\% recall with high precision with Zipf 1.5.
However, in all these cases, 0.5 \PSS{} is better than FR and SS and thus the space reduction is more than x2 across the entire range.

\section{Conclusion and Discussion}
In this paper, we have presented \pss{} (\PSS{}), a novel algorithm for approximate frequency estimation and top-$k$ identification.
We have also introduced $d$-Way \pss{} ($d$W-\PSS{}), a hardware friendly variant of \PSS{}.
We have extensively evaluated \PSS{} and $d$W-\PSS{} for both problems under two packet traces and a YouTube video trace as well as multiple synthetic Zipf traces.
These experiments exhibited significant reductions in the memory requirements of \PSS{} and $d$W-\PSS{} compared to state of the art alternatives for obtaining the same error.
In top-$k$, we showed that our algorithms achieve superior precision/recall than the alternatives in any tested situation.
In the case of frequency estimation, the only exception is the highly skewed UCLA trace~\cite{UCLA}, and even there, it is only when all schemes are allocated a very large number of counters compared to the trace.
Notice that for this case, all algorithms are precise since with many counters on such a skewed trace, the problem becomes almost trivial.
In contrast, \PSS{} and $d$W-\PSS{} are the only algorithms that performed well on heavy-tailed distributions, which are common in Internet~services.

Another benefit of \PSS{} and $d$W-\PSS{} is that they incur fewer updates to memory since they do not replace a counter with each untracked item.
This is especially true in heavy-tailed workloads.
We have not included the evaluation and quantification of this property for lack of space.

Since $d$W-\PSS{} can be implemented as a simple cache policy, in the future we would like to integrate it into real networking devices.
Interestingly, we believe that $d$W-\PSS{} may also offer benefits for software implementations.
For example, it can probably be parallelized efficiently since each operation only computes the minimum over a small counter~set.

%Since the counters are stored in SRAM, these techniques are considered online counter arrays.

%. Compared to Brick, these techniques have a

%and can be read simply by applying the estimation function, they implement online counter arrays, at the expense of counter precision.
\vspace{-0.15cm}
{ %\bibliographystyle{unsrt}
	\bibliographystyle{IEEEtran}
	\bibliography{refs}

% Generated by IEEEtran.bst, version: 1.14 (2015/08/26)
\begin{thebibliography}{10}
\providecommand{\url}[1]{#1}
\csname url@samestyle\endcsname
\providecommand{\newblock}{\relax}
\providecommand{\bibinfo}[2]{#2}
\providecommand{\BIBentrySTDinterwordspacing}{\spaceskip=0pt\relax}
\providecommand{\BIBentryALTinterwordstretchfactor}{4}
\providecommand{\BIBentryALTinterwordspacing}{\spaceskip=\fontdimen2\font plus
\BIBentryALTinterwordstretchfactor\fontdimen3\font minus
  \fontdimen4\font\relax}
\providecommand{\BIBforeignlanguage}[2]{{%
\expandafter\ifx\csname l@#1\endcsname\relax
\typeout{** WARNING: IEEEtran.bst: No hyphenation pattern has been}%
\typeout{** loaded for the language `#1'. Using the pattern for}%
\typeout{** the default language instead.}%
\else
\language=\csname l@#1\endcsname
\fi
#2}}
\providecommand{\BIBdecl}{\relax}
\BIBdecl

\bibitem{ApproximateFairness}
A.~Kabbani, M.~Alizadeh, M.~Yasuda, R.~Pan, and B.~Prabhakar, ``Af-qcn:
  Approximate fairness with quantized congestion notification for
  multi-tenanted data centers,'' in \emph{IEEE HOTI}, 2010, pp. 58--65.

\bibitem{IntrusionDetection}
B.~Mukherjee, L.~Heberlein, and K.~Levitt, ``Network intrusion detection,''
  \emph{Network, IEEE}, vol.~8, no.~3, pp. 26--41, 1994.

\bibitem{TrafficEngeneering}
T.~Benson, A.~Anand, A.~Akella, and M.~Zhang, ``Microte: Fine grained traffic
  engineering for data centers,'' in \emph{ACM CoNEXT}, 2011, p.~8.

\bibitem{IntrusionDetection2}
P.~Garcia-Teodoro, J.~E. Díaz-Verdejo, G.~Maciá-Fernández, and E.~Vázquez,
  ``Anomaly-based network intrusion detection: Techniques, systems and
  challenges,'' \emph{Computers and Security}, pp. 18--28, 2009.

\bibitem{LoadBalancing}
G.~Dittmann and A.~Herkersdorf, ``Network processor load balancing for
  high-speed links,'' in \emph{SPECTS}, vol. 735, 2002.

\bibitem{TinyLFU}
G.~Einziger and R.~Friedman, ``Tinylfu: A highly efficient cache admission
  policy,'' in \emph{Euromicro {PDP}}, 2014, pp. 146--153.

\bibitem{CounterArray1}
D.~Shah, S.~Iyer, B.~Prabhakar, and N.~McKeown, ``Maintaining statistics
  counters in router line cards.'' \emph{IEEE Micro}, pp. 76--81, 2002.

\bibitem{CounterArray2}
S.~Ramabhadran and G.~Varghese, ``{Efficient implementation of a statistics
  counter architecture},'' \emph{ACM SIGMETRICS}, pp. 261--271, 2003.

\bibitem{ICE-Buckets}
G.~Einziger, B.~Fellman, and Y.~Kassner, ``Independent counter estimation
  buckets,'' in \emph{IEEE INFOCOM}, April 2015, pp. 2560--2568.

\bibitem{CEDAR}
E.~Tsidon, I.~Hanniel, and I.~Keslassy, ``Estimators also need shared values to
  grow together,'' in \emph{IEEE INFOCOM}, 2012, pp. 1889--1897.

\bibitem{DISCO}
C.~Hu, B.~Liu, H.~Zhao, K.~Chen, Y.~Chen, C.~Wu, and Y.~Cheng, ``Disco: Memory
  efficient and accurate flow statistics for network measurement,'' in
  \emph{IEEE ICDCS}, 2010, pp. 665--674.

\bibitem{CounterBraids}
Y.~Lu, A.~Montanari, B.~Prabhakar, S.~Dharmapurikar, and A.~Kabbani, ``Counter
  braids: a novel counter architecture for per-flow measurement,'' in \emph{ACM
  SIGMETRICS}, 2008, pp. 1799--1807.

\bibitem{CUSketch}
C.~Estan and G.~Varghese, ``New directions in traffic measurement and
  accounting,'' \emph{ACM SIGCOMM}, 2002.

\bibitem{CountMinSketch}
G.~Cormode and S.~Muthukrishnan, ``An improved data stream summary: The
  count-min sketch and its applications,'' \emph{J. Algorithms}, vol.~55, pp.
  29--38, 2004.

\bibitem{SpaceSavingIsTheBest}
G.~Cormode and M.~Hadjieleftheriou, ``Finding frequent items in data streams,''
  \emph{VLDB}, vol.~1, no.~2, pp. 1530--1541, Aug. 2008.

\bibitem{SpaceSavingIsTheBest2010}
------, ``Methods for finding frequent items in data streams,'' \emph{J. VLDB},
  vol.~19, no.~1, pp. 3--20, 2010.

\bibitem{LC}
G.~S. Manku and R.~Motwani, ``Approximate frequency counts over data streams,''
  in \emph{Proc. of the Int. Conf. on V.L. Data Bases}, ser. VLDB, 2002.

\bibitem{BatchDecrement}
R.~M. Karp, S.~Shenker, and C.~H. Papadimitriou, ``A simple algorithm for
  finding frequent elements in streams and bags,'' \emph{ACM Trans. Database
  Syst.}, vol.~28, no.~1, Mar. 2003.

\bibitem{SpaceSavings}
A.~Metwally, D.~Agrawal, and A.~E. Abbadi, ``Efficient computation of frequent
  and top-k elements in data streams,'' in \emph{IN ICDT}, 2005.

\bibitem{SpaceSavingIsTheBest2011}
N.~Manerikar and T.~Palpanas, ``Frequent items in streaming data: An
  experimental evaluation of the state-of-the-art,'' \emph{Data Knowl. Eng.},
  pp. 415--430, 2009.

\bibitem{CAIDA}
``The caida ucsd anonymized internet traces 2015 - february 19th.''

\bibitem{UCLA}
``Unpublished, see http://www.lasr.cs.ucla.edu/ddos/traces/.''

\bibitem{youtube}
M.~Zink, K.~Suh, Y.~Gu, and J.~Kurose, ``Watch global, cache local: Youtube
  network traffic at a campus network: measurements and implications,'' in
  \emph{Electronic Imaging 2008}.\hskip 1em plus 0.5em minus 0.4em\relax
  International Society for Optics and Photonics, 2008, pp. 681\,805--681\,805.

\bibitem{PLC}
X.~Dimitropoulos, P.~Hurley, and A.~Kind, ``Probabilistic lossy counting: An
  efficient algorithm for finding heavy hitters,'' \emph{ACM SIGCOMM}, vol.~38,
  no.~1, Jan. 2008.

\bibitem{frequent4}
E.~D. Demaine, A.~L\'{o}pez-Ortiz, and J.~I. Munro, ``Frequency estimation of
  internet packet streams with limited space,'' in \emph{EATCS ESA}, 2002, pp.
  348--360.

\bibitem{MultiStageFilters}
C.~Estan and G.~Varghese, ``New directions in traffic measurement and
  accounting,'' \emph{ACM SIGCOMM}, vol.~32, no.~4, pp. 323--336, Aug. 2002.

\bibitem{CountSketch}
M.~Charikar, K.~Chen, and M.~Farach-Colton, ``Finding frequent items in data
  streams,'' in \emph{EATCS ICALP}, 2002, pp. 693--703.

\bibitem{counterTree}
M.~Chen and S.~Chen, ``Counter tree: {A} scalable counter architecture for
  per-flow traffic measurement,'' in \emph{IEEE ICNP}, 2015, pp. 111--122.

\bibitem{RandomizedCounterSharing}
T.~Li, S.~Chen, and Y.~Ling, ``Per-flow traffic measurement through randomized
  counter sharing,'' \emph{Networking, IEEE/ACM Trans. on}, 2012.

\bibitem{SpectralBloom}
S.~Cohen and Y.~Matias, ``Spectral bloom filters,'' in \emph{ACM SIGMOD}, 2003,
  pp. 241--252.

\bibitem{ML-CBF}
D.~Ficara, A.~D. Pietro, S.~Giordano, G.~Procissi, and F.~Vitucci, ``Enhancing
  counting bloom filters through huffman-coded multilayer structures.''
  \emph{IEEE/ACM Trans. Netw.}, vol.~18, pp. 1977--1987, 2010.

\bibitem{HeavyHitters}
R.~Ben-Basat, G.~Einziger, R.~Friedman, and Y.~Kassner, ``Heavy hitters in
  streams and sliding windows,'' in \emph{IEEE INFOCOM}, 2016.

\bibitem{TinyTable}
G.~Einziger and R.~Friedman, ``Counting with tinytable: Every bit counts!'' in
  \emph{ACM {ICDCN}}, 2016.

\bibitem{Brick}
N.~Hua, B.~Lin, J.~J. Xu, and H.~C. Zhao, ``Brick: A novel exact active
  statistics counter architecture,'' in \emph{ACM/IEEE ANCS}, 2008, pp. 89--98.

\bibitem{ApproximateCounting}
R.~Morris, ``Counting large numbers of events in small registers,''
  \emph{Commun. ACM}, vol.~21, no.~10, pp. 840--842, 1978.

\bibitem{SAC}
R.~Stanojevic, ``Small active counters,'' in \emph{IEEE INFOCOM}, 2007, pp.
  2153--2161.

\bibitem{Sample1}
B.-Y. Choi, J.~Park, and Z.-L. Zhang, ``Adaptive random sampling for load
  change detection,'' \emph{ACM SIGMETRICS}, pp. 272--273, 2002.

\bibitem{BetterNetflow}
C.~Estan, K.~Keys, D.~Moore, and G.~Varghese, ``Building a better netflow,'' in
  \emph{ACM SIGCOMM}, 2004.

\bibitem{caffein}
B.~Manes, ``Caffeine: A high performance caching library for java 8,''
  \emph{https://github.com/ben-manes/caffeine}.

\end{thebibliography}
}
\ifdefined\EXTENDED
\newpage\clearpage
\appendix
\vfill
{
	\let\section\subsection
	\let\subsection\subsubsection
\section{Theoretical Guarantees for Top-$k$ Identification Using \pss{}}\label{sec:topk\PSS{}}
To show the benefit of probabilistic admission, we describe a variant of \PSS{} (see Section~\ref{sec:\PSS{}}) that aims to minimize the number of counters needed for top-$k$ identification.
We consider a setting in which the stream elements are i.i.d, i.e., each element is sampled independently and according to the same distribution.
Also, we assume that elements are from a finite domain $\mathcal U=\set{1,2,\ldots, D}$, and without loss of generality, the frequencies of the elements are 
$$f_1\ge f_2\ge \ldots \ge f_k>f_{k+1}\ge f_{k+2}\ge\ldots\ge f_D.\footnote{If we do not wish to assume that $f_k$ is strictly larger than $f_{k+1}$, we can replace our demand to finding a set of $k$ items such that all of their frequencies are at least $f_k$. Such a model leads to similar results.}$$
For $r\in\{1,2,\dots,D\}$, we denote $F_r\triangleq\sum_{i=1}^r f_i$. 
This means that at each timestamp, element $i$ will arrive with probability $f_i$ and $\sum_{i=1}^{D}f_i = 1$.
We note that this i.i.d. setting may not be applicable to certain streams exhibiting high time locality, such as packets going through a home router, but may resemble the traffic patterns appearing on major backbone routers.
The goal of the top-$k$ problem is then to identify the set of the most frequent elements $\set{1,2,\ldots, k}$ with as few counters as possible.
Our assumption is that the stream may be arbitrarily long, but we wish to guarantee that with probability $1$ the algorithm will eventually identify all top-$k$ items.

Formally, we say that algorithm $A$ has successfully identified the top-$k$ elements at time $t$ if after the arrival of the $t$'th element, the $k$ largest counters are allocated for items $\set{1,2,\ldots, k}$. Since the actual items are random variables, we consider the probability, denoted $P_{m,k}^A(t)$, that algorithm $A$ will successfully identify the top-$k$ elements at time $t$ if allocated with $m$ counters.
Finally, our benchmark would be the minimal number of counters that $A$ requires to achieve 
$$\lim_{t\to\infty} P_{m,k}^A(t) = 1.$$
We say that an item is a \emph{tail item} if it is not amongst the top-$k$, i.e., is one of $\set{k+1,\ldots,D}$. 
We call the largest $k$ counters the \emph{main counters}, and the remaining ones \emph{tail counters}.
Our goal is then to ensure that after seeing infinitely many elements, the top-$k$ elements will be guaranteed to be allocated with the main counters.
For example, notice that $D$ counters are enough for Space Saving, regardless of the actual frequencies $\set{f_i}$. However, the interesting case is when $m\ll D$.

We start by analyzing the number of counters Space Saving requires for this task. We compare to Space Saving, because it achieves asymptotic improvement over previous work~\cite{SpaceSavings} and, to the best of our knowledge, it is considered the state of the art.
Notice that our analysis is somewhat different than the one presented in \cite{SpaceSavings}, as it further assumes that the stream is i.i.d., which reduces the number of counters required.
\subsection{Conditions for Successful Space Saving Top-$k$ Identification}
\begin{algorithm}[t]
	\scriptsize
	
	Initialization: $C=\emptyset$
	\begin{algorithmic}[1]
		\Function{add}{Item $x$}
		\If {$x\in C$}
		\State Increment $c_x$
		\Else\If {$|C| < M$}
		\State $c_x = 1$
		\State $C = C \cup \{x\}$
		\Else
		\State $m = \text{argmin}_{y\in C} c_y$ \label{line:minFinding}
		\State $C = C \setminus \{m\} \cup \{x\}$
		\State $c_x\gets c_m$
		\State Increment $c_x$ \label{line:weightUpdate}
		\EndIf\EndIf
		\EndFunction
	\end{algorithmic}
	\caption{Space-Saving}
	\label{alg:SpaceSaving}
\end{algorithm}
\normalsize
Assume that we allocate a Space Saving instance with $m$ counters, and would like to identity the top-$k$ elements.
In the algorithm, whose pseudo code appears in Algorithm~\ref{alg:SpaceSaving}, every arriving element is associated with a counter; if a counter was associated with the element prior to its arrival, the counter is incremented by 1; otherwise, the element ``takes over'' the minimal counter and increments it.
Consequently, we are guaranteed that the top-$k$ element will be associated with a counter upon arrival (unlike our algorithm described below).
The only problem for Space Saving arises when a top-$k$ element loses its main counter in favor of a ``tail item'', i.e., an item that is not within the top-$k$. Since $k$ is the least-frequent element within the top-$k$, it is enough to consider whether it is guaranteed to have a counter.
The key point of our analysis is observing that \emph{if} all top-$k$ are allocated with counters, they have a positive probability of forever being allocated with these counters if and only if the tail counters increase in a rate smaller than $f_k$. Observe that if the top-$k$ elements reside within the main counters, the expected increment rate of the \emph{sum} of tail counters is $\sum_{i=k+1}^{D}f_i = 1-F_k$. This means that the average increment rate for a tail counter is $\frac{1-F_k}{m-k}$, and therefore the increment rate for the minimum among all tail counter is at most $\frac{1-F_k}{m-k}$. On the other hand, the counter allocated for item $k$ increases with a rate of $f_k$. Next, we claim that having $f_k>\frac{1-F_k}{m-k}$ is a necessary and sufficient condition for Space Saving to identify the top-$k$. Illustration of the claim appears in Figure~\ref{fig:SSTopKCounters}.
\begin{figure}[H]
	%\medskip
	\centering
	%\frame{
	\includegraphics[width=\linewidth]{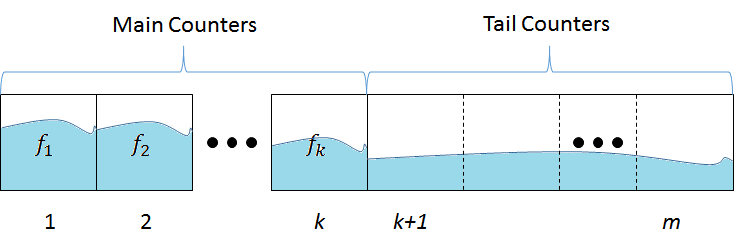}
	%}
	\caption{The counters allocated for the top-$k$ elements are increased with rates $f_1,\ldots,f_k$. The increase rate of the \emph{minimal} tail counter is at most $\frac{1-F_k}{m-k}$. }
	\label{fig:SSTopKCounters}
\end{figure}
\begin{theorem}
	Space Saving successfully identifies the top-$k$ using $m$ counters if and only if $f_k>\frac{1-F_k}{m-k}$.
\end{theorem}
\begin{proof}
	We start by noting that with probability $1$, if a top-$k$ element is not allocated with a main counter, it will get such a counter in the future. This happens because we assumed $f_k>f_{k+1}$, which means that after a tail element takes over a main counter, the counter is incremented with a frequency of at most $f_{k+1}$, while the counter allocated with the missing top-$k$ element (or the smallest counter if it is not currently allocated with one) increases at a rate of at least $f_k>f_{k+1}$.
	
	We are left with showing that $f_k>\frac{1-F_k}{m-k}$ is a characterization of the ability of the top-$k$ elements to seize the main counters without being evicted.
	As mentioned above, it is enough to argue that item $k$ will be eventually allocated within the main counters to ensure all top-$k$ items are successfully identified.
	Assume that $k$ is allocated with a counter with value $c_k$, and let $c_t$ be the value of the minimal tail counter. Notice that $c_k$ is increased with rate $f_k$ while $c_t$'s rate is at most $\frac{1-F_k}{m-k}$. 
	Consider the infinite Markov chain whose states represent the difference between $c_k$ and $c_t$, i.e., state $i$ represents the case of $c_k-c_t=i$. At any time an element arrives, it increments $c_k$ with probability $f_k$, increments $c_t$ with probability of at most $\frac{1-F_k}{m-k}$, and increments other counters otherwise.
	Therefore if we ignore other counters, the transition probabilities do not depend on the current state $i$ and can be expressed as $\forall i:\Pr[i+1\mid i] = \lambda \triangleq \frac{f_k}{\frac{1-F_k}{m-k}+f_k}$ and $\forall i:\Pr[i-1\mid i] = 1-\Pr[i+1\mid i]=\mu \triangleq \frac{\frac{1-F_k}{m-k}}{\frac{1-F_k}{m-k}+f_k}$.
	Notice that $$\Pr[i+1\mid i]>1/2 \iff f_k>\frac{1-F_k}{m-k}.$$ The stochastic process is illustrated in Figure~\ref{fig:MM1}.
	\begin{figure}[H]
		%\medskip
		\centering
		%\frame{
		\includegraphics[width=\linewidth]{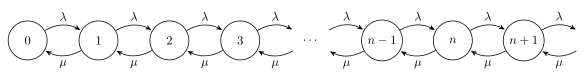}
		%}
		\caption{ The probability of moving to a larger index state is $\lambda$ and is not dependent on the process made so far or the current state index. }
		\label{fig:MM1}
	\end{figure}
	Since we know that top-$k$ elements will obtain a main counter infinitely often, the question of successful identification narrows to the question ``is there a positive probability that given a positive integer $n>1$, such that if $c_k\ge c_t+n$, the process will \emph{never} return to state $0$?'' (as then the main counter allocated for $k$ may become minimal and lead to $k$ being evicted). It is a known fact that a 1D random walk over the non-negative integers that goes left with probability $\mu$ and right with probability $\lambda$ will return to $0$ starting at state $n$ is $\parentheses{\frac{\mu}{\lambda}}^n$, which is strictly smaller than $1$ for $\lambda>\mu$.  As we are guaranteed to reach a positive difference infinitely often, when $\mu<1/2$ we will eventually guarantee that item $k$ will not be evicted if and only if \begin{align}
	f_k>\frac{1-F_k}{m-k}.\label{eq:SStopK}
	\end{align}
	This can further be expressed as a bound on the number of counters required by Space Saving:
	\begin{align}
	m > k + \frac{1-F_k}{f_k}.\label{eq:SSFinalCount}
	\end{align}
\end{proof}
\subsection{Conditions for Successful \PSS{}' Top-$k$ Identification}
In this subsection, we present a variant of the \pss{} algorithm (see Algorithm~\ref{alg:\PSS{}}), called \PSS{}', that aims to minimize the number of counters required for top-$k$ identification.
\PSS{}' takes advantage of the fact that we can ``slow'' the frequency in which the smallest counter is incremented for achieving better identification results.
Formally, \PSS{}' acts similarly to \PSS{}, except that upon arrival of a non-monitored element, the probability in which it will be allocated with the minimal counter is a constant $P$, and do not depend on the value of the minimal counter.\footnote{If there exists more than a single minimum-value counter, then the arriving element is admitted with probability $1$.}
%\begin{itemize}
%\item All counters whose value equals the smallest counter are considered not allocated (i.e., there is always at least one non-allocated counter).
%\item Whenever a non-monitored (that does not have an allocated counter) item arrives, if there are two (or more) non-allocated counters, it takes over one of them and increments its value. If there is only a single non allocated counter (i.e., the remaining counters have larger value), the counter is incremented with probability $P$, which will be determined later.
%\end{itemize}

Notice that these probabilistic allocations has both positive and negative effects on the possibility of successful identification. On the positive side, since the minimal counter is incremented slower than in Space Saving, we can relax the $f_k>\frac{1-F_k}{m}$ constraint a bit. Alas, \PSS{}' is not guaranteed that a top-$k$ element will obtain a counter infinitely often. This means that for \PSS{}' to successfully identify the top-$k$ we need to impose two constraints:
\begin{enumerate}
	\item A top-$k$ item which is not currently allocated with a counter will get one (w.p. $1$).\\
	Notice that if a top-$k$ element is not allocated with a counter (recall that now we have $m-1$ allocated counters), then the minimal, non-allocated counter increases with rate of at least $P\cdotpa{f_k + 1-F_m}$. In contrast, the slowest allocated counter rate cannot exceed $f_m$. Thus, the imposed constraint is
	\begin{align}
	f_m < P\cdotpa{f_k + 1-F_m}\label{eq:infOften}.
	\end{align}
	This gives us a lower bound on the value of $P$:
	\begin{align}
	P > \frac{f_m}{f_k + 1-F_m}\label{eq:PssAdmitProbOpt}.
	\end{align}
	For the sake of simplifying the calculations, we impose a stronger restriction on the admission probability and consider $P$ to satisfy
	\begin{align}
	P > \frac{f_m}{f_k}\label{eq:PssAdmitProb}.
	\end{align}
	\item When all top-$k$ elements are allocated with counters, there exists a positive probability that they will never be evicted from the table.
	This is the front in which \PSS{}' has advantage over Space Saving. Since the minimal counter is only incremented with sampling rate $P$, the rate in which the smallest counter within the tail increases is at most 
	$$\frac{\sum_{i=k+1}^{m}f_i + P\cdot \sum_{i=m+1}^{D}f_i}{m-k} = \frac{F_{m}-F_k + P\cdot (1-F_m)}{m-k}.$$
	Meanwhile, the counter associated with item $k$ is incremented in rate of $f_k$. This means that our constraint is:
	\begin{align}
	f_k > \frac{F_{m}-F_k + P\cdot (1-F_m)}{m-k}\label{eq:keepCounter}
	\end{align}
\end{enumerate}
Notice that for $P=1$, Inequality~\eqref{eq:PssAdmitProb} trivially holds (and thus Space Saving is guaranteed to get counters allocated for top-$k$ elements infinitely often), while Inequality~\ref{eq:keepCounter} degenerates into Inequality~\eqref{eq:SStopK}.
In the following subsection, we analyze how a smart choice of the increment probability $P$ reduces the required number of counters.

%Next, we express \eqref{eq:keepCounter} in terms of a lower bound on $m$:
%\begin{align}
%m > k + \frac{F_{m}-F_k + P\cdot (1-F_m)}{f_k}\label{eq:\PSS{}numCounters}
%\end{align}
%Plugging \eqref{eq:PssAdmitProbOpt} into \eqref{eq:\PSS{}numCounters}, we get a bound on how many counters \PSS{}' requires assuming optimal selection of $P$, i.e. $P\triangleq \frac{f_m}{f_k + 1-F_m}$:
%\begin{align}
%m > k + \frac{F_{m}-F_k + \frac{f_m}{f_k + 1-F_m}\cdot (1-F_m)}{f_k}\label{eq:\PSS{}FinalCountOpt}
%\end{align}
%
%A more simplified (and weaker) lower bound is obtained by plugging \eqref{eq:PssAdmitProb} into \eqref{eq:keepCounter}:
%\begin{align}
%m > k + \frac{F_{m}-F_k + \frac{f_m}{f_k }\cdot (1-F_m)}{f_k}\label{eq:\PSS{}FinalCount}
%\end{align}
\subsection{Performance Comparison for Zipf Distributed Streams}
In many of the previous works, Zipf distributed streams served as a popular benchmark for algorithms comparison due to its nice mathematical properties.
Here, we continue this line and compare the performance of Space Saving and \PSS{}' on i.i.d Zipf distributed streams with varying skews.
We start with a formal definition of a Zipf stream.
\begin{definition}
	Denote $\Gamma_\alpha(D) = \sum_{i=1}^{D}i^{-\alpha}$.
	A stream will be called an i.i.d Zipf stream with skew $\alpha$ over domain $D$ if all of its elements are sampled independently and follow the distribution in which item $i\in\set{1,2,\ldots,D}$ appears with probability $f_i = \frac{i^{-\alpha}}{\Gamma_\alpha(D)}$.
\end{definition}
Traditionally, streams with skew $\alpha\in(0,1]$ are called ``mildly skewed'' or ``heavy tailed'', while larger skews grants the streams the title ``highly skewed''. Streams in which every item is selected with uniform probability (skew=0) are called ``uniform''.
These names differentiate highly skewed streams, in which a small number of elements consists most of the stream, and heavy tailed ones, where most of the arriving elements are tail items. This property is also observable from the behavior of $\Gamma_\alpha(D)$; for $\alpha>1$, $\Gamma_\alpha(D)$ convergence to a constant as $D$ grows (e.g., $\Gamma_2(\infty)\approx 1.645$); for $\alpha = 1$, $\Gamma_1(D)\approx\ln(1.78D)$; lastly, for $\alpha<1$, we have $\Gamma_\alpha(D)=\frac{D^{1-\alpha}}{1-\alpha} + O(1)$. 

For heavily skewed streams, Space Saving is known to be optimal~\cite{SpaceSavings}. For more mildly skewed streams, we show that \PSS{}' could asymptotically improve the number of counters required for identifying the top-$k$ elements. This also provides theoretical grounds to the poor empirical performance of Space Saving when evaluated on heavy tailed workloads in Section~\ref{sec:results}.

Assuming a Zipf distributed stream, 
the condition for Space Saving to converge, as appears in \eqref{eq:SSFinalCount}, then becomes:
%\begin{align}
%m &> k + \frac{1-F_k}{f_k} \notag\\
%  &= k + \frac{1-\frac{\Gamma_\alpha(k)}{\Gamma_\alpha(D)}}{\frac{k^{-\alpha}}{\Gamma_\alpha(D)}}\notag\\
%  &= k + k^\alpha\parentheses{{\Gamma_\alpha(D)}-{\Gamma_\alpha(k)}}
%%  &\approx k + \frac{\Gamma_\alpha(D)-\Gamma_\alpha(k)}{k^{-\alpha}} 
%\end{align}
\begin{align}
m &> k + \frac{1-F_k}{f_k} = k + \frac{1-\frac{\Gamma_\alpha(k)}{\Gamma_\alpha(D)}}{\frac{k^{-\alpha}}{\Gamma_\alpha(D)}}\notag\\
\implies m &> k + k^\alpha\parentheses{{\Gamma_\alpha(D)}-{\Gamma_\alpha(k)}}\label{eq:SSZipfReq}
%  &\approx k + \frac{\Gamma_\alpha(D)-\Gamma_\alpha(k)}{k^{-\alpha}} 
\end{align}

%\begin{align*}
%\frac{k^{-\alpha}}{\Gamma_\alpha(D)} &> \frac{1-F_k}{m-k}\\
%\frac{k^{-\alpha}}{\Gamma_\alpha(D)} &> \frac{1-\frac{\Gamma_\alpha(k)}{\Gamma_\alpha(D)}}{m-k}\\
%k^{-\alpha} &> \frac{{\Gamma_\alpha(D)}-{\Gamma_\alpha(k)}}{m-k},
%\end{align*}
%which is equivalent to the requirement
%\begin{align}
%m &> k+k^\alpha\parentheses{{\Gamma_\alpha(D)}-{\Gamma_\alpha(k)}}
%\end{align}

%m &> k+\frac{D^{1-\alpha} - k^{1-\alpha}}{(1-\alpha)k^{\alpha}}\\

%In contrast, if we assume Zipf distribution, \PSS{}' will require (see~\eqref{eq:\PSS{}FinalCount})
%\begin{align}
%m_{\PSS{}'} > k + \frac{F_{m}-F_k + \frac{f_m}{f_k }\cdot (1-F_m)}{f_k} 
%= k + k^{\alpha}\cdotpa{\Gamma_\alpha(m)-\Gamma_\alpha(k) + \parentheses{\frac{m}{k}}^{-\alpha}\cdot (\Gamma_\alpha(D)-\Gamma_\alpha(m))}\label{eq:\PSS{}Zipf}
%%m_{\PSS{}'} 	&> k + \frac{F_{m}-F_k + \frac{f_m}{f_k + 1-F_m}\cdot (1-F_m)}{f_k}\notag\\
%%	&= k + \frac{\frac{\Gamma_\alpha(m)}{\Gamma_\alpha(D)}-\frac{\Gamma_\alpha(k)}{\Gamma_\alpha(D)} + \frac{m^{-\alpha}}{k^{-\alpha} + 1-\frac{\Gamma_\alpha(m)}{\Gamma_\alpha(D)}}\cdot (1-\frac{\Gamma_\alpha(m)}{\Gamma_\alpha(D)})}{k^{-\alpha}}\notag\\
%%\implies m_{\PSS{}'}	&> k + \frac{\Gamma_\alpha(m)-\Gamma_\alpha(k) + \frac{m^{-\alpha}\Gamma_\alpha(D)\cdot (\Gamma_\alpha(D)-\Gamma_\alpha(m))}{\Gamma_\alpha(D)(k^{-\alpha} + 1)-\Gamma_\alpha(m)}}{\Gamma_\alpha(D)k^{-\alpha}}\label{eq:\PSS{}Zipf}
%\end{align}

For analyzing \PSS{}'s performance, we will select the value of $P$ based on the skewness of the stream, as discussed below.
When plugging the Zipf distribution into the first \PSS{}' constraint (see \eqref{eq:PssAdmitProb}),  we get
%and \eqref{keepCounter}) we get the constraint set:
\begin{align}
P > \frac{f_m}{f_k} = \parentheses{\frac{m}{k}}^{-\alpha}\notag\\
\implies m > k\cdot P^{-\frac{1}{\alpha}}.\label{eq:\PSS{}1}
\end{align}
Similarly, the second constraint (see \eqref{eq:keepCounter}) is now:
\begin{align}
f_k 	>& \frac{F_{m}-F_k + P\cdot (1-F_m)}{m-k}\notag\\
\iff m 	>& k + k^{\alpha}\bigg(\Gamma_\alpha(m)-\Gamma_\alpha(k)\notag\\
&+ P\cdot (\Gamma_\alpha(D)-\Gamma_\alpha(m))\bigg)\notag
\end{align}
In order to simplify the right hand side of the inequality, we impose a stronger bound on $m$ and require is to satisfy:
\begin{align}
m 	&> k + k^{\alpha}\cdotpa{\Gamma_\alpha(m)-\Gamma_\alpha(k) + P\cdot \Gamma_\alpha(D)}\label{eq:\PSS{}2}
\end{align}
\subsubsection{Heavy Tailed Streams}
In this section, we assume that $\alpha\in(0,1)$
%\footnote{The requirement $\alpha>0.1$ can be replaced with $\alpha>\eps$ for some fixed $eps>0$. In reallity, the problem of finding top-$k$ in skews smaller than $0.1$ has very little importance, as all elements arrive in similar frequencies and the distribution is almost uniform.}
is fixed and that $k= o(D^{\frac{\alpha}{1+\alpha}})$
and analyze the number of counters required for Space Saving and \PSS{}' for successfully identifying the top-$k$ items.\footnote{In practice, values of $k$ are typically very small and may be considered sub-polynomial in $D$.}
We start by using the explicit formula of $\Gamma_\alpha(\cdot)$ for \eqref{eq:SSZipfReq}:
\begin{align*}
m_{SS} &= k + k^\alpha\parentheses{{\Gamma_\alpha(D)}-{\Gamma_\alpha(k)}}\\
&= k+k^\alpha\frac{D^{1-\alpha} - k^{1-\alpha}+\Theta(1)}{1-\alpha} \\
&= 
\frac{k^\alpha\cdot D^{1-\alpha} - \alpha\cdot k+\Theta(k^{\alpha})}{1-\alpha} = \Theta( D^{1-\alpha})
\end{align*}
Thus, we established that the number of counters required for Space Saving is $m_{SS}=\Omega( D^{1-\alpha})$.

For \PSS{}', we choose the admission probability to be 
\begin{align}
%P\triangleq k^{\alpha-\alpha^2}\cdot D^{\frac{\alpha^2-\alpha}{1+\alpha}}=\parentheses{\frac{D^{\frac{1}{1+\alpha}}}{k}}^{\alpha^2-\alpha}. \label{eq:p\PSS{}}
P\triangleq D^{\frac{\alpha^2-\alpha}{1+\alpha}}\label{eq:p\PSS{}}.
\end{align} 
Notice that 
%since $k \le D^{\frac{1}{1+\alpha}}$, 
$P\in(0,1)$ is a valid probability.
Next, we will show that using 
\begin{align}
m_{\PSS{}'}\triangleq %\frac{2}{1-\alpha}\cdot 
%k
%^\alpha
c\cdot k\cdot D^{\frac{1-\alpha}{1+\alpha}} \label{eq:m\PSS{}}
\end{align}
counters, where $c$ is a (large enough) constant, we can satisfy both \eqref{eq:\PSS{}1} and \eqref{eq:\PSS{}2}, thus $m_{\PSS{}'}$ counters are enough for successful identification of the top-$k$ elements.

Constraint~\eqref{eq:\PSS{}1} requires that $m > k\cdot P^{-\frac{1}{\alpha}}$. Plugging in \eqref{eq:p\PSS{}} and \eqref{eq:m\PSS{}}, we get:
%This means that to satisfy constraint~\eqref{eq:PssAdmitProb} we require:
$$ m_{\PSS{}'} = 
%\frac{2}{1-\alpha}\cdot 
%k^\alpha\cdot D^{\frac{1-\alpha}{1+\alpha}}  
% = 
% \frac{2}{1-\alpha}\cdot 
% k\cdot k^{\alpha - 1}\cdot D^{\frac{1-\alpha}{1+\alpha}}
%  = 
%  \frac{2}{1-\alpha}\cdot 
%  k\cdot P^{-\frac{1}{\alpha}}
%   > k\cdot P^{-\frac{1}{\alpha}}
ck^\cdot D^{\frac{1-\alpha}{1+\alpha}}  
%k^\alpha\cdot D^{\frac{1-\alpha}{1+\alpha}}  
> k\cdot P^{-\frac{1}{\alpha}}
,$$
as required.
Next, we show an inequality that will be useful later:
\begin{align}
(m_{\PSS{}'})^{1-\alpha} &= 
%\frac{2}{1-\alpha}
% (k^\alpha\cdot D^{\frac{1-\alpha}{1+\alpha}})^{1-\alpha} 
%k^\alpha\cdot D^{\frac{1}{1+\alpha}}\cdot 
%\frac{2D^{-\frac{\alpha}{1+\alpha}}}{1-\alpha}
%D^{-\frac{\alpha}{1+\alpha}}
(ck\cdot D^{\frac{1-\alpha}{1+\alpha}})^{1-\alpha} \notag\\
&=(D^{\frac{1}{1+\alpha}}\cdot \frac{ck}{D^{\frac{\alpha}{1+\alpha}}})^{1-\alpha}
%< (D^{\frac{1}{1+\alpha}})^{1-\alpha}
%=
< D^{\frac{1-\alpha}{1+\alpha}},
\end{align}
where the last inequality holds for large enough $c$.
%This allows us to conclude that
%-------------------------
%\begin{align}
%m_{\PSS{}'} ^ {1-\alpha} 
%\end{align}
%which gives us the following lower bound on the number of counters:
%\begin{align}
%m > k \cdot D^{\frac{1-\alpha}{1+\alpha}}\label{eq:\PSS{}3}
%\end{align}

Finally, we show that our choice of $P$ and $m_{\PSS{}'}$ also satisfies~\eqref{eq:keepCounter}:%, which becomes
\begin{align}
%	&k + k^{\alpha}\cdotpa{\Gamma_\alpha(m)-\Gamma_\alpha(k) + P\cdot \Gamma_\alpha(D)}\notag\\
%	&= k + \frac{k^{\alpha}\cdotpa{m^{1-\alpha}-k^{1-\alpha} + k^{\alpha^2-\alpha}\cdot D^{\frac{\alpha^2-\alpha}{1+\alpha}}\cdot D^{1-\alpha}+\Theta(1)}}{1-\alpha}\notag\\
%	&=  \frac{k^{\alpha}\cdotpa{m^{1-\alpha} + k^{\alpha^2-\alpha}\cdot D^{\frac{1-\alpha}{1+\alpha}}}-\alpha k + \Theta(k^{\alpha})}{1-\alpha}\notag\\
%	&< \frac{2k^{\alpha}\cdot{k^{\alpha^2-\alpha}\cdot D^{\frac{1-\alpha}{1+\alpha}}}-\alpha k + \Theta(k^{\alpha})}{1-\alpha}\notag \\
%	&= m_{\PSS{}'} \cdot  k^{\alpha^2-\alpha} +O(k) = o(m_{\PSS{}'})
&k + k^{\alpha}\cdotpa{\Gamma_\alpha(m)-\Gamma_\alpha(k) + P\cdot \Gamma_\alpha(D)}\notag\\
&= k + \frac{k^{\alpha}\cdotpa{m^{1-\alpha}-k^{1-\alpha} + D^{\frac{\alpha^2-\alpha}{1+\alpha}}\cdot D^{1-\alpha}+\Theta(1)}}{1-\alpha}\notag\\
&=  \frac{k^{\alpha}\cdotpa{m^{1-\alpha} + D^{\frac{1-\alpha}{1+\alpha}}}-\alpha k + \Theta(k^{\alpha})}{1-\alpha}\notag\\
&< \frac{2k^{\alpha}\cdot{D^{\frac{1-\alpha}{1+\alpha}}}-\alpha k + \Theta(k^{\alpha})}{1-\alpha}\notag \\
&= m_{\PSS{}'} \cdot  \frac{2}{c(1-\alpha)} +O(k) < m_{\PSS{}'}
\end{align}
Since we have shown that our selection of admission probability and number of counters satisfies both constraints, we have proved that \PSS{}' requires only $O(D^{\frac{1-\alpha}{1+\alpha}})$ counters to successfully identify the top-$k$ hitters on a Zipf stream with skew $\alpha$.
%Finally, notice for practical values of $k$, we have $k\ll D^{\frac{\alpha}{1+\alpha}}$, which implies that if we choose the number of counters to be $m\approx k \cdot D^{\frac{1-\alpha}{1+\alpha}}$, we have $m^{1-\alpha} \ll D^{\frac{1-\alpha}{1+\alpha}}$, and thus such choice of $m$
%%$$(k \cdot D^{\frac{1-\alpha}{1+\alpha}})^{1-\alpha} = o(D^{\frac{1-\alpha}{1+\alpha}}),$$
%%the choice of $m=\frac{1.1}{1-\alpha}\cdot(k \cdot D^{\frac{1-\alpha}{1+\alpha}})$ 
%satisfies both \eqref{eq:\PSS{}3} and \eqref{eq:\PSS{}4}.

We conclude that for the problem of identifying top-$k$ over i.i.d. heavy tailed streams, Space Saving requires $\Theta(D^{1-\alpha})$ while \PSS{}' requires $\Theta(D^{\frac{1-\alpha}{1+\alpha}})$. Notice that for values of $\alpha$ that are close to $1$, this is nearly a quadratic space reduction.
For example, consider trying to find the top-$32$ flows on a backbone router whose traffic is approximately Zipf0.8 with domain of $D=2^{64}$ elements. Space Saving requires about $570K$ counters; in contrast, \PSS{}' could allocate roughly $44K$ counters to achieve the same. The admission probability for these input parameters is slightly less than $2\%$.

\subsubsection{Skew=1 Streams}
Heavy tailed streams are usually not analyzed in the literature, perhaps because the existing algorithm cannot find the top-$k$ elements in these using a reasonable amount of counters (see Section~\ref{sec:results}).
However, skew $1$ Zipf streams were analyzed in some previous works for both Space Saving and Count Sketch~\cite{CountSketch,SpaceSavings}.
In this section, we show that by introducing an admission probability, \PSS{}' is able to achieve asymptotic space improvement on i.i.d. Zipf streams.
The convergence condition for Space Saving, which appears in \eqref{eq:SSZipfReq} is now:
\begin{align}
m &> k + k\parentheses{{\Gamma_1(D)}-{\Gamma_1(k)}}\notag\\
&\approx k + k\ln\parentheses{1.78\frac{D}{k}} = \Theta (\log D).
\end{align}

For \PSS{}', we choose the admission probability to be 
%\footnote{This assumes $k^{\alpha/2} < \sqrt{\frac{1}{\ln D}}$. If this is not the case, we can get a $O(k\sqrt{\ln D})$ counters bound by choosing $P=\sqrt{\frac{1}{D}}$}
\begin{align}
P \triangleq 
%	k\cdot 
\sqrt{\frac{1}{\ln D}}.
\end{align}
We show that using probabilistic admission, we reduce the number of required counters to $m_{\PSS{}'}\triangleq c\cdot 
k\sqrt{
	%	^{1+\alpha}
	\ln D} = O(\sqrt{\log D})$, where $c$ is a positive constant.
We start by showing that this choice of $m_{\PSS{}'}$ and $P$ satisfies~\eqref{eq:\PSS{}1}:
\begin{align}
m_{\PSS{}'} = c\cdot k\sqrt{\ln D} > k\cdot P^{-1}
\end{align}
Next, we consider \eqref{eq:\PSS{}2}:
\begin{align*}
&k + k\cdotpa{\Gamma_1(m)-\Gamma_1(k) + P\cdot \Gamma_1(D)}\\
\approx& k + k\cdotpa{\ln\parentheses{1.78\frac{m}{k}} + \sqrt{\frac{1}{\ln D}}\ln\parentheses{1.78D}}\\
=& k + k\cdotpa{\ln\parentheses{c\cdot \sqrt{\ln D}} + \sqrt{\ln D}+O(1)} < m_{\PSS{}'}, 
\end{align*}
where the last inequality holds for large enough $c$.

We conclude that by introducing a probabilistic admission filter, one can reduce the number of counters required for successful top-$k$ identification over skew=1 Zipf streams from $O(\log D)$ to $O(\sqrt{\log D})$.
\subsection{Discussion and Future Work}
In this section, we have shown how the concept of admission filters, introduced in Section~\ref{sec:\PSS{}} can be adapted for reducing the number of counters needed for top-$k$ identification. Nevertheless, our analysis has a few drawbacks; first, we are only able to provide theoretical analysis for i.i.d. streams, which may not truly represent the all practical settings; second, our analysis assumes that the stream may be arbitrarily long, and only considers \emph{eventual convergence}. While this may fit massive data streams, in cases where we wish to process smaller streams, perhaps because we reset the process every once in a while to provide freshness, our analysis does not hold. 
Lastly, we have assumed a prior knowledge of the data skew. In practice, one can estimate the skew from the data, but this will require the admission probability to be adaptive.

We plan to evaluate \PSS{}' on real data traces and compare it with existing techniques, and specifically with \PSS{} (see Algorithm~\ref{alg:\PSS{}}) whose admission probability was optimized for frequency estimation but proved effective also for top-$k$. We also wish to find a method for dynamically adapting the admission probability without assuming knowledge about the stream length or skew.
}
%\clearpage
{
	
	\let\section\subsection
	\let\subsection\subsubsection
\begin{figure*}[]
	\begin{tabular}{ccc}
		\subfloat[On-Arrival MSE]{\label{fig:CaidaMSE}\includegraphics[width = \matrixCellWidth]
			{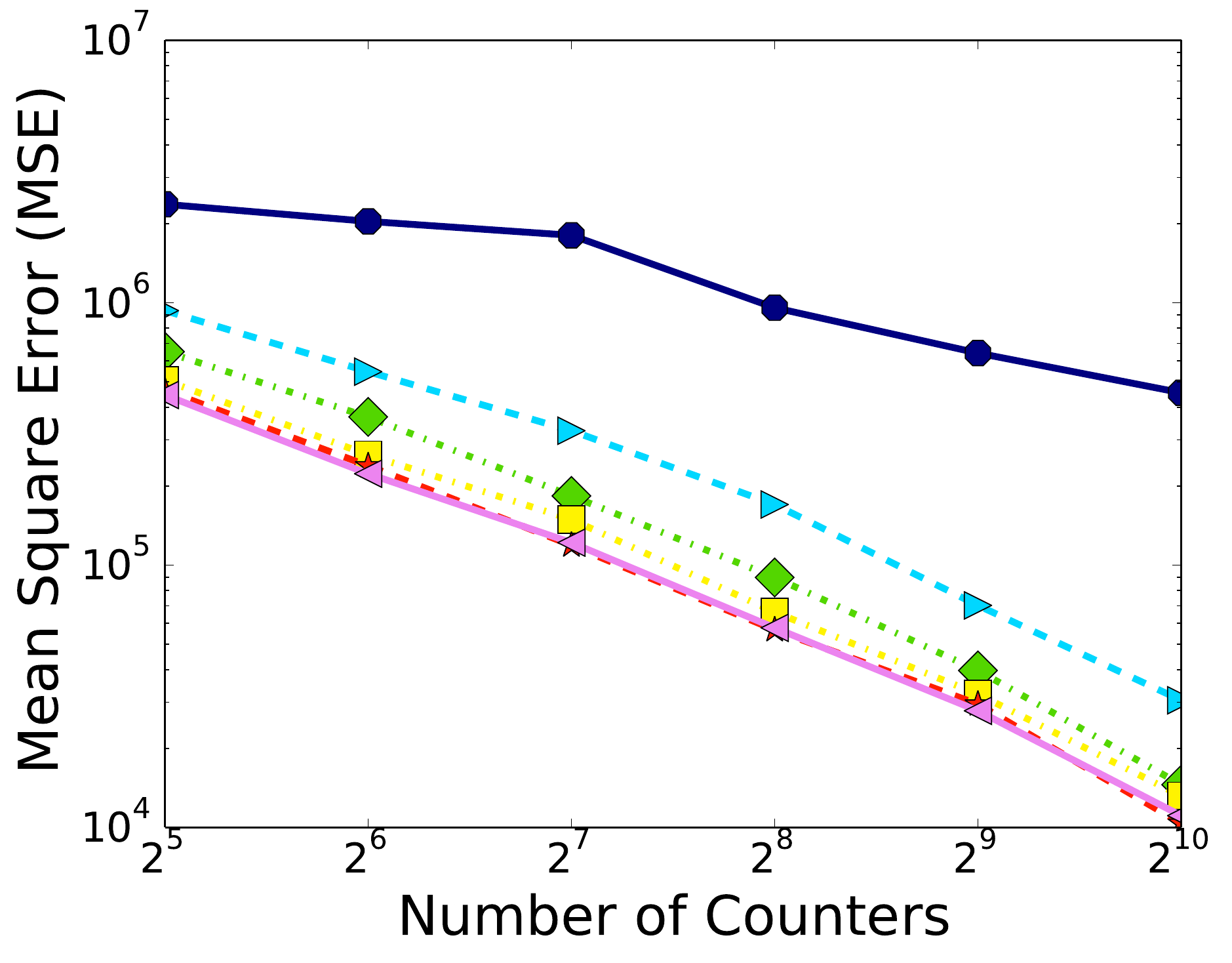}} &
		\subfloat[Top-$32$ recall]{\label{fig:Zipf0.6MSE}\includegraphics[width = \matrixCellWidth]
			{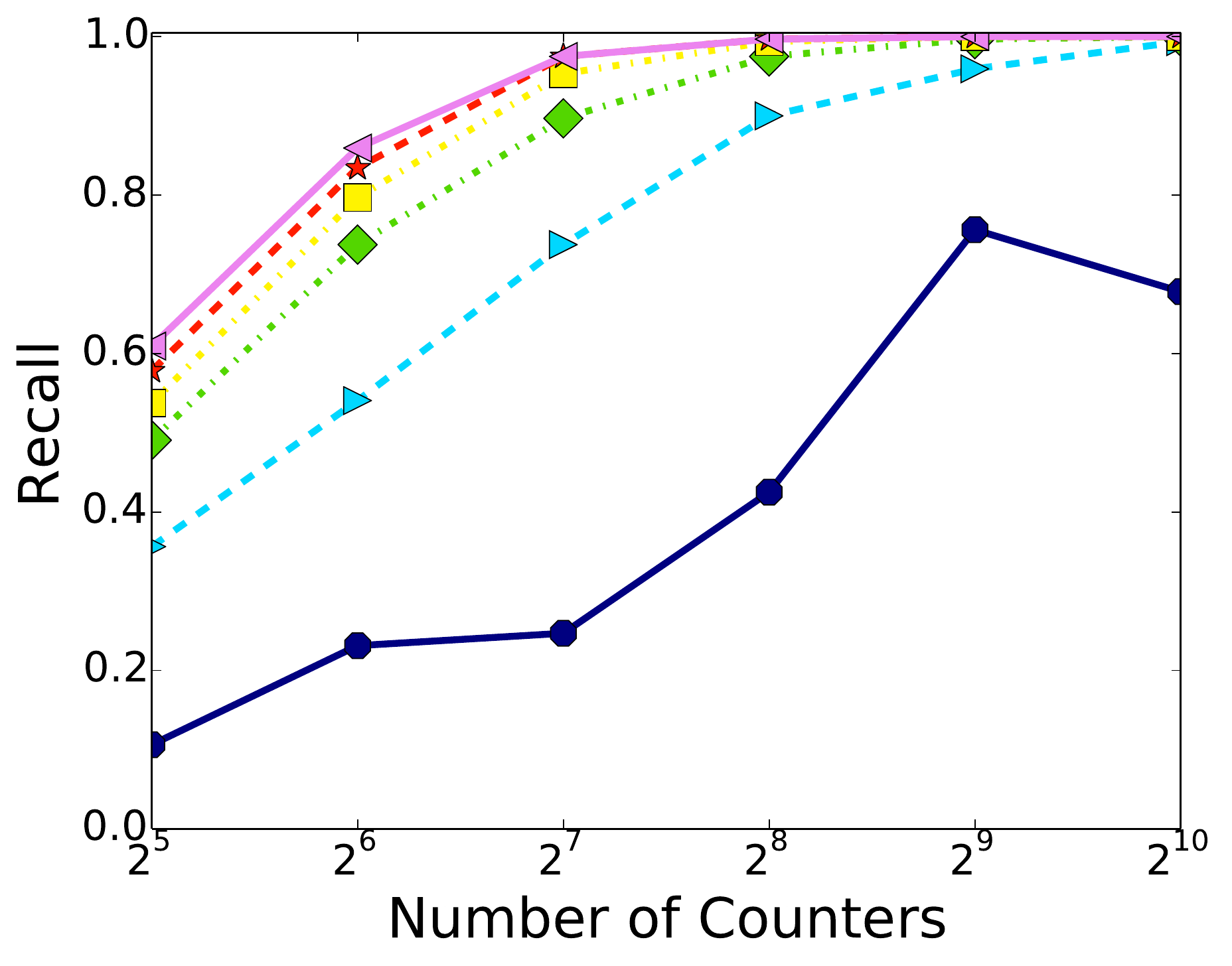}} &		
		\subfloat[Precision-Recall Curves]{\label{fig:YouTubeMSE}\includegraphics[width = \matrixCellWidth]
			{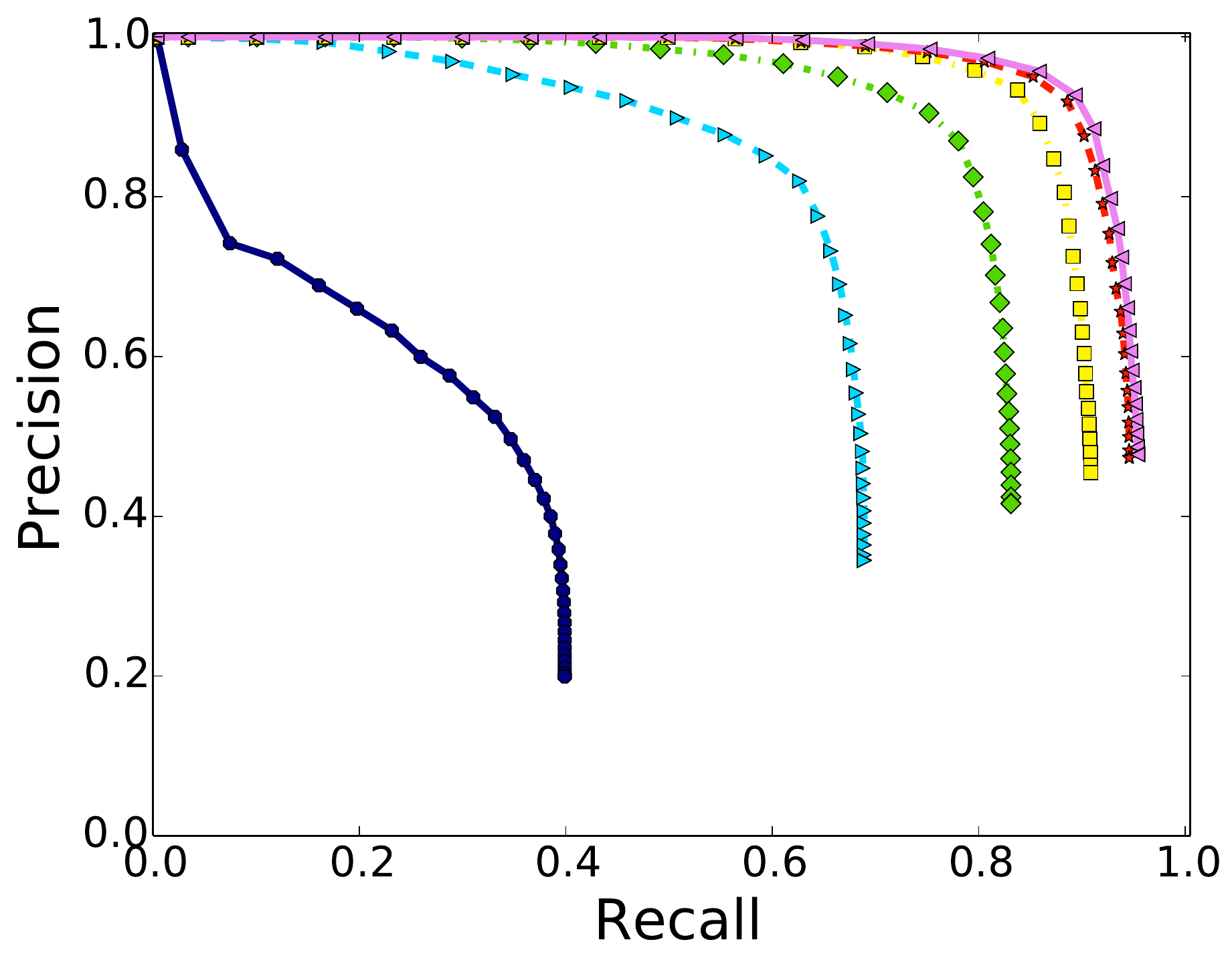}}\\
		\subfloat[Recall vs. Stream Size]{\label{fig:UCLAMSE}\includegraphics[width = \matrixCellWidth]
			{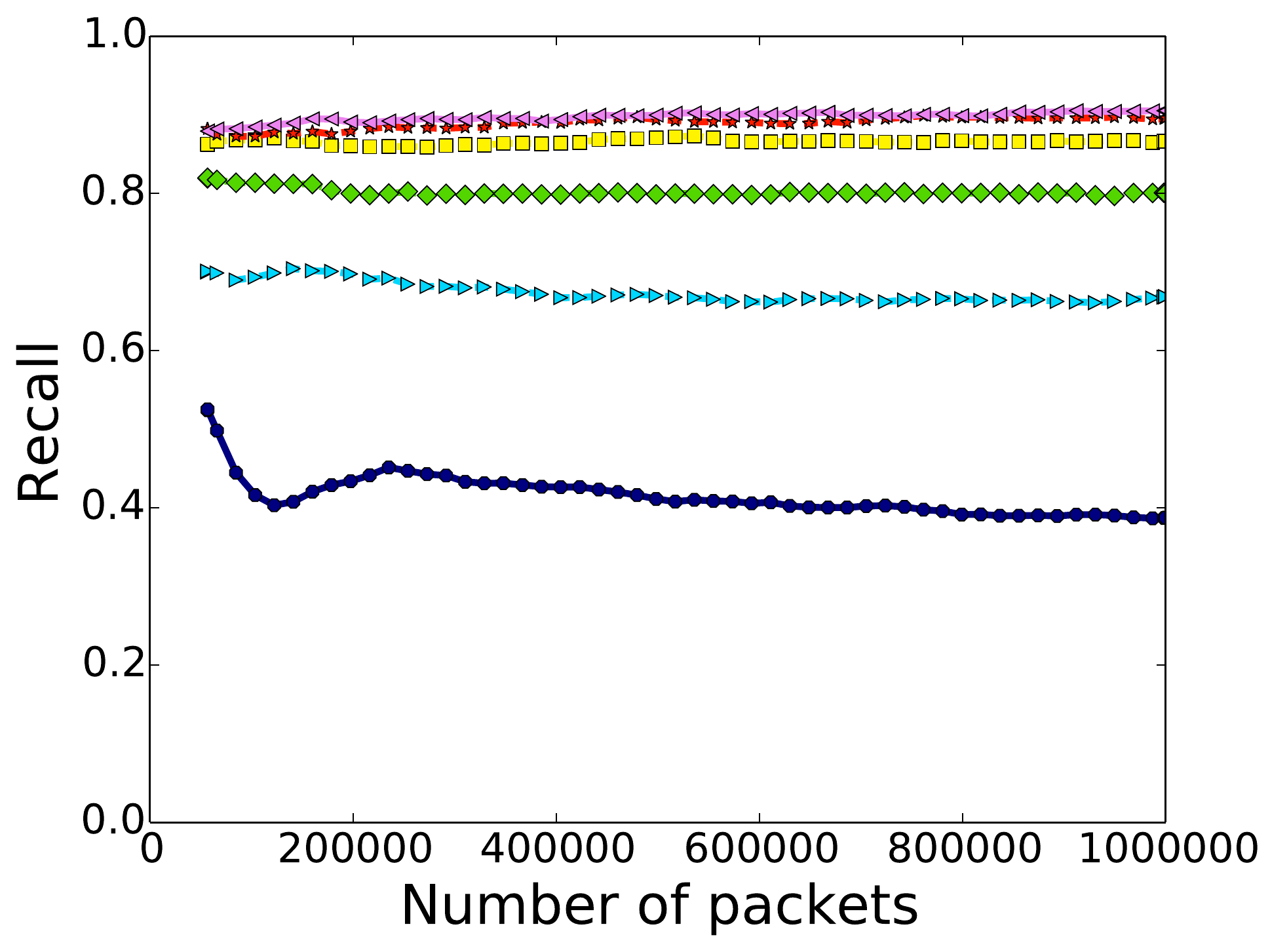}} &	
		\subfloat[Legend]{\includegraphics[width=4.7cm,height=4.2cm]
			{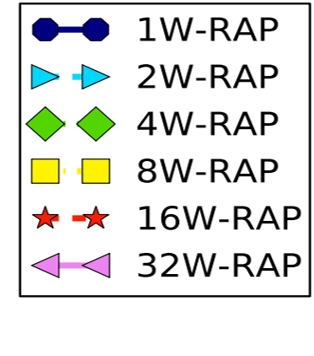}} 	&	
	\end{tabular}
	\caption{\label{fig:assoc}Comparison of the performance of $d$-Way \PSS{} for different associativity levels.}
\end{figure*}
\section{Limited Associativity Impact}\label{apx:assoc}
In this appendix, we compare the performance of $d$-Way \PSS{} for different values of $d$. 
We evaluate the associativity levels effect over several metrics which are presented in detail in Section~\ref{sec:eval-on-arrival} and Section~\ref{sec:eval-top-k}.
These include the following:
\begin{enumerate}
	\item On-Arrival Mean Square Error, in which every arriving element is queried and we compute the average square error.
	\item The percentage (recall) of elements within the top-$32$ successfully using various space allocations.
	\item The recall for identifying the top-$512$ elements using $1024$ counters, compared with the number of observed packets. 
	\item The precision-recall curve for identifying the top-$512$ elements using $1024$ counters. 
\end{enumerate}
Figure~\ref{fig:assoc} shows the performance of the different associativity levels, averaging over 10 batches of 1M packets each from the CAIDA~\cite{CAIDA} dataset. The results show a diminishing return pattern as associativity is increased; while $1$W-\PSS{} performs rather poorly, $2$W-\PSS{} is already comparable with the previous algorithms, $4$W and $8$W offer increased accuracy while $16$W-\PSS{} works almost as good as the $32W$. Further, our evaluation in Section~\ref{sec:eval-on-arrival} and Section~\ref{sec:eval-top-k} shows that $16W$-\PSS{} is roughly comparable to the fully associative \PSS{}. Our experiments suggest that associativity of $16$ counters per set is a highly attractive alternative to complete associativity as it does not require any sophisticated data structures (as suggested in~\cite{HeavyHitters,SpaceSavings,BatchDecrement}). By not using data structures, we get both simpler implementation, as well as reduction in the memory overhead they require, which can be used for allocating the algorithms with additional counters for increased accuracy.

}

\fi
\end{document}